\documentclass[11pt,english]{article}
\usepackage[latin9]{inputenc}
\usepackage{float,amsmath,nicefrac,amsthm,color,amssymb,graphicx,setspace,soul,xcolor,mwe,comment,paralist}
\usepackage{multirow,array}
\usepackage[lofdepth,lotdepth]{subfig}
\usepackage[title]{appendix}

\makeatletter

\theoremstyle{plain}
\newtheorem{ax}{\protect\axiomname}
\theoremstyle{plain}
\newtheorem{axb}{\protect\axiombname}
\theoremstyle{plain}
\newtheorem{thm}{\protect\theoremname}
\theoremstyle{plain}
\newtheorem{lem}{\protect\lemmaname}
\theoremstyle{remark}
\newtheorem{claim}{\protect\claimname}
\theoremstyle{plain}
\newtheorem{prop}{\protect\propositionname}
\theoremstyle{remark}

\theoremstyle{remark}

\theoremstyle{plain}

\theoremstyle{remark}

\theoremstyle{remark}

\usepackage[T1]{fontenc}
\usepackage{lmodern, microtype}
\usepackage[small]{titlesec}

 \usepackage[left=1.3in, right=1.3in, top=1.3in,bottom=1.3in]{geometry}

\usepackage[colorlinks=true, pdfstartview=FitV, linkcolor=blue,citecolor=blue, urlcolor=blue]{hyperref}

\makeatother
\usepackage{natbib}
\setcitestyle{authoryear,open={(},close={)}}
\providecommand{\definitionname}{Definition}
\providecommand{\axiomname}{Axiom}
\providecommand{\axiombname}{Axiom}
\providecommand{\claimname}{Claim}
\providecommand{\lemmaname}{Lemma}
\providecommand{\theoremname}{Theorem}
\providecommand{\observationname}{Observation}
\providecommand{\examplename}{Example}
\providecommand{\propositionname}{Proposition}
\providecommand{\remarkname}{Remark}

\global\long\def\RR{\mathbb{R}}
\newcommand{\argmax}{\operatornamewithlimits{argmax}}

   \def\dd{\mathrm{d}}
        \def\PP{\mathbb{P}}
\def\ee{\mathrm{e}}         
        \def\P{\mathcal{P}}
\def\dkl{D_{\mathrm{KL}}}

\title{\textbf{\Large{The Cost of Information: The Case of Constant Marginal Costs}}\thanks{We thank Kim Border, Ben Brooks, Simone Cerreia-Vioglio, Tommaso Denti, Federico Echenique, Drew Fudenberg, Ed Green, Adam Kapor, Massimo Marinacci, Jeffrey Mensch, Filip Mat{\v{e}}jka, Stephen Morris, Teemu Pekkarinen, Doron Ravid, and Yangfan Zhou for their comments. All errors and omissions are our own.}
}

\author{ \large Luciano Pomatto\thanks{Caltech. Email: 	luciano@caltech.edu.} \ \ \ %
Philipp Strack\thanks{Yale. Email:  philipp.strack@gmail.com. Philipp Strack was supported by a Sloan Fellowship.} \ \ \ %
Omer Tamuz\thanks{Caltech. Email: tamuz@caltech.edu. Omer Tamuz was supported by a grant from the Simons Foundation (\#419427), a Sloan research fellowship, and a BSF award (\#2018397).}}

\begin{document}

\maketitle
\begin{abstract}

    We develop an axiomatic theory of information acquisition that captures the idea of constant marginal costs in information production: the cost of generating two independent signals is the sum of their costs, and generating a signal with probability half costs half its original cost. Together with Blackwell monotonicity and a continuity condition, these axioms determine the cost of a signal up to a vector of parameters. These parameters have a clear economic interpretation and determine the difficulty of distinguishing states. 
    
\end{abstract}

\section{Introduction}
%
%

%
%

 Much of contemporary economic theory is built on the idea that information is scarce and valuable. A proper understanding of information as an economic commodity requires theories for its value, as well as for its production cost. While the literature on the value of information \citep*{bohnenblust1949reconnaissance, blackwell1951comparison} is by now well established, modeling the cost of producing information has remained an unsolved problem.\footnote{For example, \citet*{arrow1985informational} makes the following statement: ``The choice of information structures must be subject to some limits, otherwise, of course, each agent would simply observe the entire state of the world. There are costs of information, and it is an important and incompletely explored part of decision theory in general to formulate reasonable cost functions for information structures.''} In this paper, we develop an axiomatic theory of costly information acquisition.


 We characterize all cost functions over Blackwell experiments  that satisfy three main axioms: First, experiments that are more informative in the sense of \cite*{blackwell1951comparison} are more costly. Second, the cost of generating independent experiments equals the sum of their individual costs. Third, the cost of generating an experiment with probability half equals half the cost of generating it with probability one.


%
%
 Our three axioms admit a straightforward economic interpretation. The first one is a form of monotonicity: more precise information is more costly. The second and third axioms capture the idea of linear cost. The second axiom implies that the cost of collecting $n$ independent random samples is linear in $n$. For example, if the variable is the perceived quality of a new product, and information is generated by surveying random customers, the axiom is satisfied if the cost of calling an additional customer is constant: i.e.\ calling 20 customers is twice as costly as calling 10. More generally, the axiom requires the cost to be additive with respect to experiments that are independent conditional on the state. Similarly, the third axiom implies that the cost of producing a sample with probability $\alpha$ is a fraction $\alpha$ of the cost of acquiring the same sample with probability one. This axiom is satisfied by all posterior separable costs, which include nearly all models of information cost in the literature.

 
 We propose these linearity assumptions as a way of studying cost functions over information structures. In the context of traditional commodities, a standard avenue for studying cost functions is by categorizing them in terms of decreasing, increasing, or constant marginal costs, with the latter being arguably the conceptually simplest case. In this paper we take a similar approach for studying the cost of information acquisition, and our axioms make an attempt at formalizing the assumption of constant marginal costs for information. As in the case of traditional commodities, assuming linear costs is restrictive, and it is easy to conceive of decision problems where our axioms are violated. For example, if customers are hard to find, surveying 20 customers might cost more than twice as much as surveying 10. Conversely, economies of scale may result in decreasing marginal costs. Nevertheless, our axioms have the advantage of admitting a clear economic interpretation, making it possible to judge for which applications they are appropriate. We thus propose the study of linear cost functions as a first step towards the wider goal of studying general information costs in terms of their economic properties.

\paragraph*{Representation.} The main result of this paper is a characterization theorem for cost functions over experiments.
We are given a finite set $\Theta$ of states of nature. An experiment $\mu$ produces a signal realization $s \in S$ with probability $\mu_{i}(s)$ in state $i \in \Theta$.
We show that for any cost function $C$ that satisfies the above postulates, together with a continuity assumption, there exist unique non-negative coefficients $(\beta_{ij})$, one for each ordered pair of states of nature $i$ and $j$, such that%
\footnote{Throughout the paper we assume that the set of states of nature $\Theta$ is finite. We do not assume a finite set $S$ of signal realizations and the generalization of \eqref{eq:cost} to infinitely many signal realizations is given in \eqref{eq:thm_1}.}
\begin{equation}\label{eq:cost}
    C(\mu) = \sum_{i,j\in\Theta}\beta_{ij}\left(\sum_{s\in S}\mu_{i}(s)\,\log\frac{\mu_{i}(s)}{\mu_{j}(s)}\right)\,.
\end{equation}
%
%
Each coefficient $\beta_{ij}$ can be interpreted as capturing the difficulty of discriminating between state $i$ and state $j$, as the cost can be expressed as a linear combination
$$
  C(\mu) = \sum_{i,j\in\Theta}\beta_{ij}\dkl(\mu_i \Vert \mu_j),
$$
where the \emph{Kullback-Leibler divergence}
$$
\dkl(\mu_i\Vert \mu_j) = \sum_{s\in S}\mu_{i}(s)\log\frac{\mu_{i}(s)}{\mu_{j}(s)}
$$ 
is the expected log-likelihood ratio between state $i$ and state $j$ when the state equals $i$. The term $\dkl(\mu_i \Vert \mu_j)$ is thus large if the experiment $\mu$ on average produces evidence that strongly favors state $i$ over $j$, conditional on the state being $i$. Hence, the larger the coefficient $\beta_{ij}$, the more costly it is to reject the hypothesis that the state is $j$ when it truly is $i$. Formally, $\beta_{ij}$ is the marginal cost of increasing the expected log-likelihood ratio of an experiment with respect to states $i$ and $j$, conditional on $i$ being the true state. We refer to the cost \eqref{eq:cost} function as the \textit{log-likelihood ratio cost} (or \textit{LLR cost}).

In many common information acquisition problems, states of the world are one-dimensional quantities. For instance, this is the case when the unknown state is a physical quantity such as height or weight, or an economic quantity such as the inflation rate. In these examples, an experiment can be seen as a noisy measurement of the unknown underlying state $i \in \RR$. We provide a framework for choosing the coefficients $\beta_{ij}$ in these contexts. Our main hypotheses are that the difficulty of distinguishing between two states $i$ and $j$ is a function of the distance between them, and that the cost of performing a measurement with standard Gaussian noise does not depend on the set of states $\Theta$ in the particular information acquisition problem; this is a feature that is commonly assumed in models that exogenously restrict attention to normal experiments.

Under these assumptions, we   show that there exists a constant $\kappa \geq 0$ such that, for every pair of states $i,j \in \Theta$,
\[
\beta_{ij} = \frac{\kappa}{(i-j)^2}.
\]
In this functional form, the difficulty of distinguishing between states is a quadratic decreasing function of  the distance between them. As we show, this choice of parameters offers a simple and tractable framework for analyzing the implications of the LLR cost. 

\medskip

The concept of a Blackwell experiment makes no direct reference to subjective probabilities nor to  Bayesian reasoning.\footnote{Blackwell experiments have been studied both within and outside the Bayesian framework. See, for instance, \cite{le1996comparison} for a review of the literature on Blackwell experiments.
}
Likewise, our axioms and  characterization theorem do not presuppose the existence of a prior over the states of nature. Nevertheless, given a prior $q$ over $\Theta$, an experiment induces a distribution over posteriors $p$, making $p$ a random variable. 
Under this formulation, the LLR cost \eqref{eq:cost} of an experiment can be represented as the expected change of the function 
\[
F(p)=\sum_{i,j \in \Theta}\beta_{ij}\frac{p_{i}}{q_{i}}\log\left(\frac{p_{i}}{p_{j}}\right)
\]
from the prior $q$ to the posterior $p$ induced by the signal. That is, the cost of an experiment equals
\[
\mathbb{E}\left[F(p)-F(q)\right]
\]
where the expectation is taken with the distribution of posterior beliefs induced by the experiment and the prior. This establishes that LLR cost is posterior-separable, and makes it possible to apply techniques and insights derived for posterior-separable costs functions \citep*{caplin2013behavioral, caplindeanleahy2017}.

\paragraph*{Relation to Mutual Information Cost.}

 Following the seminal work of \cite*{sims2003implications,sims2010rational} on rational inattention, cost functions based on mutual information have been commonly used in applications;  \cite*{mackowiak2018rational} review the literature on rational inattention. Mutual information costs are defined as the expected change $$\mathbb{E}\left[H(q)-H(p)\right]$$ of the Shannon entropy
 $H(p) = -\sum_{i\in\Theta}p_i\log p_i$ between the decision maker's prior belief $q$ and posterior $p$. Equivalently, in this formulation, the cost of an experiment is given by the mutual information between the state of nature and the signal.\footnote{\label{linear-mutual} Related specifications discussed in the literature include models where the decision maker can acquire, for free, any experiment whose mutual information is below an upper bound \citep{sims2003implications}, as well as costs that are increasing transformation of mutual information \citep{denti2022posterior}.} \label{minor:10}One of the main differences between mutual information and the LLR cost, is that the first is subadditive rather than additive \citep[see, e.g.][]{lindley1956measure}, so that the cost of $n$ independent copies of an experiment is a strictly concave function of $n$. In applications, the LLR cost function leads to predictions which are qualitatively different from those induced by mutual information cost. We illustrate the differences in \S\ref{sec:examples} and \S\ref{sec:stochastic}.


 \paragraph*{Examples and Applications.} In \S\ref{sec:stochastic} we apply the LLR cost function to information acquisition problems and derive a number of predictions. Our applications include binary prediction problems, where a decision maker needs to predict whether the state is above or below a given threshold. An example of this is an analyst trying to predict which party will obtain the majority of votes in an election. Another example is a perception task where a subject is asked to observe a number of dots of two different colors on a screen, and must predict which color is predominant.\footnote{The two examples have a similar structure but are, of course, quite different in terms of data collection since perception tasks are usually performed with experimental subjects in controlled environments.}
 
 We show that in binary prediction problems the decision maker is strictly more likely to make the correct choice when the quantity to be predicted is farther from the desired threshold, under general assumptions on the coefficients $(\beta_{ij})$. For example, it is harder for the agent to  predict the winner in a close election than in an election where one of the candidates has a large lead.  Moreover, we show that under the specification $\beta_{ij} = \frac{\kappa}{(i-j)^2}$, the decision maker's probability of a choosing an action is a sigmoidal function of the state---a prediction in line with psychometric evidence on perception tasks.
 
 This and other examples illustrate how the LLR cost function leads to optimal choice probabilities that  take into account the difficulty of distinguishing between states. While intuitive, this property is ruled out by cost functions such as mutual information that treat states symmetrically.  
 
 \paragraph{Scope and Limitations.} There are many applications where the our  additivity assumption is violated, and so the LLR cost function is inadequate. A stark case, which we discuss in the next section, is that of  experiments that completely rules out a state; these would have infinite LLR cost. Thus our framework is incompatible with partitional information structures, which are an important modelling tool. Moreover, the fact that our representation has a number of parameters that grows with the number of states makes calculations and identification more difficult.

 A natural question is how the LLR cost can be applied in dynamic settings in which agents decide sequentially what information to acquire. As discussed in depth by \cite{bloedel2020cost}, it is impossible---under reasonable assumptions---to have a cost function that satisfies the assumption of constant marginal costs and is independent of the prior of the decision maker. This is a subtle issue which we explore in more detail in \S\ref{sec:bayesian}.

\section{Model}\label{sec:model}

A decision maker acquires information on an unknown state of nature belonging to a finite set $\Theta$. Elements of $\Theta$ will be denoted by $i,j,k$, etc.
Following \citet*{blackwell1951comparison},
we model the information acquisition process by means of \textit{experiments}.
An experiment $\mu=(S,(\mu_i)_{i \in \Theta})$ consists of a set $S$ of signal realizations equipped with a sigma-algebra $\Sigma$, and for each state $i \in \Theta$ a probability measure $\mu_i$ defined on $(S,\Sigma)$. The set $S$ represents the possible outcomes of the experiment, and each measure $\mu_i$ describes the distribution of outcomes when the true state is $i$.

We assume throughout that the measures $(\mu_i)$ are mutually absolutely continuous, so that each derivative (i.e.\ ratio between densities) $\frac{\dd \mu_i}{\dd \mu_j}$ is finite almost everywhere. In the case of finite signal realizations these derivatives are simply equal to ratio between probabilities $\frac{\mu_i(s)}{\mu_j(s)}$.\footnote{This assumption means that no signal can ever rule out any state, and in particular can never completely reveal the true state.}

Given an experiment $\mu$, we denote by%
\[
\ell_{ij}(s)=\log\frac{\dd \mu_{i}}{\dd \mu_{j}}(s)
\]
the log-likelihood ratio between states $i$ and $j$ upon observing the realization $s$. We define the vector
\[
    \left(\ell_{ij}(s)\right)_{i,j\in \Theta}
\]
of log-likelihood ratios among all pairs of states. The distribution of $\ell$ depends on the true state generating the data.
Given an experiment $\mu$, we denote by $\bar\mu_i$ 
the distribution of $\ell$ conditional on state $i$.%
\footnote{The measure $\bar\mu_i$ is defined as $\bar\mu_i(A)= \mu_i (\{s: \left(\ell_{ij}(s)\right) \in A\})$ for every measurable $A\subseteq\mathbb{R}^{\Theta\times\Theta}$. }

We restrict our attention to experiments where the induced log-likelihood ratios $\left(\ell_{ij}\right)$ have finite moments. That is, experiments such that for every state $i$ and every  vector of integers $\alpha\in\mathbb{N}^{\Theta}$ the expectation $\int_S | \prod_{k\neq i}\ell_{ik}^{\alpha_{k}}|\dd\mu_i$ is finite. We denote by $\mathcal{E}$ the class of all such experiments.\footnote{We refer to $\mathcal{E}$ as a class, rather than a set, since Blackwell experiments do not form a  well-defined set. In doing so we follow a standard convention in set theory \citep*[see, for instance,][p. 5]{jech2013set}.} The restriction to  $\mathcal{E}$ is a technical condition that rules out experiments whose  log-likelihood ratios have very heavy tails, but, to the best of our knowledge, includes all (not fully revealing) experiments commonly used in applications.

The cost of producing information is described by an \textit{information cost function}
\[
C\colon \mathcal{E} \to\mathbb{R}_{+}
\]
assigning to each experiment $\mu \in \mathcal{E}$ its cost $C(\mu)$. In the next section we introduce and characterize four basic properties for information cost functions.

\subsection{Axioms}

Our first axiom postulates that the cost of an experiment should depend only on its informational content.
For instance, it should not be sensitive to the way signal realizations are labelled. In making this idea formal we follow \citet*[Section 4]{blackwell1951comparison}.

Let $q \in \P(\Theta)$ be the uniform prior assigning equal probability to each element of $\Theta$.%
\footnote{Throughout the paper, $\P(\Theta)$ denotes the set of probability measures on $\Theta$ identified with their representation in $\RR^{\Theta}$, so that for every $q \in \P(\Theta)$, $q_i$ is the probability of the state $i$.}
Let $\mu$ and $\nu$ be two experiments, inducing the distributions over posteriors $\pi_{\mu}$ and $\pi_{\nu}$ given the uniform prior $q$. Then $\mu$ dominates $\nu$ in the Blackwell order if $$\int_{\P(\Theta)} f(p) \,\dd \pi_{\mu}(p) \geq \int_{\P(\Theta)} f(p) \,\dd \pi_{\nu}(p) $$ for every convex function $f\colon \P(\Theta) \to \RR$. As is well known, dominance with respect to the Blackwell order is equivalent to the requirement that in any decision problem, a Bayesian decision maker achieves a (weakly) higher expected utility when basing her action on $\mu$ rather than $\nu$. We say that two  experiments are \textit{Blackwell equivalent} if they dominate each other.

It is natural to require the cost of information to be increasing in the Blackwell order. For our main result, it is sufficient to require that any two experiments that are Blackwell equivalent lead to the same cost. Nevertheless, it will turn out that our axioms imply the stronger property of Blackwell monotonicity, as shown by Proposition \ref{prop:llr-monotnoe} below.

\begin{ax}\label{axm:info-content}
If $\mu$ and $\nu$ are Blackwell equivalent, then $C(\nu)=C(\mu).$
\end{ax}

The lower envelope of a cost function assigns to each $\mu$ the minimum cost of producing an experiment that is Blackwell equivalent to $\mu$. If experiments are optimally chosen by a decision maker then we can, without loss of generality, identify a cost function with its lower envelope. This results in a cost function for which Axiom \ref{axm:info-content} is automatically satisfied.

For the next axiom, we study the cost of performing multiple independent experiments. Given two experiments $\mu = (S,(\mu_i))$ and $\nu=(T,(\nu_i))$ we define their product
\[
    \mu \otimes \nu = (S \times T,(\mu_i \times \nu_i))
\]
where $\mu_i \times \nu_i$ denotes the product of the two measures.\footnote{When the set of signal realizations is finite, the measure $\mu_i \times \nu_i$ assigns to each realization $(s,t)$ the probability $\mu_i(s)\nu_i(t)$.} 
Under the experiment $\mu \otimes \nu$, the realizations of both experiments $\mu$ and $\nu$ are observed, and the two observations are independent conditional on the state.
To illustrate, suppose $\mu$ and $\nu$ consist of drawing a random sample from two possible populations. Then $\mu \otimes \nu$ is the experiment where two independent samples, one for each population, are collected.

Our second axiom states that the cost function is additive with respect to combining independent experiments:
\begin{ax}\label{axm:additivity}
The cost of performing two independent experiments is the sum of their costs:
\[
C(\mu \otimes \nu) = C(\mu)+C(\nu)\text{ for all }\mu \text{ and } \nu.
\]
\end{ax}
An immediate implication of Axioms  \ref{axm:info-content} and \ref{axm:additivity} is that a completely uninformative experiment has zero cost. This follows from the fact that  an uninformative experiment $\mu$ is Blackwell equivalent to the product experiment $\mu \otimes \mu$.

In many settings, an experiment can sometimes fail to produce new evidence. The next axiom states that the cost of an experiment is linear in the probability that it will generate information. Given $\mu$, we define a new experiment, which we call a \emph{dilution} of $\mu$ and denote by $\alpha \cdot \mu$. In this new experiment, with probability $\alpha$ the experiment $\mu$ is produced, and with probability $1-\alpha$ a completely uninformative signal is observed. Formally, given $\mu = (S,(\mu_i))$, fix a new signal realization $o \notin S$ and a probability $\alpha\in\left[0,1\right]$. We define 
$$\alpha \cdot \mu =(S\cup\{ o\} ,(\nu_i)),$$ where $\nu_{i}(E) = \alpha \mu_i(E)$ for every measurable $E \subseteq S$, and $\nu_i(\{o\})= 1 - \alpha$. The next axiom specifies the cost of such an experiment:

\begin{ax}
\label{axm:affinity-1}The cost of a dilution $\alpha \cdot \mu$ is linear in the probability $\alpha$:
\[
C(\alpha \cdot \mu)=\alpha\,C(\mu)\text{ for every }\mu \text{ and }\alpha\in\left[0,1\right].
\]
\end{ax}

Our final assumption is a continuity condition. We first introduce a (pseudo)-metric over $\mathcal{E}$. Recall that for every experiment $\mu$, $\bar{\mu}_{i}$ denotes its distribution of
log-likelihood ratios conditional on state $i$. We denote by $d_{tv}$ the total-variation distance.\footnote{That is, $d_{tv}(\bar\mu_i,\bar\nu_i) = \sup|\bar\mu_i(A) - \bar\nu_i(A)|$, where the supremum is over all measurable subsets of $\RR^{\Theta \times \Theta}$.} Given a vector $\alpha\in\mathbb{N}^{\Theta}$, let $M_{i}^{\mu}(\alpha)= \int_S | \prod_{k\neq i}\ell_{ik}^{\alpha_{k}}|d\mu_i$ be the $\alpha$-moment of the vector of log-likelihood ratios $(\ell_{ik})_{k\neq i}$. Given an upper bound $N\geq1$, we define the distance:
\begin{equation*}
d_N(\mu,\nu)=\max_{i\in\Theta}d_{tv}\left(\bar\mu_i,\bar\nu_i\right)+\max_{i\in\Theta}\max_{\alpha\in\left\{ 0,...,N\right\} ^{n}}\left|M_{i}^{\mu}(\alpha) - M_{i}^{\nu}(\alpha)\right|.
\end{equation*}
According to the metric $d_N$, two experiments $\mu$ and $\nu$ are close if, for each state $i$, the induced distributions of log-likelihood ratios are close in total-variation and, in addition, have similar moments, for any moment $\alpha$ lower or equal to $\left( N,\ldots,N\right)$.

\begin{ax}\label{axm:continuity}
For some $N \geq 1$ the function $C$ is uniformly continuous with respect to $d_N$.
\end{ax}

As is well known, convergence with respect to the total-variation distance is a demanding requirement, as compared to other topologies such as the weak topology. So, continuity with respect to $d_{tv}$ is a relatively mild assumption. Continuity with respect to the stronger metric $d_N$ is, therefore, an even milder assumption.\footnote{We discuss this topology in detail in \S\ref{sec:topology}. Any information cost function that is continuous with respect to the metric $d_N$ satisfies Axiom \ref{axm:info-content}. For expositional clarity, we maintain the two axioms as separate throughout the paper.} As we show in Theorem~\ref{thm:mu} in the Appendix, our characterization holds for the case of two states and bounded experiments even if one only imposes Blackwell monotonicity, Axiom~\ref{axm:additivity} and Axiom~\ref{axm:affinity-1}, without requiring continuity.

\subsection{Discussion}

 We now discuss the interpretation of our axioms as well as some limitations imposed by our modeling assumptions. Axiom \ref{axm:additivity} has a simple interpretation. Consider the classical problem of learning the bias of a coin by flipping it multiple times. This experiment could correspond to the act of surveying customers, who either like a product or not, in order to learn whether the product is popular. It could also represent a political party surveying voters to discover the appeal of a potential candidate.  

 Suppose the coin either yields heads 80\% of the time or tails 80\% of the time and that either bias is equally likely. We compare the cost of observing a single coin flip versus a long sequence of coin flips.  Under the additivity axiom, the cost of observing $k$ coin flips is linear in $k$. 

  Additivity assumptions in the spirit of Axiom \ref{axm:additivity} have appeared in multiple parametric models of information acquisition. A standard assumption in Wald's classic model of sequential sampling and its variations is that the cost of acquiring $n$ independent samples is linear in $n$ \citep*[see, e.g.,][]{wald1945sequential, arrow1949bayes}. A similar condition appears in the continuous-time formulation of the sequential sampling problem, where the information structure consists of observing a signal with Brownian noise over a time period of length $t$, under a cost that is linear in $t$ \citep*{dvoretzky1953sequential,chan2017deliberating, morrisstrack2018}. Likewise, in static models where information is acquired by means of normally distributed experiments, a standard specification is that the cost of an experiment is inversely proportional to its variance \citep*[see, e.g.][]{wilson1975informational}. This amounts to an additivity assumption, since the product of two independent normal experiments is Blackwell equivalent to a normal experiment whose precision is the sum of the original precisions. Underlying these different models is the notion that the cost of an additional independent experiment is constant. Axiom \ref{axm:additivity} captures this idea in a non-parametric context, with no a priori restrictions  over the domain of feasible experiments. 


 Axiom \ref{axm:affinity-1} expresses the idea that the marginal cost of increasing the probability of success of an experiment is constant. The axiom is implied by posterior separability---the standard assumption in the literature for cost functions over experiments.\footnote{\label{r1:5}A posterior separable cost function is affine with respect to the distribution of beliefs induced by an experiment. The distribution of beliefs induced by the diluted experiment $\alpha \cdot \mu$ is a convex combination that puts weight $\alpha$ on the distribution generated by $\mu$ and weight $1-\alpha$ on the prior. Thus, under posterior separability the cost of $\alpha \cdot \mu$ is affine in $\alpha$.} It is however, a strictly weaker assumption. We also note that for proving our results it suffices to restrict this axiom to $\alpha=1/2$.\footnote{This axiom admits an additional interpretation. Suppose the decision maker is allowed to randomize her choice of experiment. Then, the property
 \begin{equation}\label{eq:noarb1}
    C(\alpha \cdot \mu) \leq \alpha \, C(\mu)
 \end{equation}
 ensures that the cost of the diluted experiment $\alpha \cdot \mu$ is not greater than the expected cost of performing $\mu$ with probability $\alpha$ and collecting no information with probability $1-\alpha$. Hence, if  \eqref{eq:noarb1} was violated, the experiment $\alpha \cdot \mu$ could be replicated at a strictly lower cost through a simple randomization by the decision maker. Now assume Axiom $\ref{axm:additivity}$ holds, and the decision maker is allowed to perform independent copies of the diluted experiment $\alpha \cdot \mu$ until it succeeds. Then, the converse inequality
 \[
    C(\alpha \cdot \mu) \geq \alpha \, C(\mu)
 \]
 ensures that the cost $C(\mu)$ of an experiment is not greater than the expected cost $(1/\alpha)C(\alpha \cdot \mu)$ of performing the experiment $\alpha \cdot \mu$ until it succeeds.}

 \label{r4:5} The domain of our cost function rules out experiments that with positive probability allow the decision maker to be certain that a state did not happen. Such experiments, if included in the domain, would have infinite cost under our axioms.\footnote{For example, if a cost function $C$ is Blackwell monotone, additive, and assigns strictly positive cost to at least one experiment $\mu$ that is not perfectly revealing, then it must assign infinite cost to a perfectly revealing experiment. Indeed, by Blackwell monotonicity, the cost of the $n$-times repeated experiment $\mu^{\otimes n}$ must always be below the cost of a perfectly informative experiment. By additivity, $C(\mu^{\otimes n}) = n C(\mu)$, and thus a perfectly informative experiment must have infinite cost.} 
 While this is not special to our framework---the same issue applies to Wald's model and others---it is nevertheless an important limitation, since information structures that rule out states with certainty are a common  modeling tool. An example are  partitional information structures, which are standard in information economics. A disadvantage of the LLR cost function is that it cannot be applied in such settings.

 To gain some intuition for the sort of experiments that are ruled out, consider an urn containing 100 balls. Suppose there are only two states: either all balls are red, or all balls are blue. In this case, sampling from the urn perfectly reveals the state, and thus such an experiment cannot be accommodated by the LLR cost. Indeed, it conflicts with the constant marginal cost assumption: If the experiment had finite cost, then repeating it twice would have twice the cost. But repeating the experiment does not provide any additional information, since one sample is enough to reveal the state. Thus, the constant marginal cost assumption fails in this example.
 
 Suppose instead the urn contains either 1 blue ball and 99 red balls, or 1 red ball and 99 blue ones. In this case, drawing from the urn is an experiment that does not exclude states with certainty, and fits with the assumption of additivity. As the number of samples grows, the decision maker obtains more and more accurate statistical evidence of the true state, but without ever reaching full certainty. 

\section{Representation}
\begin{thm}\label{thm:repr-1}
An information cost function $C$ satisfies Axioms \ref{axm:info-content}-\ref{axm:continuity} if and only if there exists a collection $\left(\beta_{ij}\right)_{i, j\in\Theta, i \neq j}$
in $\mathbb{R}_{+}$ such that for every experiment $\mu=(S, (\mu_i))$,
\begin{equation}\label{eq:thm_1}
C(\mu)=\sum_{i,j}\beta_{ij}\int_S \log\frac{\dd \mu_{i}}{\dd \mu_{j}}(s) \, \dd \mu_i (s).
\end{equation}
Moreover, the collection $(\beta_{ij})$ is unique given $C$.
\end{thm}
We refer to a cost function that satisfies Axioms 1-4 as a \emph{log-likelihood ratio (LLR) cost}. As shown by the theorem, this class of information cost functions is uniquely determined up to the parameters $\left(\beta_{ij}\right)$. The expression $\int_S \log(\dd \mu_{i}/ \dd \mu_{j}) \dd \mu_i$ is the Kullback-Leibler divergence $\dkl(\mu_{i} \Vert \mu_{j})$ between the two distributions,  a well understood and tractable measure of informational content \citep{kullback1951information}.
The representation \eqref{eq:thm_1} can be rewritten as
\[
C(\mu)=\sum_{i,j}\beta_{ij}\dkl{(\mu_{i} \Vert \mu_{j})}.
\]

A higher value of $\dkl(\mu_{i} \Vert \mu_{j})$ describes an experiment which, conditional on state $i$, produces stronger evidence in favor of state $i$ compared to $j$, as represented by a higher expected value of the log-likelihood ratio $\log\dd \mu_{i}/\dd \mu_{j}$. 
The coefficient $\beta_{ij}$ thus measures the marginal cost of increasing the expected log-likelihood ratio between states $i$ and $j$, conditional on $i$, while keeping all other expected log-likelihood ratios fixed.\footnote{As we formally show in Lemma \ref{lem:domain} in the Appendix, this operation of increasing a single expected log-likelihood ratio while keeping all other expectations fixed is well-defined: for every
experiment $\mu$ and every $\varepsilon>0$, if $\dkl(\mu_i\Vert\mu_j)>0$ then there exists a new experiment $\nu$ such that $\dkl(\nu_i\Vert\nu_j) = \dkl(\mu_i\Vert\mu_j) + \varepsilon$, and all other divergences are equal. Hence the difference in cost between $\nu$ and the experiment $\mu$ is given by $\beta_{ij}$ times the difference $\varepsilon$ in the expected log-likelihood ratio. The result formally justifies the interpretation of each coefficient $\beta_{ij}$ as a marginal cost.}

The specification of the parameters $(\beta_{ij})$ must of course depend on the particular application at hand. Consider, for instance, a doctor who must choose a treatment for a patient displaying a set of symptoms, and who faces uncertainty regarding their cause. In this example, a state of nature $i$ represents a possible pathology affecting the patient. In order to distinguish between two possible diseases $i$ and $j$ it is necessary to collect samples and run tests, whose costs will depend on factors that are specific to the two conditions, such as their similarity, or the prominence of their physical manifestations. These differences in costs can then be reflected by the coefficients $\beta_{ij}$ and $\beta_{ji}$. For example, suppose that  $i$ and $i'$ are two types of viral infections, $k$ is a bacterial infection, and $i$ and $i'$ are difficult to tell apart, but telling $i$ and $k$ apart is easier. This can be captured  by setting $\beta_{ii'} > \beta_{ik}$. In \S\ref{sec:verification} we discuss environments where the coefficients might naturally be assumed to be asymmetric, in the sense that $\beta_{ij} \neq \beta_{ji}$.%
 \footnote{Since we do not impose symmetry axioms, it is in a sense a natural finding that the LLR cost function can capture differences in the costs of learning about different states. It is perhaps more surprising that the cost function has a relatively small set of $n(n-1)$ parameters, where $n$ is the number of states.
 }
 In environments where no pair of states is a priori harder to distinguish than another,
a simple choice is to set all the coefficients $(\beta_{ij})$ to be equal.\footnote{An example common in the literature \citep[e.g.,][]{christie} is that of a detective who has to discover which member of a finite group of people committed a violent crime in some isolated setting, such as a train.} Finally, in the next section  we propose a specific functional form in the more structured case where states represent a one-dimensional quantity.

We end this section by noting that the  LLR cost function is monotone with respect to the Blackwell order:
\begin{prop}
\label{prop:llr-monotnoe}
Let $\mu$ and $\nu$ be experiments such that $\mu$ Blackwell dominates $\nu$. Then every LLR cost $C$ satisfies $C(\mu) \geq C(\nu)$.
\end{prop}

\section{Learning about a One-Dimensional State}\label{sec:betas}

 Many information acquisition problems involve learning about a one-dimensional characteristic, so that each state $i$ is a real number. In macroeconomic applications, the state may represent the inflation rate. In perceptual experiments, the state can correspond to the number of red/blue dots on a screen. In a polling problem, the state may correspond to the number of voters voting for a given party.  Alternatively, $i$ might represent a physical quantity to be measured.

 In this section we propose a choice of parameters $(\beta_{ij})$ for one-dimensional information acquisition problems. Given a problem where each state $i\in\Theta \subset \RR$ is a real number, we propose to set each coefficient $\beta_{ij}$ to be equal to $\frac{\kappa}{ (i-j)^{2}}$ for some constant $\kappa \geq 0$. Each $\beta_{ij}$ is therefore inversely proportional to the squared distance between the corresponding states $i$ and $j$. Under this specification, two states that are closer to each other are harder to distinguish.

 The main result of this section shows that this choice of parameters captures two main hypotheses: (a) the difficulty of producing a signal that allows to distinguish between states $i$ and $j$ is a function only of the distance $|i - j|$ between the two, and (b) the cost of a noisy measurement of the state with standard normal error is the same across information acquisition problems. Both assumptions take as a working hypothesis that the cost of making a measurement depends only on its precision, and not on the other aspects of the model, such as the set of states $\Theta$. For example, the cost of measuring the height of a person with a given instrument does not depend on whether the person's height is known to be in  $\Theta = \{190,\ldots,210\}$ or $\Theta' = \{160,\ldots,180\}$.
 
 
 \medskip
 
 Let $W$ be a nonempty open interval of $\RR$; we think of this set as the range of reasonable values of the state, where our hypotheses apply. We denote by $\mathcal{T}$ the collection of finite subsets of $W$ with at least two elements. Each set $\Theta \in \mathcal{T}$ represents the set of states of nature in a different, one-dimensional, information acquisition problem. To simplify the language, we refer to each $\Theta$ as a \textit{problem}. For each $\Theta \in \mathcal{T}$ we are given an LLR cost function $C^\Theta$ with coefficients $(\beta^\Theta_{ij})$. The next two axioms formalize the two hypotheses described above by imposing restrictions, across problems, on the cost of information. 

 The first axiom states that $\beta_{ij}^\Theta$, the marginal cost of increasing the expected LLR between two states $i$ and $j$ is a function of the distance between the two, and is unaffected by changing the values of the other states.

\renewcommand*{\theaxb}{\alph{axb}}
\begin{axb}\label{axb1}
For all $\Theta,\Xi \in \mathcal{T}$ such that $|\Theta|=|\Xi|$, and for all $i,j \in \Theta$ and $k,l \in \Xi$, 
\[
\text{if~~} |i - j| = |k - l|  \text{~~then~~} \beta^\Theta_{ij} = \beta^{\Xi}_{kl}.
\]
\end{axb}


For each $i \in W$ we denote by $\zeta_i$ a normal probability measure on the real line with mean $i$ and variance 1. Given a problem $\Theta$, we denote by $\zeta^{\Theta}$ the experiment $(\RR, (\zeta_i)_{i\in\Theta})$. This is the canonical experiment consisting of a noisy measurement of the state plus standard normal error. Expressed differently, if $i \in \Theta$ is the true state, then the outcome of the experiment $\zeta^\Theta$ is distributed as $s = i + \varepsilon$, where $\varepsilon$ is normally distributed with mean zero and variance $1$ independent of the state. The next axiom states that the cost of such a measurement does not depend on the particular values that the state can take.

\begin{axb}\label{axb2}
For all $\Theta,\Xi \in \mathcal{T}$, $C^\Theta(\zeta^\Theta) = C^{\Xi}(\zeta^{\Xi})$.
\end{axb}

%

Axioms \ref{axb1} and \ref{axb2} lead to a simple parametrization for the coefficients of the LLR cost in one-dimensional information acquisition problems:

\begin{prop}\label{prop:onedimens}
The collection $C^\Theta, \Theta \in \mathcal{T},$ satisfies Axioms \ref{axb1} and \ref{axb2} if and only if there exists a constant $\kappa > 0$ such that for all $i,j \in \Theta \text{~and~} \Theta \in \mathcal{T}$,
\[
    \beta^\Theta_{ij} = \frac{\kappa}{n(n-1)}\,\,\frac{1}{(i-j)^{2}}
\]
where $n$ is the cardinality of $\Theta$.
\end{prop}

 Thus, under Axioms \ref{axb1} and \ref{axb2} each coefficient $\beta^\Theta_{ij}$ is decreasing in the distance between the states, so that distinguishing states that are closer to each other is more costly. Each coefficient is also divided by a factor $n(n-1)$ that normalizes the cost with respect to the number of states. This is an implication of Axiom b, which states that the cost of performing a noisy measurement does not depend on the particular values the state can take. As we show in the proof, the quadratic term $(i-j)^2$ in the expression of the coefficients is related to the assumption, in the same axiom, of normally distributed noise. In the Appendix we show how the results can be extended to different families of distributions.
 
 Applied to normal experiments, Proposition~\ref{prop:onedimens} implies that for any $\Theta \in \mathcal{T}$, a normal experiment with mean $i$ and variance $\sigma^2$ has cost $\kappa\sigma^{-2}$ proportional to its precision. Thus, this functional form generalizes a specification often found in the literature---where the cost of a normal experiment is assumed to be proportional to its precision \citep*{wilson1975informational}---to arbitrary information structures that are not necessarily normal.
 
 As we will see in \S\ref{sec:stochastic}, the specification of Proposition~\ref{prop:onedimens} allows to compute numerical solutions, and thus can be useful for deriving quantitative predictions. At the same time, this functional form may be too simple to capture certain intuitive comparative statics with respect to changes of the state space. For example, the precision of a measurement made using a measuring tape is quite different when measuring a person's height than when measuring the length of a field. More generally, any measurement instrument has a range of reliability, and as one moves toward the extremes it becomes noisier. We partially address this issue by allowing the state to only take value in some interval $W \subseteq \RR$.
 
 Axiom~\ref{axb1} assumes that only the  distance between states determines the cost of an experiment. But in many situations states with a given distance are harder to distinguish at larger scales. Consider, for instance, a subject in a laboratory experiment who is asked to guess the number of pennies in a jar. A problem where this state can take the values either 1 or 2 is easier than a problem where the state can take the values 101 or 102.
 
 Such examples form a well known empirical regularity in psychophysics, known as \textit{Weber's Law}, according to which the change in stimulus intensity that is necessary for subjects to exhibit a certain response is a constant fraction of the starting intensity of the stimulus. A way to model such situations is to change the units in which states are measured by applying a logarithmic transformation to the states.  This is equivalent to changing Axiom~\ref{axb1} to consider ratios between states instead of differences, and changing Axiom~\ref{axb2} to consider log-normal measurement errors instead of normal. The resulting coefficients are
 \[
    \beta^\Theta_{ij} = \frac{\kappa}{n(n-1)}\,\,\frac{1}{(\log i/j)^2}\,.
 \]
 This results in predictions in line with Weber's Law, making it easier to distinguish 1 from 2 than 101 from 102.

\section{Illustrative Examples}\label{sec:examples}

\subsection{LLR Cost for Normal and Binary Experiments}
Closed form solutions for the Kullback-Leibler divergence between standard distributions such as normal, exponential or binomial, are readily available. This makes it immediate to compute the cost  of common parametric families of experiments.

\paragraph{Normal Experiments.} Consider a normal experiment $\mu^{m, \sigma}$ according to which the signal $s$ is given by
\[
    s = m_i + \varepsilon
\]
where the mean $m_i \in \RR$ depends on the true state $i$, and $\varepsilon$ is state independent and normally distributed with standard deviation $\sigma$. In this example, each $m_i$ is a feature of the information structure: choosing an experiment where the distances between states $|m_i - m_j|$ are higher provides stronger information about the states.

By substituting \eqref{eq:thm_1} with the well-known expression for the Kullback-Leibler divergence between normal distributions, we obtain that the cost of such an experiment is given by
\begin{equation}\label{eq:normal-signal}
    C(\mu^{m, \sigma}) = \sum_{i,j}\beta_{ij}\frac{(m_{j}-m_{i})^{2}}{2\sigma^{2}}\,.
\end{equation}
The cost is decreasing in the variance $\sigma^2$, as one may expect. Increasing $\beta_{ij}$ increases the cost of an experiment $\mu^{m, \sigma}$ by a factor that is proportional to the squared distance between the means of the two experiments.

\paragraph{Binary Experiments.}
Another canonical example is the binary-binary setting in which the set of states is $\Theta = \{H,L\}$, and the experiment $\nu^p = (S,(\nu_i))$ is also binary: $S = \{0,1\}$, $\nu_H = B(p)$ and $\nu_L = B(1-p)$ for some $p > 1/2$, where $B(p)$ is the Bernoulli distribution on $\{0,1\}$ assigning probability $p$ to 1. In this case, the cost increases in $p$ and given by
\begin{equation}\label{eq:binary-signal}
    C(\nu^p) = (\beta_{HL}+\beta_{LH})\left[p\log\frac{p}{1-p}+(1-p)\log\frac{1-p}{p}\right]\,.
\end{equation}

\subsection{Hypothesis Testing}

 In this section, we apply the log-likelihood ratio cost to a standard hypothesis testing problem. We study a decision maker performing an experiment with the goal of learning about a hypothesis, i.e.\ whether the state is in a subset $H \subset \Theta$.

 We consider an experiment that reveals with some probability whether the hypothesis is true or not, and study how its cost depends on the structure of $H$. For a given hypothesis $H$ and a precision $\alpha$ let $\mu$ be the symmetric binary experiment with signal realizations $S = \{H,H^c\}$, where $H^c$ denotes the complement of $H$:
\begin{equation}
    \mu_i( s ) = \begin{cases}
    \alpha &\text{ for } i \in s\\
    1-\alpha &\text{ for } i \notin s
    \end{cases} 
\end{equation}
Conditional on any state, this experiment yields a correct signal with probability $\alpha$.
Under LLR cost, the cost of such an experiment is given by
\begin{equation}\label{eq:hypothesis}
    \left(\sum_{i \in H, j \in H^c}\beta_{ij} + \beta_{ji}\right)\left(\alpha \log\frac{\alpha}{1 - \alpha} + (1 - \alpha) \log\frac{1 - \alpha}{\alpha}\right)
\end{equation}
The first term captures the difficulty of discerning between $H$ and $H^c$. The harder the states in $H$ and $H^c$ are to distinguish, the larger the sum of the coefficients $\beta_{ij}$ and $\beta_{ji}$ will be, and the more costly it will thus be to learn whether the hypothesis $H$ is true.
The second term is monotone in the precision $\alpha$ and is independent of the hypothesis.
We illustrate with an example how this captures the fact that testing two different hypotheses can lead to very different costs even if they involve the same number of states.

\label{gdp}\paragraph{Learning about the GDP.} For concreteness, we take a state to be a natural number $i$ in the interval $\Theta = \{20000,\ldots,80000\}$, representing the current US GDP per capita. We consider the following two hypotheses:
\begin{enumerate}
\item[(H1)] The GDP is above 50000.
\item[(H2)] The GDP is an even number.
\end{enumerate}
Intuitively, producing enough information to answer with high accuracy whether H1 is true should be less expensive than producing enough information to answer whether H2 is true, a practically impossible task. 
Our model captures this intuition: As the state is one-dimensional, we set $\beta_{ij} = \kappa/(i - j)^2$ following \S\ref{sec:betas}; the same qualitative conclusion will hold as long as $\beta_{ij}$ is strictly decreasing in the distance $|i-j|$. Then
\begin{align*}
    \sum_{i \in \mathit{H1}, j \in \mathit{H1}^c}\beta_{ij} + \beta_{ji} \approx 22 \, \kappa\hspace{2cm}
    \sum_{i \in \mathit{H2}, j \in \mathit{H2}^c}\beta_{ij} + \beta_{ji} \approx 148033\, \kappa.
\end{align*}
That is, learning whether the GDP is even or odd is by several orders of magnitude more costly than learning whether the GDP is above or below $50000$.\footnote{\label{roland-example}Beyond the challenge of learning about the state, which is the focus of this paper, it might be computationally difficult to determine the set that corresponds to a given hypothesis. Consider, for example, the hypothesis (H1) The number of pages in this manuscript is an even number, vs the hypothesis (H2) The number of pages in this manuscript is greater than $\sqrt{4000}$. The relative ``distance'' properties of the states are in both cases exactly the same as in the GDP example, but the cost of telling states apart is considerably higher in the high-distance case than in the low distance one. We thank the editor for suggesting this example.}

It is useful to compare these observations with the results that would be obtained under mutual information and a uniform prior on $\Theta$. In such a model, the cost of a symmetric binary experiment with precision $\alpha$ is determined solely by the cardinality of $H$. In particular, under mutual information learning whether the GDP is above or below $50000$ is equally costly as learning whether it is even or odd. This is a well known property of cost functions that are invariant with respect to a relabelling of the states.

 \section{Information Acquisition in Decision Problems}
\label{sec:stochastic}

 In this section we study the implications of the log-likelihood ratio cost function for decision problems. We consider a decision maker choosing an action $a$ from a finite set $A$. The payoff from $a$ depends on the state $i$ and is given by $u(a,i)$. The agent is endowed with a full-support prior $q$ over the set of states. Before making her choice, the agent can acquire an experiment $\mu \in \mathcal{E}$ at cost $C(\mu)$, where $C$ is an LLR cost function where the coefficients $(\beta_{ij})$ are assumed to be positive. 
 
 As is well known, for a cost function that is monotone with respect to the Blackwell order, it is without loss of generality to restrict attention to experiments where the set of realizations $S$ equals the set of actions $A$, and to assume that upon observing a signal $s = a$ the decision maker will choose the action recommended by the signal. Throughout this section, we will therefore identify an experiment $\mu$ with a vector of probability measures over actions $\mu \in \P(A)^{n}$ where $n=|\Theta|$.

 An optimal experiment $\mu^\star = (\mu^\star_i)$ solves
 \begin{equation} \label{eq:choice-probabilities}
	\mu^\star \in \argmax_{\mu \in \P(A)^n} \, \sum_{i \in \Theta} q_i \left( \sum_{a \in A} \mu_i (a) u(a,i) \right) - C(\mu)  \,.
 \end{equation}
 Hence, the optimal action $a$ is chosen in state $i$ with probability $\mu^\star_i(a)$. The maximization problem  \eqref{eq:choice-probabilities} is well behaved: the maximand is upper-semicontinuous and concave (see Proposition~\ref{prop:llr-convex} in the Appendix), and there always exists an optimal solution.\footnote{To establish existence of an optimal solution, recall that the Kullback-Leibler divergence $\dkl \colon \P(A) \times \P(A) \to [0,\infty]$ is a lower-semicontinuous function \citep[][Lemma 1.4.3]{dupuis2011weak}. The maximand in \eqref{eq:choice-probabilities}, being a sum of upper-semicontinuous functions, is upper-semicontinuous. Since $\P(A)^n$ is compact, the problem admits a solution.} Thus, an optimal experiment can be found by applying standard methods in concave optimization.%
 
 It is without loss of generality to restrict attention to choice probabilities where an action that is chosen with strictly positive probability in one state is chosen with strictly positive probability in every state, since otherwise the experiment is not in the domain $\mathcal{E}$.


\subsection{Implications for Optimal Choice Probabilities}

 We obtain a characterization of the decision maker's optimal choice probabilities. The characterization is based on the study of first-order conditions, and is therefore analogous to that obtained by \cite*{matvejka2015rational} for the case of mutual information cost.

 The result is based on a standard economic intuition. For choice probabilities to be optimal, the marginal benefit of choosing an action $a$ marginally more often then a different action $b$ must exactly offset its marginal cost. Formally, given a vector $\mu$ of choice probabilities, we denote by $\mathrm{supp}(\mu)$ the support of $\mu$, i.e.\ the set of actions which are played with strictly positive probability under $\mu$.\footnote{That is, $\mathrm{supp}(\mu) = \{a \in A\colon \mu_i(a) > 0 \text{ for some } i \in \Theta \}\,.$ } Given two actions $a$ and $b$ in the support of $\mu$, consider perturbing $\mu$ by increasing  $\mu_i(a)$  while decreasing $\mu_i(b)$ by the same amount. The marginal benefit of this perturbation is denoted by $\mathrm{MB}_i(a,b)$ and is equal to
 \[
    \mathrm{MB}_i(a,b) = q_i \left[ u(a,i) - u(b,i) \right].
 \]
 Such a transfer of probabilities has an effect on the information cost of the experiment $\mu$. This is given by the expression:
 \begin{equation}\label{eq:mc}
    \mathrm{MC}_i(a,b) = \sum_{j \in \Theta} \beta_{ij}\left(\log \frac{\mu_i(a)}{\mu_j(a)} - \log \frac{\mu_i(b)}{\mu_j(b)} \right) - \sum_{j \in \Theta} \beta_{j i}\left( \frac{\mu_j(a)}{\mu_i(a)}  - \frac{\mu_j(b)}{\mu_i(b)} \right).
 \end{equation}
 It measures the change in information acquisition cost necessary to choose action $a$ marginally more often and action $b$ marginally less often. For the choice probabilities $\mu$ to be chosen optimally, this change in information cost must equal the difference $q_i \left[ u(i,a) - u(i,b) \right]$ in expected benefits. This is the content of the next proposition.
 
%
 \begin{prop}\label{prop:foc}
Let $\mu = (\mu_i)_{i\in\Theta}$ be the vector of choice probabilities that solves the optimization problem \eqref{eq:choice-probabilities}. Then, for every state $i \in \Theta$  it holds that 
 \begin{equation}\label{eq:foc-stochastic-chocice}
    \mathrm{MB}_i(a,b) = \mathrm{MC}_i(a,b) \, \text{~~~for all~} a,b \in \mathrm{supp}(\mu)\,.
 \end{equation}

\end{prop}

 Figure~\ref{fig:focs} illustrates this result in a simple decision example with two states and two actions where the decision maker's goal is to match the state. 
 Proposition \ref{prop:foc} characterizes the optimal choice probabilities in terms of necessary first-order conditions. These conditions are in general not sufficient, because they do not verify that the support of $\mu$ is optimal. In the case of mutual information, \cite*{caplin2016rational} and \cite*{denti2020note} give a characterization of the set of actions that are taken with positive probability, and arrive at first-order conditions that are both sufficient and necessary. We do not know whether analogous first-order conditions can be obtained for the LLR cost function.
 
 
 \begin{figure}[t!]
    \centering
    \includegraphics[scale=0.8]{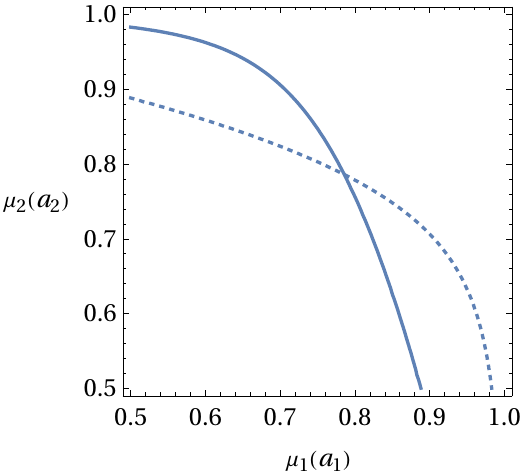}
    \caption{A decision problem where $\Theta = \{1,2\}$ and $A = \{a_1,a_2\}$, the prior $q$ is uniform, $\beta_{12} = \beta_{21} = 1$, and payoffs are $u_1(a_1) = u_2(a_2) = 3$ and $u_1(a_2) = u_2(a_1) = 0$. The solid line is the locus of choice probabilities such that $\mathrm{MB}_1(a_1,a_2) = \mathrm{MC}_1(a_1,a_2)$. The dotted line is the locus where $\mathrm{MB}_2(a_2,a_1) = \mathrm{MC}_2(a_2,a_1)$. The optimal vector of choice probabilities is given by the intersection of the two curves.}
    \label{fig:focs}
\end{figure}
 
 

\subsection{Continuity of Choice Probabilities}



 A feature of the LLR cost is its ability to model the fact that closer states are harder to distinguish, in the sense that acquiring information that finely discriminates between them is more costly. This, in turn, suggests that choice probabilities cannot vary abruptly across nearby states. 

 To formalize this intuition, we assume that the state space $\Theta$ is endowed with a distance $d\colon \Theta \times \Theta \to \RR$. 
 We say that \emph{nearby states are hard to distinguish} if for all $i,j\in \Theta$
\begin{equation}
    \label{eq:hard}
\beta_{ij} \geq \frac{1}{d(i,j)^2} \,.
\end{equation}

Under this assumption the cost of acquiring information that discriminates between states $i$ and $j$ is high for states that are close to each other. 
Our next result shows that when nearby states are hard to distinguish, the optimal choice probabilities are Lipschitz continuous in the state: the agent will choose actions with similar probabilities in similar states. For this  result, we denote by $\Vert u \Vert = \max_{i,a}|u(a,i)|$ the norm of the decision maker's utility function.

%
%
 \begin{prop}[Continuity of Choice]\label{prop:choice-continuity} Suppose that nearby states are hard to distinguish. Then the optimal choice probabilities $\mu^\star$ solving \eqref{eq:choice-probabilities} are uniformly Lipschitz continuous with constant $\sqrt{\Vert u \Vert}$, i.e.\ satisfy
 \begin{equation}\label{eq:Lipschitz}
 	\sum_{a \in A} \left| \mu^\star_i(a) - \mu^\star_j(a) \right| \leq \sqrt{ \Vert u \Vert} \, d(i,j) \, \text{~~~for all~} i,j \in \Theta.
 \end{equation}
 \end{prop}
 

  Lipschitz continuity is a standard notion of continuity in discrete settings, such as the one of this paper, where the relevant variable $i$ takes finitely many values. A crucial feature of the bound \eqref{eq:Lipschitz} is that the Lipschitz constant depends only on the norm $\Vert u \Vert$ of the utility function, independently of the exact form of the coefficients $(\beta_{ij})$, and of the number of states.\footnote{
  Proposition \ref{prop:choice-continuity} suggests that the analysis of choice probabilities might be extended to the case where the set of states $\Theta$ is an interval in $\RR$, or, more generally, a metric space. Given a (possibly infinite) state space $\Theta$ endowed with a metric, and a sequence of finite discretizations $(\Theta_n)$ converging to $\Theta$, the bound \eqref{eq:Lipschitz} implies that if the corresponding sequence of choice probabilities converges, then it must converge to a collection of choice probabilities that are continuous, and moreover Lipschitz.
  }
  In addition, assumption \eqref{eq:hard} can be generalized to arbitrary ordinal transformations of the distance $d$. The proof of Proposition~\ref{prop:choice-continuity} shows that if the coefficients satisfy $\beta_{ij} \geq 1/f(d(i,j))^2$ for a monotone increasing function $f$, then the conclusion of the proposition holds with the right hand side of \eqref{eq:Lipschitz} replaced with $\sqrt{ \Vert u \Vert} \, f(d(i,j))$.
 
  

  This result highlights a contrast between the predictions of mutual information cost and LLR cost. Mutual information predicts behavior that displays a discontinuity with respect to the state (see \S\ref{sec:perception-tasks} for an example). Under LLR cost, when nearby states are harder to distinguish, the change in choice probabilities across states can be bounded by the distance between them.
 
  
  This difference has stark implications in coordination games.
  \cite*{morris2016coordination} study information acquisition in coordination problems. In their model, continuity of the choice probabilities with respect to the state leads to a unique equilibrium; if continuity fails, then there are multiple equilibria. This suggests that different choices of information cost can lead to different predictions in coordination games and their economic applications. 


 
\subsection{Comparative Statics with Respect to the Coefficients $\beta_{ij}$}

 While so far we have focused on the effect that the coefficients $\beta_{ij}$ have on the cost of a given experiment, we now address the question of their effect on behavior. The next proposition is a comparative statics result describing how choice probabilities vary with the parameters $\beta_{ij}$.
 
 \begin{prop}\label{prop:comparative-statics-in-beta}
 Consider a decision problem, and let $\mu$ and $\mu'$ be the optimal choice probabilities obtained under an LLR cost function with coefficients $(\beta_{ij})$ and $(\beta'_{ij})$, respectively. Then
 \[
    \sum_{i \neq j} \big(\beta'_{ij} - \beta_{ij}\big)\big(\dkl(\mu'_i \Vert \mu'_j\big) - \dkl\big(\mu_i \Vert \mu_j)\big) \leq 0. 
 \]
 \end{prop}
 
 All other things equal, an increase of the coefficient $\beta_{ij}$ decreases the Kullback-Leibler divergence $D(\mu_i \Vert \mu_j)$ between the corresponding optimal choice probabilities, and thus makes the decision maker's behavior more similar in the two states.
 
 Proposition~\ref{prop:comparative-statics-in-beta} follows from the same logic underlying the law of supply in standard microeconomic models of production. Under the LLR cost function, the decision maker solves an optimization problem that is mathematically equivalent to a profit maximization problem. Each expected log-likelihood ratio $\dkl(\mu_i \Vert \mu_j)$ is an intermediate ``input'' which accrues to the decision maker's expected payoff. Each such input is ``priced'' according to a linear price $\beta_{ij}$. The comparative statics described by the result follows from such a linearity property, together with a standard revealed-preference argument.

\subsection{Identifying the Cost from Observed Choices}\label{sec:identifying-the-cost}

 Proposition~\ref{prop:foc} can be applied to the problem of identifying and testing the LLR model from observed choices. We illustrate this in the context of a simple example.
 We consider a binary choice problem where we are given two a priori equally likely states $\Theta = \{1,2\}$. The agent can take one of two actions, $a_1$ and $a_2$, and receives a payoff $v > 0$ if the action matches the state and $0$ otherwise.

 An analyst observes the agent's choice probabilities $(\mu_i(a))_{i \in \Theta, a \in A}$, and is interested in testing if such probabilities are consistent with LLR cost. This is true if there exist coefficients $(\beta_{12},\beta_{21})$ that satisfy equation \eqref{eq:foc-stochastic-chocice}. The equation simplifies to
\begin{equation}\label{eq:foc-example}
\begin{aligned}
    \frac{v}{2} &= -\left[\beta_{12} \left( \log l_1 - \log l_2 \right) +\beta_{21}\left( l_1 - l_2 \right)\right]\\
    -\frac{v}{2} &= -\left[\beta_{21} \left( -\log l_1 + \log l_2 \right) +\beta_{12}\left( 1/l_1-1/l_2\right)\right] \,.
    \end{aligned}
\end{equation}
where $l_1= \frac{\mu_2(1)}{\mu_1(1)}$ and $l_2 = \frac{\mu_2(2)}{\mu_1(2)}$.
Rearranging the above conditions yields that one can infer the information cost parameters $(\beta_{ij})$ from her choice probabilities $\mu$ as
\begin{equation}
\label{eq:identified-beta-example}
\begin{aligned}
    \beta_{12} &= \frac{v}{2} \frac{l_2-l_1+\log \frac{l_1}{l_2}}{\frac{(l_1-l_2)^2}{l_1 l_2}- (\log \frac{l_1}{l_2})^2} \hspace{2cm}
    \beta_{21} = \frac{v}{2} \frac{ \frac{l_2-l_1}{l_1 l_2} +  \log \frac{l_1}{l_2}}{ \frac{(l_1-l_2)^2}{l_1 l_2}-(\log \frac{l_1}{l_2})^2} \,.
\end{aligned}
\end{equation}

 For example, if the agent takes the correct action $80\%$ of the time in state $1$ and $60\%$ of the time in state $2$, we have that $(\mu_1(1),\mu_1(2),\mu_2(1),\mu_2(2))=(0.8,0.2,0.4,0.6)$ and the above formula yields that $(\beta_{12},\beta_{21}) \approx (0.37 v,  -0.07 v)$.
As the implied $\beta_{21}$ is negative these choice probabilities are inconsistent with any LLR cost function and this type of choice behavior would reject our model.
In contrast, if the agent takes the correct action $80\%$ of the time in state $1$ and $70\%$ of the time in state $2$, we have that $(\mu_1(1),\mu_1(2),\mu_2(1),\mu_2(2))=(0.8,0.2,0.3,0.7)$ which implies that $(\beta_{12},\beta_{21})\approx(0.18 v,0.03 v)$, and thus that this choice behavior can be explained by an LLR cost.
\begin{figure}[t!]
    \centering
    \includegraphics[scale=0.8]{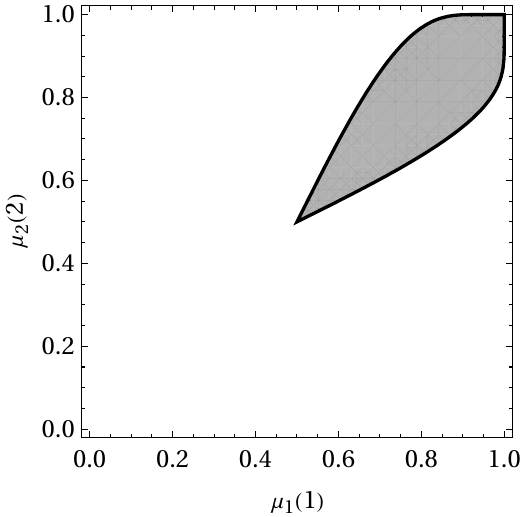}
    \caption{The probabilities of choosing correctly in state 1 and state 2 that are consistent with LLR cost.}
    \label{fig:consistent-prob}
\end{figure}
Figure~\ref{fig:consistent-prob} more generally depicts all probabilities of choosing correctly in state 1 and state 2 that are consistent with LLR cost.

 This example illustrates how an analyst could use choice data to either reject LLR cost or to identify the information cost parameters $(\beta_{12},\beta_{21})$. In Appendix~\ref{sec:identification} we formally show that choice probabilities are consistent with LLR cost if and only if a solution of the form \eqref{eq:identified-beta-example} exists.
 
 In general, when there are more than two states and actions the analyst might need data from multiple decision problems to point identify $\beta$. For a general decision problem  the model admits $|\Theta|\Big(|\Theta|-1\Big)$ degrees of freedom and \eqref{eq:foc-stochastic-chocice} imposes $|\Theta| \times \frac{1}{2} |A| \times (|A|-1)$ linear equations on $\beta$ which suggests that to identify the analyst needs to observe behavior in
\[
    \frac{|\Theta|-1}{\frac{1}{2} |A| \times (|A|-1)}
\]
decision problems. Given the data, solving numerically from the coefficients $\beta$ is easy as the corresponding system of equations is linear.\footnote{Due to the linear structure of the implied restrictions, one could also construct finite sample tests for the LLR model using standard econometric methods, but this is beyond the scope of this paper.}


\subsection{Perception Tasks}\label{sec:perception-tasks}
 
 In this section we study the implications of the LLR cost function for perception tasks, a well known and long studied family of decision problems. In a perception task an agent is shown an even number of dots, with each dot  either red or blue. The agent guesses whether there are more blue or red dots, and get rewarded if they guess correctly.
 
 
 First, the total number $n$ of dots is fixed. Then, subjects are told the value of $n$, and that the number $i$ of red dots will be drawn uniformly from the set $\Theta = \{0, \ldots, n/2 -1, n/2 + 1,\ldots, n\}$. The state where the number of blue and red dots is exactly equal to $n/2$ is ruled out to simplify the exposition. The set of actions is $A=\{R,B\}$ and the utility function is
 \[
    u(a,i) = \begin{cases} 1 &\text{ if } a=B \text{ and } i > n/2\\
    1 &\text{ if } a = R \text{ and } i < n/2\\
    0 &\text{ otherwise. }
    \end{cases}
 \]
 
 \label{r1:1}Such perception tasks can be used to model many applied learning problems.
 For example, each dot could correspond to a voter whose color indicates whether they vote for the red or blue party and the agent is an analyst trying to predict which party will obtain the majority of votes in the election.\footnote{Polling provides an interesting example of flexible information acquisition. Even if the only basic experiment available to a pollster is to call a voter and ask for her opinion, practically any experiment can be constructed as a compound experiment by deciding when to stop polling. I.e., the pollster with prior $p$ can choose (perhaps at random) thresholds $p_1 < p < p_2$ and keep polling until her posterior reaches either $p_1$ or $p_2$. See \cite{morrisstrack2018} for a formalization of this idea.}
 In a typical experiment subjects observe 100 dots each of which is either red or blue on a screen \citep[see, e.g.][]{caplin2013behavioral,dean2017experimental} and are asked whether there are more red or blue dots. 
 
 As in the case of binary decision problems, it is without loss of generality to assume that $\mu_i(B)$ is strictly between $0$ and $1$ in every state. For a vector of distributions over actions $(\mu_i)$, the decision maker guesses correctly in state $i$  with probability
 \[
    m_i = \begin{cases} \mu_i(B) &\text{ if } i > n/2 \\
    \mu_i(R) &\text{ if } i < n/2.
    \end{cases}
\]
 Intuitively, it should be harder to guess correctly when the difference in the number of dots of different colors is small, i.e.\ when $i$ is close to $n/2$. For example, it should be harder to predict the winner in a close election than in an election where one of the candidates has a large lead. 
 Also, it is a well established fact in the psychology\footnote{See, e.g., Chapter 7 in \cite{green1966signal} or Chapter 4 in \cite{gescheider2013psychophysics}.}, neuroscience\footnote{E.g., \cite{krajbich2010visual, tavares2017attentional}.} and economics\footnote{See, e.g., \cite{mosteller1951experimental}.} literatures that so called {\em psychometric functions}---the relation between the strength of a stimulus offered to a subject and the probability that the subject identifies this stimulus---are sigmoidal (i.e.\ S-shaped), so that the probability that a subject chooses $B$ transitions smoothly from values close to $0$ to values close to $1$ when the number of blue dots increases.

 As \cite*{dean2017experimental} note, under mutual information cost (and a uniform prior, as in the experimental setup described above), the optimal experiment $\mu^\star$ must induce a probability of guessing correctly that is state-independent.\footnote{It is well known that under mutual information costs the physical features of the states (such as  distance or similarity) do not affect the cost of information acquisition \citep*[see, e.g.,][]{mackowiak2018rational}.}
%
%
 As shown by \cite*{matvejka2015rational}, conditional on a state $i$, the log-likelihood ratio $\log(\mu_i(B)/\mu_i(R))$ between the two actions must equal the difference in payoffs $u(B,i) - u(R,i)$, up to a constant. Hence, the probability of a correct choice must be the same for any two states that lead to the same utility function over actions, such as the state in which there are 51 blue dots out of 100 and the state in which there are 99 blue dots.


\begin{figure}[thp]
\centering
\includegraphics[scale=0.8]{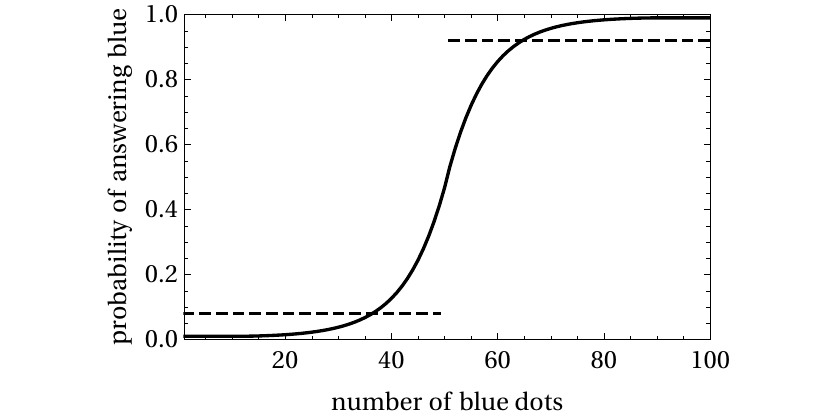}
\caption{Predicted probability of guessing that there are more blue dots as a function of the state, for the LLR cost with $\beta_{ij}=1/(i-j)^2$ (solid line) and for mutual information cost (dashed line).}
\label{fig:jump}
\end{figure}





 As this is a one-dimensional information acquisition problem, we can apply the specification $\beta_{ij} = \kappa/(i-j)^2$ of the LLR cost. As can be seen in Figure \ref{fig:jump}, this LLR cost predicts a sigmoidal relation between the state and the choice probability.
 Thus, the model matches the qualitative features of choice probabilities commonly observed in practice. Of course, this could be similarly achieved using other cost functions that take into account the difficulty of distinguishing between similar states, such as the neighborhood-cost function introduced by \cite{hebert2020neighborhood}.
 
 To gain additional insight, we now consider a more basic assumption on the cost function. Rather than assuming a particular specification, we assume that the coefficients $(\beta_{ij})$ are \textit{strictly decreasing in the distance between states}: There exists a positive and strictly decreasing function $f$ such that $\beta_{ij} = f(\vert i - j \vert)$ for all pairs of states. The condition captures the idea that states that are closer to each other are harder to distinguish. Even under this general non-parametric assumption, the LLR cost function leads to the intuitive prediction that the decision maker will guess correctly with strictly higher probability when the difference in the number of dots of different colors is smaller:
 
 \begin{prop}\label{prop:perception-monotonicity}
 Consider the above perception task. Let $C$ be an LLR cost function where the parameters $(\beta_{ij})$ are strictly increasing in the distance between states. Then, the resulting optimal probabilities $(m_i)$ of guessing correctly satisfy $m_i > m_j$ whenever $\vert i - \frac{n}{2} \vert > \vert j - \frac{n}{2} \vert$.
 \end{prop}

 \subsection{The Effect of Greater Incentives}
\label{sec:effect-of-incentives}

 We now apply the characterization of Proposition~\ref{prop:foc} to study in more detail the classic problem of predicting the probability of choosing between two options as a function of their relative values. In its simplest implementation, it consists of a task where there are two equally likely states, a subject must choose between two actions $a_1$ and $a_2$, and each action yields a reward with payoff $v \in \mathbb{R}$ when chosen in the corresponding state, and $0$ otherwise. Compared to \S\ref{sec:identifying-the-cost}, we focus here on the question of how the decision maker's behavior varies as a function of $v$.
 
 \label{r1:2} In order to interpret changes in the parameter $v$ it is necessary to fix a cardinal representation of payoffs and to define an interval of possible values for $v$. If the decision maker is risk neutral and rewards are monetary, then $v$ can represent the amount paid to the subject. If the decision maker is risk averse or her risk attitudes are unknown, then subjects can be paid using probabilistic prizes.\footnote{To illustrate, let $x$ and $y$ be two monetary prizes, with $x > y$. We continue to assume that the decision maker's preferences are consistent with expected utility, and normalize, without loss of generality, their utility function to assign utility $1$ to $x$ and utility $-1$ to $y$. A lottery that delivers $x$ with probability $p$ and $y$ with probability $1-p$ has expected utility $2p - 1$. We define each payoff $v$ in the interval $[-1,1]$ as the expected utility of such a lottery. This approach is well known in the implementation of scoring rules, where it allows to reward a decision maker using a linear payoff, and circumvents the need of eliciting the decision maker's degree of risk aversion \citep[see, for example,][and the references therein]{lambert2011probability,sandroni2013eliciting}. The same approach has been more recently applied in rational inattention by \cite*{caplin2020rational}.}
 
 
 The next result derives the optimal choice probabilities in a binary choice problem under a symmetric LLR cost function. Without loss of generality we restrict our attention to choice probabilities where both actions are chosen with strictly positive probability in every state. The result follows by rearranging the optimality conditions of Proposition~\ref{prop:foc}.
 
 \begin{prop}\label{prop:binary-choice}
 In a binary choice problem, let $\mu_i[v]$ denote the optimal choice probability of choosing action $a_i$, in state $i$, as a function of the reward $v$, under an LLR cost function. Assume the cost function satisfies $\beta_{12} = \beta_{21} = \beta$. Then $\mu_1[v] = \mu_2[v] = m[v]$, where
 \[
    m[v] = \frac{\ee^{\eta \left(\frac{v}{2 \beta}\right)}}{1+\ee^{\eta \left(\frac{v}{2 \beta}\right)}}
 \]
 and $\eta \colon \RR \to \RR$ is the inverse of the function $x \mapsto 2 x + e^x - e^{-x}$.
 \end{prop}
 
 As shown in Figure~\ref{fig:precision-incetives}, and as can be easily proved analytically, the optimal choice probabilities $\mu[v]$ are a sigmoidal function of the payoff $v$. \label{r4:3} The prediction is in line with other standard models that involve noise or unobserved heterogeneity, including mutual information cost. Indeed, as shown by \cite{matvejka2015rational}, under mutual information the optimal choice probabilities follow a logistic relation, where the probability of matching the state, as a function of $v$, is given by
 \[
    \frac{e^{\frac{v}{2 \lambda}}}{1+e^{\frac{v}{2 \lambda}}} \,,
 \]
 and $\lambda > 0$ is the parameter controlling the cost of information acquisition. The two functional forms are similar, with the only difference being the transformation $\eta$. The function is strictly increasing and S-shaped, onto, and satisfies $\eta(x) = \eta(-x)$ (and hence $\eta(0)=0$).\footnote{\label{minor:12}For the two models to be distinguished empirically, it is necessary to isolate the nonlinearity described by $\eta$ from other confounding effects. This can be difficult when $v$ represents dollar amounts, as the same choice probabilities obtained under LLR and risk neutrality would obtain under mutual information and a utility function $\eta$ over money. This is however not an issue if payoffs are defined using probabilistic prizes and preferences are consistent with expected utility. Allowing for more general preferences can lead to new difficulties. For example, the same choice probabilities we obtain with LLR cost can be obtained under mutual information when the decision maker has non-trivial attitudes towards how lotteries resolve over time, captured by the curvature of $\eta$. For preferences beyond expected utility, \cite{caradonna2021preference} provides a methodology for obtaining quasi-linear representations which could be used to extend our approach.}
 
 While both the LLR and the mutual information models lead to choice probabilities that are sigmoidal, the two theories lead to  different predictions on how the probabilities of errors scale with the payoff $v$. Figure~\ref{fig:precision-incetives} displays the implied probabilities with which a decision maker takes a correct choice as a function of $v$, under the two theories. To make the comparison meaningful, the parameters $\beta$ and $\lambda$ are chosen so that in both models the agent chooses incorrectly with probability $20\%$ when the payoff is $v=1$. 
 \begin{figure}[t!]
    \centering
    \includegraphics[scale=1.1]{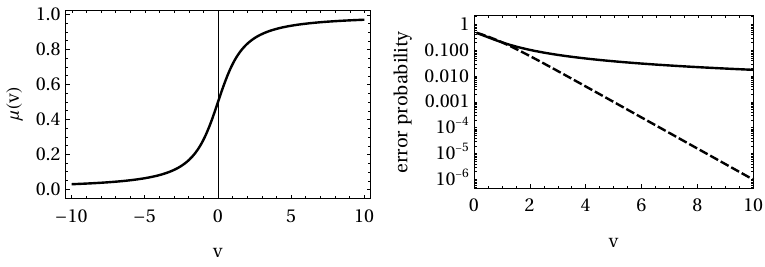}
    \caption{On the left: The optimal choice probabilities in a binary decision problem for an LLR cost function. On the right: The implied probabilities of choosing incorrectly at different levels of incentives $v$ if the agent chooses correctly with $80\%$ probability for $v=1$ for the LLR cost (solid line) and mutual information cost (dashed line) on a log-scale.}
    \label{fig:precision-incetives}
 \end{figure}
 
 As one can see in the figure, the probability of choosing correctly reacts more strongly to incentives under mutual information cost. For example, suppose that the payoff $v$ is measured in dollars. A simple calculation shows that under mutual information cost, if the decision maker chooses incorrectly with probability $20\%$ when $v=\$1$, then she must choose incorrectly with probability less than one in million if $v=\$10$. LLR costs imply that this probability is about $1/60$.\footnote{For an alternative interpretation, suppose the decision maker is paid in chance rather than money, so that the payoff $v$ denotes the probability of receiving a prize conditional on making a correct choice. Suppose that when $v$ is $0.05\%$, the decision maker makes a mistake with probability $20\%$. Then, if the probability $v$ is increased to $0.5\%$, the prediction under mutual information is that the decision maker must make a mistake with a probability that is less than one in a million. Under LLR the probability is about $1/60$.} These are starkly different predictions about behavior which can be tested experimentally.
 
 The finding is not special to this example. Under logistic choice (e.g., as in \cite{matvejka2015rational}), the probability of making a mistake decays quickly, as $v$ grows, at the exponential rate $e^{-v}$.  Under the LLR cost function the same probability decreases at the much slower rate $1/v$. This follows from Proposition~\ref{prop:binary-choice}, together with the fact that as $v$ increases, the transformation $\eta$ approximates the logarithm.

\section{Bayesian LLR Cost}
\label{sec:bayesian}
 Given a prior $q$ and an LLR cost function $C$, one can express the cost of an experiment $\mu$ in terms of the  distribution $\pi_\mu$ of the posterior belief $p \in \Delta(\Theta)$ that it induces, via
\begin{equation}\label{eq:bayesian-llr-1}
    C(\mu) = \int F(p) - F(q) \dd \pi_\mu(p) 
\end{equation}
where
\[
    F(p) = \sum_{i,j}\beta_{ij}\frac{p_{i}}{q_{i}}\log\left(\frac{p_{i}}{p_{j}}\right).
\]
%
%
%
 This follows from the definition of the LLR cost, together with Bayes' law, which states that given a prior $q$ and a signal $s$, the posterior $p$ is given by $\log \frac{p_i}{p_j} = \log\frac{q_i}{q_j}+\log\frac{\dd \mu_{i}}{\dd \mu_{j}}(s)$. This reformulation shows that the LLR cost is posterior separable \citep{caplin2013behavioral}.

 A stronger property studied in the literature is \textit{uniform} posterior separability, where the function $F$ is independent of the prior $q$. In addition to being standard, this assumption ensures, for instance, that in a dynamic environment an agent is indifferent between performing two experiments---with the choice of the second one perhaps depending on the outcome of the first---and carrying out the Blackwell equivalent one-shot experiment.

 As we now show, this assumption can be accommodated in our framework by allowing the cost $C(\mu,q)$ of an experiment $\mu$ to be a function the prior, where for each prior the cost function $C(\cdot,q)$ belongs to the LLR family, and the resulting coefficients $(\beta_{ij}(q))$ depend on the prior. While any functional relation between the prior and the coefficients is consistent with LLR cost, there is a unique choice that makes the Bayesian LLR cost function uniformly posterior-separable, as the next proposition shows. An analogous result was derived independently by \cite{bloedel2020cost}.
\begin{prop}
\label{prop:uniform-posterior}
A Bayesian LLR cost function $C$ given by 
\[
    C(\mu,q) = \sum_{i,j}\beta_{ij}(q)\dkl(\mu_i \Vert \mu_j),
\]
is uniform posterior separable if and only if there exist positive constants $(b_{ij})_{i,j\in \Theta, i \neq j}$ such that for all priors $q \in \P(\Theta)$ with full support, $\beta_{ij}(q) = b_{ij}q_i$.
\end{prop} 

 Both prior independence and constant marginal costs are reasonable assumptions when modeling common actions of information acquisitions, such as performing a measurement or drawing samples. A first implication of Proposition~\ref{prop:uniform-posterior} is that the two assumptions are incompatible with uniform posterior separability, a desirable property in a dynamic setting. This is discussed in depth by \cite{bloedel2020cost}.
 
 \medskip
 
 Proposition~\ref{prop:uniform-posterior} also shows that uniform posterior separability is possible if the coefficients $\beta_{ij}$ are allowed to change with the prior. Letting
\[
  F(p) = \sum_{i,j}b_{ij} p_{i}\log\left(\frac{p_{i}}{p_{j}}\right),
\] 
and substituting this into \eqref{eq:bayesian-llr-1}, we see that  Bayesian LLR cost of an experiment can be represented as the expected change of $F$
from the prior $q$ to the posterior $p$ induced by the signal, for a fixed choice of $(b_{ij})$. That is, the cost of the experiment equals
\begin{align}
\label{eq:bayesian-llr-2}
 C(\mu,q) = \int\left[F(p)-F(q)\right]\,\dd\pi(p),
\end{align}
and in particular it is uniformly  posterior-separable. For a given prior, this cost is the LLR cost with $\beta_{ij}=b_{ij}q_i$, so that, in terms of the distributions $(\mu_i)$, this cost is
\begin{align}
\label{eq:bayesian-llr-3}
C(\mu,q)= \sum_{i,j}b_{ij}q_i\dkl(\mu_i\Vert\mu_j).
\end{align}
%

Proposition~\ref{prop:uniform-posterior} implies that the only uniformly posterior separable LLR cost potentially assigns different cost to the same experiment at different prior beliefs. Nevertheless, some experiments may be assigned a cost that does not depend on the prior. In \S\ref{app:uniform} of the Appendix we explore which experiments have prior dependent cost and which do not.

\section{Verification and Falsification}
\label{sec:verification}

 All the specifications of the LLR cost we have discussed in the previous sections have the property that the coefficients are symmetric across states, so that $\beta_{ij}=\beta_{ji}$. In this section we explain why some information costs are best modeled by specifications that break this symmetry.

 It is well understood that verification and falsification are fundamentally different forms of empirical research. This can be seen most clearly through Karl Popper's famous example of the statement ``all swans are white.'' Regardless of how many white swans are observed, no amount of evidence can imply that the next one will be white. However, observing a single black swan is enough to prove the statement false. Popper's argument highlights a crucial asymmetry between verification and falsification. A given experiment, such as the observation of swans, can make it feasible to reject a hypothesis, yet have no power to prove that the same hypothesis is true.

 This principle extends from science to everyday life. In a legal case, the type of evidence necessary to prove that a person is guilty can be quite different from the type of evidence necessary to demonstrate that a person is innocent. In a similar way, corroborating the claim ``Ann has a sibling'' might require empirical evidence (such as the outcome of a DNA test) that is distinct from the sort of evidence necessary to prove that she has no siblings. 
 
 In this section we show that the asymmetry between verification and falsification can be captured by the LLR cost. As an example, we consider a state space $\Theta = \{a,e\}$ that consists of two hypotheses. For simplicity, let $\{a\}$ corresponds to the hypothesis ``all swans are white'' and $\{e\}$ the complementary hypothesis ``there exists a nonwhite swan.'' Imagine a decision maker who attaches equal probability to the each state, and consider the experiments described in Table 1:%
 \footnote{\cite{popper1959logic} intended verification and falsifications as deterministic procedures, which exclude even small probabilities of error. In our informal discussion we do not distinguish between events that are deemed extremely unlikely (such as thinking of having observed a black swan in world where all swans are white) and events that have zero probability. In their work on falsifiability, \cite{olszewski2011falsifiability} ascribe to \cite{cournot1843exposition} the idea that unlikely events must be treated as impossible.}

\begin{itemize}

\item In experiment I, regardless of the state, an uninformative signal realization $s_1$ occurs with probability greater than $1-\varepsilon$, where $\varepsilon$ is positive and small. If a nonwhite swan exists, then one is observed with probability $\varepsilon$. Formally, this corresponds to observing the signal realization $s_2$. If all swans are white, then signal $s_1$ is observed, up to a minuscule probability of error $\varepsilon^2$. Hence, conditional on observing $s_2$, the decision maker's belief in state $a$ approaches zero, while conditional on observing $s_1$ the decision maker's belief remains close to the prior. So, the experiment can reject the hypothesis that the state is $a$, but cannot verify it. We set the probability of observing a nonwhite swan in state $a$ equal to $\varepsilon^2$ rather than zero, to ensure that log-likelihood ratios are finite for each observation, and hence that the experiment has finite cost.

\item In experiment II the roles of the two states are reversed: if all swans are white, then this fact is revealed to the decision maker with probability $\varepsilon$. If there is a non-white swan, then the uninformative signal $s_1$ is observed (up to the small probability of error $\varepsilon^2$). Conditional on observing $s_2$, the decision maker's belief in state $a$ approaches one, and conditional on observing $s_1$ the decision maker's belief is essentially unchanged. Thus, the experiment can verify the hypothesis that the state is $a$, but cannot reject it.

\end{itemize}

\renewcommand{\arraystretch}{1.5}

\begin{table}[]
\centering
\subfloat[][Experiment I]
{
\begin{tabular}{cc|c|c|}
      & \multicolumn{1}{c}{} & \multicolumn{2}{c}{}\\
      & \multicolumn{1}{c}{} & \multicolumn{1}{c}{$s_1$}  & \multicolumn{1}{c}{$s_2$} \\\cline{3-4}
      \multirow{2}*{}  & $a$ & $1 - \varepsilon^2$ & $\varepsilon^2$ \\\cline{3-4}
      & $e$ & $1 - \varepsilon$ & $\varepsilon$ \\\cline{3-4}
\end{tabular}
}
\qquad
\subfloat[][Experiment II]
{
\begin{tabular}{cc|c|c|}
      & \multicolumn{1}{c}{} & \multicolumn{2}{c}{}\\
      & \multicolumn{1}{c}{} & \multicolumn{1}{c}{$s_1$}  & \multicolumn{1}{c}{$s_2$} \\\cline{3-4}
      \multirow{2}*{}  & $a$ & $1 - \varepsilon$ & $\varepsilon$ \\\cline{3-4}
      & $e$ & $1 - \varepsilon^2$ & $\varepsilon^2$ \\\cline{3-4}
\end{tabular}
}
\qquad
\caption{The set of states is $\Theta = \{a,e\}$. In both experiments $S = \{s_1,s_2\}$. Under experiment I, observing the signal realization $s_2$ rejects the hypothesis that the state is $a$ (up to a small probability of error $\varepsilon^2$). Under experiment II, observing $s_2$ verifies the same hypothesis.}
\end{table}

As shown by the example, permuting the state-dependent distributions of an experiment may affect its power to verify or falsify an hypothesis. However, permuting the role of the states may, in reality, correspond to a completely different type of empirical investigation. For instance, experiment I can be easily implemented in practice: as an extreme example, the decision maker may look up in the sky. There is a small chance a nonwhite swan will be observed; if not, the decision maker's belief will not change by much. It is not obvious exactly what tests or samples would be necessary to implement experiment II, which must be able to reveal that all swans are white, let alone to conclude that the two experiments should be equally costly.

We conclude that in order for a model of information acquisition to capture the difference between verification and falsification, the cost of an experiment should not necessarily be invariant with respect to a permutation of the states. In our model, this can be captured by assuming that the coefficients $(\beta_{ij})$ are non-symmetric, i.e. that $\beta_{ij}$ and $\beta_{ji}$ are are not necessarily equal. For instance, the cost of experiments I and II in Table 1 will differ whenever the coefficients of the LLR cost satisfy $\beta_{ae} \neq \beta_{ea}$. For example, set $\beta_{ae}=\kappa$ and $\beta_{ea}=0$, and  consider small $\varepsilon$. Then, to first order in $\varepsilon$, the cost of experiment I is $\kappa \varepsilon$, while the cost of experiment II is a factor of $\log(1/\varepsilon)$ higher. Hence the ratio between the costs of these experiments is arbitrarily high for small $\varepsilon$.

\section{Related Literature}

The question of how to quantify the amount of information provided by an experiment is the subject of a long-standing and interdisciplinary literature. \cite*{kullback1951information} introduced the notion of Kullback-Leibler divergence as a measure of distance between statistical populations. \cite*{kelly1956new}, \cite*{lindley1956measure}, \cite*{marschak1959remarks} and \cite*{arrow1971value} apply mutual information to the problem of ordering information structures.

More recently, \cite*{hansen2001robust} and \cite*{strzalecki2011axiomatic} adopted KL-divergence  as a tool to model robust decision criteria under uncertainty. \cite*{cabrales2013entropy} derive Shannon entropy as an index of informativeness for experiments in the context of portfolio choice problems \citep*[see also][]{cabrales2017normalized}. \cite*{frankel2018quantifying} put forward an axiomatic framework for quantifying the value and the amount of information in an experiment.

\paragraph{Rational Inattention.} As discussed in the introduction, our work is also motivated by the recent literature on rational inattention. A complete survey of this area is beyond the scope of this paper; we instead refer the interested reader to \cite*{caplin2016measuring} and \cite*{mackowiak2018rational} for perspectives on this growing literature.

\paragraph{Decision Theory.} Our axiomatic approach differs both in terms of motivation and techniques from other results in the literature. \cite*{caplin2015revealed} study the revealed preference implications of rational inattention models, taking as a primitive state-dependent random choice data. Within the same framework, \cite*{caplindeanleahy2017} characterize mutual information cost, \cite*{chambers2017nonseparable} study non-separable models of costly information acquisition, and \cite*{denti2022posterior} provides a revealed preference of posterior separability. Decision theoretic foundations for models of information acquisition have been studied by \cite*{de2014axiomatic}, \cite*{de2017rationally}, and \cite*{ellis2018foundations}. \cite*{mensch2018cardinal} provides an axiomatic characterization of posterior-separable cost functions.

%
%

\paragraph{The Wald Model of Sequential Sampling.}\label{par:wald} The notion of constant marginal costs over independent experiments goes back to Wald's \citeyearpar{wald1945sequential} classic sequential sampling model; our axioms extend some of Wald's ideas to a model of flexible information acquisition. In its most general form, Wald's model considers a decision maker who acquires information by collecting multiple independent copies of a fixed experiment, and incurs a cost equal to the number of repetitions. In this model, every stopping strategy corresponds to an experiment, and so every such model defines a cost over some family of experiments. It is easy to see that such a cost satisfies our axioms.

\cite{morrisstrack2018} consider a continuous-time version where the decision maker observes a one-dimensional diffusion process whose drift depends on the state, and incurs a cost proportional to the expected time spent observing. This cost is again easily seen to satisfy our axioms, and indeed, for the experiments that can be generated using this sampling process, they show that the expected cost of a given distribution over posteriors is of the form obtained in Proposition \ref{prop:onedimens}. 
 One may view the result in \citeauthor{morrisstrack2018} as complementary evidence that the cost function obtained in Proposition \ref{prop:onedimens} is a natural choice for one-dimensional information acquisition problems.


 

\paragraph{Dynamic Information Acquisition Models.}
\cite*{hebert2019rational, hebert2020neighborhood}, \cite{zhong2017lininfo,zhong2019dyninfo}, \cite*{morrisstrack2018}, and \cite*{bloedel2020cost} relate cost functions over experiments and sequential models of costly information acquisition. In these papers, the cost $C(\mu)$ is the minimum expected cost of generating the experiment $\mu$ by means of a dynamic sequential sampling strategy. 

\cite*{hebert2020neighborhood} propose and characterize a family of ``neighborhood-based'' cost functions that generalize mutual information, and allow for the cost of learning about states to be affected by their distance. In a perception task, these costs are flexible enough to accommodate optimal response probabilities that are S-shaped, similarly to our analysis in \S\ref{sec:stochastic}. The LLR cost does not generalize mutual information, but has a structure similar to a neighborhood-based cost where the neighboring structure consists of all pairs of states.

\cite{zhong2017lininfo} and \cite*{bloedel2020cost} provide general conditions for a cost function over experiments to be induced by some dynamic model of information acquisition. \cite{zhong2019dyninfo} studies a dynamic model of non-parametric information acquisition, where a decision maker can choose any dynamic signal process as an information source, and pays a flow cost that is a function of the informativeness of the process. A key assumption is discounting of delayed payoffs. The paper shows that the optimal strategy corresponds to a Poisson experiment.

\paragraph{Information Theory.} This paper is also related to the axiomatic literature in information theory characterizing different notions of entropy and information measures. \cite*{ebanks1998characterizations} and \cite*{csiszar2008axiomatic} survey and summarize the literature in the field. In the special case where $|\Theta| = 2$ and the coefficients $(\beta_{ij})$ are set to $1$, the function \eqref{eq:cost} is also known as \emph{J-divergence}. \cite*{kannappan1988axiomatic}  provide an axiomatization of J-divergence, under axioms very different from the ones in this paper. A more general representation appears in \cite*{zanardo2017measure}.

\cite*{ebanks1998characterizations} characterize functions over tuples of measures with finite support. They show that a condition equivalent to our additivity axiom leads to a functional form similar to \eqref{eq:cost}. Their analysis is however quite different from ours: their starting point is an assumption which, in the notation of this paper, states the existence of a map $F \colon \RR^\Theta \to \RR$ such that the cost of an experiment $(S,(\mu_i))$ with finite support takes the form  $C(\mu) = \sum_{s \in S} F((\mu_i(s))_{i \in \Theta})$. This assumption of additive separability does not seem to have an obvious economic interpretation, nor to be related to our motivation of capturing constant marginal costs in information production.

\paragraph{Probability Theory.} The results in \cite*{mattner1999dedicated, mattner2004cumulants} have, perhaps, the closest connection with this paper. Mattner studies functionals over the space probability measures over $\RR$ that are additive with respect to convolution. As we explain in the next section, additivity with respect to convolution is a property that is closely related to Axiom \ref{axm:additivity}. We draw inspiration from \cite{mattner1999dedicated} in applying the study of cumulants to the proof of Theorem \ref{thm:repr-1}. However, the difference in domain makes the techniques in \cite*{mattner1999dedicated, mattner2004cumulants} not applicable to this paper.

\section{Proof Sketch}

In this section we informally describe some of the ideas involved in the proof of Theorem \ref{thm:repr-1}. We consider the binary case where $\Theta=\left\{ 0,1\right\}$ and so there is only one relevant log-likelihood ratio $\ell=\ell_{10}$. The proof of the general case is more involved, but conceptually similar. 

\smallskip

\noindent \textbf{Step 1.} Let $C$ satisfy Axioms 1-4. Conditional on each state $i$, an experiment $\mu$ induces a distribution $\sigma_i$ for $\ell$. Two experiments that induce the same pair of distributions $(\sigma_0,\sigma_1)$ are equivalent in the Blackwell order. Thus, by Axiom \ref{axm:info-content}, $C$ can be identified with a map $c(\sigma_0,\sigma_1)$ defined over all pairs of distributions induced by some experiment $\mu$.

\smallskip
\noindent \textbf{Step 2.} Axioms \ref{axm:additivity} and \ref{axm:affinity-1} translate into the following properties of $c$. The product $\mu \otimes \nu$ of two experiments induces, conditional on $i$, a distribution for $\ell$ that is  the convolution of the distributions induced by the two experiments.\footnote{Recall that given two distributions $\sigma$ and $\nu$ over $\RR$, their convolution is the distribution of the random variable $X + Y$, where $X$ is a random variable distributed according to $\sigma$, $Y$ according to $\nu$, and the two random variables are independent. When two experiments are independent (in the sense described in \S\ref{sec:model}), their log-likelihood ratios are independent random variables conditional on the state. The crucial observation is that the log-likelihood ratio of the product experiment is the sum of the individual log-likelihood ratios, and thus its distribution conditional on the state is the convolution of theirs.} Axiom \ref{axm:additivity} is equivalent to $c$ being additive with respect to convolution, i.e. $$c(\sigma_0 * \tau_0,\sigma_1 * \tau_1) = c(\sigma_0,\sigma_1) + c(\tau_0,\tau_1)$$
Axiom \ref{axm:affinity-1} is equivalent to $c$ satisfying for all $\alpha\in [0,1]$,
\[
c(\alpha \sigma_0 + (1-\alpha)\delta_0,\alpha \sigma_1 + (1-\alpha)\delta_0) = \alpha c(\sigma_0,\sigma_1)
\]
where $\delta_0$ is the degenerate measure at $0$. Axiom \ref{axm:continuity} translates into continuity of $c$ with respect to total variation and the first $N$ moments of $\sigma_0$ and $\sigma_1$.

\smallskip
\noindent \textbf{Step 3.} As is well known, many properties of a probability distribution can be analyzed by studying its moments. We apply this idea to the study of experiments, and show that under our axioms the cost $c(\sigma_0,\sigma_1)$ is a function of the first $N$ moments of the two measures, for some (arbitrarily large) $N$. Given an experiment $\mu$, we consider the experiment
\[
    \frac{1}{n} \cdot (\mu \otimes \cdots \otimes \mu)
\]
in which with probability $(n-1)/n$ no information is produced, and with the remaining probability the experiment $\mu$ is carried out $n$ times. By Axioms \ref{axm:additivity}
and \ref{axm:affinity-1}, the cost of this experiment is equal to the cost of $\mu$.\footnote{For $n$ large, this experiment has a very simple structure: With high probability it is uninformative, and with probability $1/n$ is highly revealing about the states.} We show that these properties, together with the continuity axiom, imply that the cost of an experiment is a function $G$ of the moments of $(\sigma_0,\sigma_1)$:
\begin{equation}\label{eq:sketch1}
    c(\sigma_0,\sigma_1) = G\left[m_{\sigma_0}(1),\ldots,m_{\sigma_0}(N),m_{\sigma_{1}}(1),\ldots,m_{\sigma_1}(N)\right]
\end{equation}
where $m_{\sigma_i}(n)$ is the $n$-th moment of $\sigma_i$. Each $m_{\sigma_i}(n)$ is affine in $\sigma_i$, hence Step 2 implies that $G$ is affine with respect to mixtures with the zero vector.

\smallskip
\noindent \textbf{Step 4.} It will be useful to analyze a distribution not only through its moments but also through its cumulants. The \emph{$n$}-th\emph{ cumulant }$\kappa_{\sigma}(n)$ of a probability measure $\sigma$ is the $n$-th derivative at $0$ of the logarithm of its characteristic function. By a combinatorial characterization due to \cite{leonov1959method}, $\kappa_{\sigma}(n)$ is a polynomial function of the first $n$ moments $m_{\sigma}(1),\ldots,m_{\sigma}(n)$. For example, the first cumulant is the expectation $\kappa_{\sigma}(1)=m_{\sigma}(1)$, the second is the variance, and the third is $\kappa_{\sigma}(3)=m_{\sigma}(3)-2m_{\sigma}(2)m_{\sigma}(1)+2m_{\sigma}(1)^{3}$. Step 3 and the result by \cite{leonov1959method} imply that the cost of an experiment is a function $H$ of the cumulants of $(\sigma_0,\sigma_1)$:
\begin{equation}\label{eq:sketch2}
c(\sigma_0,\sigma_1) = H\left[\kappa_{\sigma_{0}}(1),\ldots,\kappa_{\sigma_{0}}(N),\kappa_{\sigma_{1}}(1),\ldots,\kappa_{\sigma_{1}}(N)\right]
\end{equation}
where $\kappa_{\sigma_i}(n)$ is the $n$-th cumulant of $\sigma_i$.

\smallskip
\noindent \textbf{Step 5.} Cumulants satisfy a crucial property: the cumulant of a sum of two independent random variables is the sum of their cumulants. So, they are additive with respect to convolution. By Step 2, this implies that $H$ is additive. We show that $H$ is in fact a linear function. This step is reminiscent of the classic Cauchy equation problem. That is, understanding under what conditions a function $\phi\text{\ensuremath{\colon\mathbb{R}}}\to\mathbb{R}$ that satisfies $\phi(x+y)=\phi(x)+\phi(y)$ must be linear.
In Theorem \ref{thm:cauchy} we show, very generally, that any additive function from a subset $\mathcal{K}\subset\mathbb{R}^{d}$ to $\mathbb{R}_{+}$ is linear, provided $\mathcal{K}$ is closed under addition and has a non-empty interior. We then proceed to show that both of these conditions are satisfied if $\mathcal{K}$ is taken to be the domain of $H$, and thus deduce that $H$ is linear.

\smallskip
\noindent \textbf{Step 6.} In the last step we study the implications of \eqref{eq:sketch1} and \eqref{eq:sketch2}. We apply the characterization by \cite{leonov1959method} and show that the affinity with respect to the origin of the map $G$, and the linearity of $H$, imply that $H$ must be a function solely of the first cumulants $\kappa_{\sigma_0}(1)$ and $\kappa_{\sigma_1}(1)$. That is, $C$ must be a weighted sum of the expectations of the log-likelihood ratio $\ell$ conditional on each state.

\section{Conclusions}

We put forward an axiomatic approach to modeling the cost of information acquisition, characterizing a family of cost functions that capture a notion of constant marginal costs in the production of information.
We propose a number of possible avenues for future research, all of which would require the solution of some non-trivial technical challenges: The first is an extension of our framework beyond the setting of a finite set of states to a continuum of states. This is natural in the context of one-dimensional problems.  Second, one could consider multidimensional problems in which $\Theta$ is a finite subset of $\mathbb{R}^d$, and study a generalization of the one-dimensional functional form we obtain in \S\ref{sec:betas}. Third, there are a number of settings which have been modeled using mutual information cost, where it may be of interest to understand the sensitivity of the conclusions to this assumption \citep[see, e.g.,][]{van2010information}.
Finally, a possible definition for convex cost functions over experiments is given by the supremum over a family of LLR costs. It may be interesting to understand if such costs are characterized by simple axioms.

\bibliography{lit}

\newpage


\begin{appendices}


\section{Discussion of the Continuity Axiom}
\label{sec:topology}

Our continuity axiom may seem technical,
and in a sense it is. However, there are some interesting technical
subtleties involved with its choice. Indeed, it seems that a more
natural choice of topology would be the topology of \emph{weak convergence}
of likelihood ratios. Under that topology,
two experiments would be close if they had close expected utilities
for decision problems with continuous bounded utilities. The disadvantage
of this topology is that \emph{no cost }that satisfies the rest of
the axioms is continuous in this topology. To see this, consider the
sequence of experiments in which a coin (whose bias depends on the
state) is tossed $n$ times with probability $1/n$, and otherwise
is not tossed at all. Under our axioms these experiments all have
the same cost\textemdash the cost of tossing the coin once. However,
in the weak topology these experiments converge to the trivial experiment
that yields no information and therefore has zero cost. 

In fact, even the stronger \emph{total variation} topology suffers
from the same problem, which is demonstrated using the same sequence
of experiments. Therefore, one must consider a \emph{finer }topology
(which makes for a weaker continuity assumption), which we do by
also requiring the first $N$ moments to converge. Note that increasing
$N$ makes for a finer topology and therefore a weaker continuity
assumption, and that our results hold for all $N>0$. An even stronger
topology (which requires the convergence of all moments) is used by
\citet*{mattner1999dedicated,mattner2004cumulants} to characterize all continuous additive
linear functionals on the space of all random variables on $\mathbb{R}$.

Nevertheless, the continuity axiom is technical. As we show in Theorem~\ref{thm:repr-bayes} it is not required when there are only two states, and we conjecture that it is not required in general.

\section{Preliminaries}
\label{sec:pre}
To simplify the notation, throughout the appendix we set $\Theta = \{0,1,\ldots,n\}$.

\subsection{Properties of the Kullback-Leibler Divergence}\label{sec:dkl}
In this section we summarize some well known properties of the Kullback-Leibler divergence, and derive from them straightforward properties of the LLR cost.

Given a measurable space $(X,\Sigma)$ we denote by $\P(X,\Sigma)$ the space of probability measures on $(X,\Sigma)$. If $X=\mathbb{R}^{d}$ for some $d\in\mathbb{N}$ then $\Sigma$ is implicitly assumed to be the corresponding Borel $\sigma$-algebra and we simply
write $\mathcal{P}(\mathbb{R}^{d})$.

For the next result, given two measurable spaces $(\Omega,\Sigma)$ and $(\Omega',\Sigma')$, a measurable map $F \colon \Omega \to \Omega'$, and a measure $\eta \in \P(\Omega,\Sigma)$, we define the {\em push-forward} measure $F_*\eta \in \P(\Omega',\Sigma')$ by  $[F_*\eta](A) = \eta(F^{-1}(A))$ for all $A \in \Sigma'$.

\begin{prop}
\label{clm:dkl}
Let $\nu_1,\nu_2,\eta_1,\eta_2$ be measures in $\P(\Omega,\Sigma)$, and let $\mu_1,\mu_2$ be probability measures in $\P(\Omega',\Sigma')$. Assume that $\dkl(\nu_1\Vert\nu_2)$, $\dkl(\eta_1\Vert\eta_2)$ and $\dkl(\mu_1\Vert\mu_2)$ are all finite. Let $F \colon \Omega \to \Omega'$ be measurable. Then:
\begin{enumerate}
    \item $\dkl(\nu_1\Vert\nu_2) \geq 0$ with equality if and only if $\nu_1=\nu_2$.
    \item $\dkl(\nu_1\times\mu_1\Vert\nu_2\times\mu_2) =\dkl(\nu_1\Vert\nu_2)+\dkl(\mu_1\Vert\mu_2)$.
    \item For all $\alpha \in (0,1)$, 
    \[
    \dkl(\alpha\nu_1+(1-\alpha)\eta_1\Vert\alpha\nu_2+(1-\alpha)\eta_2) \leq \alpha \dkl(\nu_1\Vert\nu_2)+(1-\alpha)\dkl(\eta_1\Vert\eta_2).
    \]
    \item $\dkl(F_*\nu_1\Vert F_*\mu_1) \leq \dkl(\nu_1\Vert \mu_1)$.
\end{enumerate}
\end{prop}

It is well known that KL-divergence satisfies the first three properties in the statement of the proposition. We refer the reader to \cite[][Proposition 2.4]{austinentropy} for a proof of the last property.

\begin{lem}\label{lem:llr}
Two experiments $\mu = (S,(\mu_i))$ and $\nu = (T,(\nu_i))$ that satisfy $\bar{\mu}_i = \bar{\nu}_i$ for every $i \in \Theta$ are equivalent in the Blackwell order. 
\end{lem}
\begin{proof}
The result is standard, but we include a proof for completeness. Suppose $\bar{\mu}_i = \bar{\nu}_i$ for every $i \in \Theta$. Given the experiment $\mu$ and a uniform prior on $\Theta$, the posterior probability of state $i$ conditional on $s$ is given almost surely by
\begin{equation}\label{eq:bayes}
    p_i(s) = \frac{\dd\mu_i}{\dd \sum_{j\in\Theta} \mu_j}(s) = \frac{1}{\sum_{j\in\Theta}\frac{\dd\mu_j}{\dd\mu_i}(s)} =  \frac{1}{\sum_{j\in\Theta}\ee^{\ell_{ji}(s)}}
\end{equation}
and the corresponding expression applies to experiment $\nu$. By assumption, conditional on each state the two experiments induce the same distribution of log-likelihood ratios $(\ell_{ij})$. Hence, by \eqref{eq:bayes} they must induce the same distribution over posteriors, hence be equivalent in the Blackwell order.
\end{proof}

A consequence of Proposition \ref{clm:dkl} is that the LLR cost is monotone with respect to the Blackwell order:
\begin{proof}[Proof of Proposition \ref{prop:llr-monotnoe}]
Let $C$ be an LLR cost. It is immediate that if $\bar{\mu}_i = \bar{\nu}_i$ for every $i$ then $C(\mu) = C(\nu)$. We can assume without loss of generality that $S = T = \mathcal{P}(\Theta)$, endowed with the Borel $\sigma$-algebra. This follows from the fact that we can define a new experiment $\rho = (\mathcal{P}(\Theta),(\rho_i))$ such that $\bar{\mu}_i = \bar{\rho}_i$ for every $i$ (see, e.g. \cite*{le1996comparison}), and apply the same result to $\nu$ . By Blackwell's Theorem there exists a probability space $\left(R,\lambda\right)$ and a ``garbling''
map $G\colon S\times R\to T$ such that for each $i\in\Theta$ it holds that $\nu_i = G_*(\mu_i \times \lambda)$. Hence, by the first, second and fourth statements in Proposition~\ref{clm:dkl},
\begin{align*}
\dkl(\nu_{i}\Vert\nu_{j})
&=\dkl(G_*(\mu_i \times \lambda)\Vert G_*(\mu_j \times \lambda))\\
&\leq \dkl(\mu_i \times \lambda\Vert\mu_j \times \lambda)\\
&= \dkl(\mu_i\Vert\mu_j)+\dkl(\lambda\Vert\lambda)\\
&=\dkl(\mu_i \Vert\mu_j ).
\end{align*}
Therefore, by Theorem~\ref{thm:repr-1}, we have
\[
C(\nu)=\sum_{i,j}\beta_{ij}\dkl(\nu_{i}\Vert\nu_{j})\leq\sum_{i,j}\beta_{ij}\dkl(\mu_{i}\Vert\mu_{j})=C\left(\mu\right).\qedhere
\]
\end{proof}
We note that a similar argument shows that if all the coefficients $\beta_{ij}$ are positive then $C(\mu) > C(\nu)$ whenever $\mu$ Blackwell dominates $\nu$ but $\nu$ does not dominate $\mu$.

An additional direct consequence of Proposition~\ref{clm:dkl} is that the LLR cost is convex:
\begin{prop}
\label{prop:llr-convex}
Let $\mu=(S,(\mu_i))$ and $\nu=(S,(\nu_i))$ be experiments in $\mathcal{E}$. Given $\alpha \in (0,1)$, define the experiment $\eta = (S, (\nu_i))$ as $\eta_i = \alpha\nu_i+(1-\alpha)\mu_i$ for each $i$. Then any LLR cost $C$ satisfies
\[
  C(\eta) \leq \alpha C(\nu)+(1-\alpha)C(\mu).
\]
\end{prop}
The result follows immediately from the third statement in Proposition~\ref{clm:dkl}.  We now study the set
\[
    \mathcal{D} = \{(\dkl(\mu_i \Vert \mu_j))_{i\neq j} : \mu \in \mathcal{E}\} \subseteq \RR_{+}^{(n+1)n}
\]
of all possible pairs of expected log-likelihood ratios induced by some experiment $\mu$. The next result shows that $\mathcal{D}$ contains the strictly positive orthant.

\begin{lem}\label{lem:domain}
$\RR_{++}^{(n+1)n} \subseteq \mathcal{D}$
\end{lem}
\begin{proof}
  The set $\mathcal{D}$ is convex. To see this, let $\mu = (S,(\mu_i))$ and $\nu = (T,(\nu_i))$ be two experiments. Without loss of generality, we can suppose that $S = T$, and $S = S_1 \cup S_2$, where $S_1,S_2$ are disjoint, and $\mu_i(S_1) = \nu_i(S_2) = 1$ for every $i$.
  
  Fix $\alpha \in (0,1)$ and define the new experiment $\tau = (S, (\tau_i))$ where $\tau_i = \alpha \mu_i + (1-\alpha)\nu_i$ for every $i$. It can be verified that $\tau_i$-almost surely, $\frac{\dd\tau_i}{\dd\tau_j}$ satisfies $\frac{\dd\tau_i}{\dd\tau_j}(s) = \frac{\dd\mu_i}{\dd\mu_j}(s)$ if $s\in S_1$ and $\frac{\dd\tau_i}{\dd\tau_j}(s) = \frac{\dd\nu_i}{\dd\nu_j}(s)$ if $s \in S_2$. It then follows that
  \[
    \dkl(\tau_i \Vert \tau_j) = \alpha \dkl(\mu_i \Vert \mu_j) + (1-\alpha) \dkl(\nu_i \Vert \nu_j).
  \]
  Hence $\mathcal{D}$ is convex. We now show $\mathcal{D}$ is a convex cone. First notice that the zero vector belongs to $\mathcal{D}$, since it corresponds to the totally uninformative experiment. In addition (see \S\ref{sec:dkl}),
  \[
    \dkl((\mu \otimes \mu)_i \Vert (\mu \otimes \mu)_j) = \dkl(\mu_i \times \mu_i \Vert \mu_j \times \mu_j) = 2 \dkl(\mu_i \Vert \mu_j) 
  \]
  Hence $\mathcal{D}$ is closed under addition. Because $\mathcal{D}$ is also convex and contains the zero vector, it is a convex cone.
  
  Suppose, by way of contradiction, that the inclusion $\RR_{++}^{(n+1)n} \subseteq \mathcal{D}$ does not hold. This implies we can find a vector $z \in \RR_{+}^{(n+1)n}$ that does not belong to the closure of $\mathcal{D}$. Therefore, there exists a nonzero vector $w \in \RR^{(n+1)n}$ and $t \in \RR$ such that $w\cdot z > t \geq w\cdot y$ for all $y \in \mathcal{D}$. Because $\mathcal{D}$ is a cone, then $t \geq 0$ and $0 \geq w\cdot y$ for all $y \in \mathcal{D}$. Hence, there must exist a coordinate $i_oj_o$ such that $w_{i_oj_o} > 0$. We now show this leads to a contradiction.

   Consider the following three cumulative
distribution functions on $[2,\infty)$:
\begin{align*}
F_{1}(x) & =1-\frac{2}{x}\\
F_{2}(x) & =1-\frac{\log^{2}2}{\log^{2}x}\\
F_{3}(x) & =1-\frac{\log2}{\log x},
\end{align*}
and denote by $\pi_{1},\pi_{2},\pi_{3}$ the corresponding measures.
A simple calculation shows that $\dkl(\pi_{3}\Vert\pi_{1})=\infty$, whereas
$\dkl(\pi_{a}\Vert\pi_{b})<\infty$ for any other choice of $a,b\in\left\{ 1,2,3\right\} $. 

Let $\pi_{a}^{\varepsilon}=\left(1-\varepsilon\right)\delta_{2}+\varepsilon\pi_{a}$
for every $a\in\left\{ 1,2,3\right\} $, where $\delta_{2}$ is the
point mass at $2$. Then still $\dkl(\pi_{3}^{\varepsilon}\Vert\pi_{1}^{\varepsilon})=\infty$,
but, for any other choice of $a$ and $b$ in $\{1,2,3\}$, the divergence
$D(\pi_{a}^{\varepsilon}\Vert\pi_{b}^{\varepsilon})$ vanishes as $\varepsilon$
goes to zero. Let $\pi_{a}^{\varepsilon,M}$ be the measure $\pi_{a}^{\varepsilon}$
conditioned on the interval $[2,M]$. Then $\dkl(\pi_{a}^{\varepsilon,M}\Vert\pi_{b}^{\varepsilon,M})$
tends to $\dkl(\pi_{a}^{\varepsilon}\Vert\pi_{b}^{\varepsilon})$ as $M$
tends to infinity, for any $a,b$. It follows that for every $N\in\mathbb{N}$
there exist $\varepsilon$ small enough and $M$ large enough such
that $\dkl(\pi_{3}^{\varepsilon,M}\Vert\pi_{1}^{\varepsilon,M})>N$ and,
for any other choice of $a,b$, $\dkl(\pi_{a}^{\varepsilon,M}\Vert\pi_{b}^{\varepsilon,M})<1/N$.

Consider the experiment $\mu=\left(\mathbb{R},\left(\mu_{i}\right)\right)$ where $\mu_{i_{0}}=\pi_{3}^{\varepsilon,M}$, $\mu_{j_{0}}=\pi_{1}^{\varepsilon,M}$ and $\mu_{k}=\pi_{2}^{\varepsilon,M}$ for all $k\not\in\left\{ i_{0},j_{0}\right\}$ and with $\varepsilon$ and $M$ so that the inequalities above hold for $N$ large enough. Then $\mu \in \mathcal{E}$  since all measures have bounded support. It satisfies $\dkl(\mu_{i_o} \Vert \mu_{j_o}) > N$ and $\dkl(\mu_i \Vert \mu_j) < 1/N$ for every other pair $ij$.

Now let $y \in \mathcal{D}$ be the vector defined by $\mu$. Then $w \cdot y > 0$ for $N$ large enough. A contradiction.
\end{proof}

\subsection{Experiments and Log-likelihood Ratios}

It will be convenient to consider, for each experiment, the distribution over log-likelihood ratios with respect to the state $i=0$ conditional on a state $j$. Given an experiment, we define $\ell_{i}=\ell_{i0}$ for every $i\in\Theta$. We say that a vector $\sigma  = (\sigma_{0},\sigma_{1},\ldots,\sigma_{n})\in\P(\mathbb{R}^{n})^{n+1}$ of measures is \emph{derived from the experiment} $(S,(\mu_i))$ if for every $i=0,1,\ldots,n$,
\[
\sigma_{i}(E) = \mu_i\left(\left\{ s:(\ell_{1}(s),\ldots,\ell_{n}(s))\in E\right\} \right)\text{ for all measurable }E\subseteq\RR^{n}.
\]
That is, $\sigma_{i}$ is the distribution of the vector $(\ell_{1},\ldots,\ell_{n})$ of log-likelihood ratios (with respect to state $0$) conditional on state $i$. There is a one-to-one relation between the vector $\sigma$ and the collection $(\bar{\mu}_i)$ of distributions defined in the main text: notice that $\ell_{ij} = \ell_{i0} - \ell_{j0}$ almost surely, hence knowing the distribution of $(\ell_{0i})_{i\in\Theta}$ is enough to recover the distribution of $(\ell_{ij})_{i,j\in \Theta}$. Nevertheless, working directly with $\sigma$ (rather than $(\bar{\mu}_i)$) will simplify the notation considerably.

We call a vector $\sigma\in\P(\mathbb{R}^{n})^{n+1}$ \textit{admissible} if it is derived from some experiment. The next result provides a straightforward characterization of admissible vectors of measures.
\begin{lem}\label{lem:admis}
A vector of measures $\sigma = (\sigma_{0},\sigma_{1},\ldots,\sigma_{n}) \in\P(\mathbb{R}^{n})^{n+1}$
is admissible if and only if the measures are mutually absolutely
continuous and, for every $i$, satisfy $\frac{\dd\sigma_{i}}{\dd\sigma_{0}}(\xi)=e^{\xi_{i}}$ for $\sigma_i$-almost every $\xi\in\RR^{n}$.
\end{lem}
\begin{proof}
If $(\sigma_{0},\sigma_{1},\ldots,\sigma_{n})$ is admissible then there exists an experiment $\mu =(S,(\mu_i))$ such that for any measurable $E\subseteq\RR^{n}$
\begin{eqnarray*}
\int_{E}e^{\xi_{i}}\,\dd\sigma_{0}(\xi) & = & \int1_{E}\left(\left(\ell_{1}(s),\ldots\ell_{n}(s)\right)\right)e^{\ell_{i}(s)}\,\dd\mu_0(s)\\
 & = & \int1_{E}\left(\left(\ell_{1}(s),\ldots\ell_{n}(s)\right)\right)\,\dd\mu_i(s)
\end{eqnarray*}
where $1_{E}$ is the indicator function of $E$. So, $\int_{E}e^{\xi_{i}}\,\dd\sigma_{0}(\xi)=\sigma_{i}(E)$
for every $E\subseteq\RR^{n}$. Hence $e^{\xi_{i}}$ is a version of  $\frac{\dd\sigma_{i}}{\dd\sigma_{0}}$.

Conversely, assume $\frac{\dd\sigma_{i}}{\dd\sigma_{0}}(\xi)=e^{\xi_{i}}$ for almost every $\xi\in\RR^{n}$. Define an experiment $(\RR^{n+1},(\mu_i))$ where $\mu_i=\sigma_{i}$
for every $i$. The experiment $(\RR^{n+1},(\mu_i))$ is such that $\ell_{i}\left(\xi\right)=\xi_{i}$ for every $i>0$. Hence, for $i>0$, $\mu_i\left(\left\{ \xi:\left(\ell_{1}(\xi),\ldots,\ell_{n}(\xi)\right)\in E\right\} \right)$ is equal to
\begin{eqnarray*}
  \int1_{E}\left(\left(\ell_{1}(\xi),\ldots\ell_{n}(\xi)\right)\right)e^{\xi_i}\, \dd \sigma_{0}(\xi)
  =  \int1_{E}(\xi)e^{\xi_i}\, \dd\sigma_{0}(\xi) = \sigma_{i}(E)
\end{eqnarray*}
and similarly $\mu_0 \left(\left\{ \xi:\left(\ell_{1}\left(\xi\right),\ldots,\ell_{n}\left(\xi\right)\right)\in E\right\} \right)=\sigma_{0}\left(E\right)$. So $(\sigma_0,\sigma_1,\ldots,\sigma_n)$ is admissible.
\end{proof}

\subsection{Properties of Cumulants}\label{sec:cumulants}

The purpose of this section is to formally describe cumulants and their relation to moments. We follow \cite*{leonov1959method} and \citet*[][p. 289]{shiryaev}. Given a vector $\xi\in\RR^{n}$ and an integral vector $\alpha\in\mathbb{N}^{n}$ we write $\xi^{\alpha}=\xi_{1}^{\alpha_{1}}\xi_{2}^{\alpha_{2}}\cdots\xi_{n}^{\alpha_{n}}$ and use the notational conventions $\alpha!=\alpha_{1}!\alpha_{2}!\cdots\alpha_{n}!$ and $|\alpha|=\alpha_{1}+\cdots+\alpha_{n}$.

Let $A=\{ 0,\ldots,N\}^{n} \backslash\{0,\ldots,0\}$, for some constant $N \in \mathbb{N}$ greater or equal than $1$. For every probability measure $\sigma_1 \in\P (\RR^{n})$ and $\xi\in\mathbb{R}^{n}$, let $\varphi_{\sigma_1}(\xi)=\int _{\RR^n}e^{i\left\langle z,\xi\right\rangle }\,\dd\sigma_1(z)$
denote the characteristic function of $\sigma_1$ evaluated at $\xi$.
We denote by $\P_{A}\subseteq\P (\RR^{n})$ the subset of measures $\sigma_1$ such that $\int_{\RR^n}\left|\xi^{\alpha}\right|\,\dd\sigma_1(\xi)<\infty$ for every $\alpha \in A$. Every $\sigma_1\in\P_{A}$ is such that in a neighborhood of $\boldsymbol{0}\in\mathbb{R}^{n}$ the cumulant
generating function $\log\varphi_{\sigma_1}$ is well defined
and the partial derivatives 
\[
\frac{\partial^{|\alpha|}}{\partial \xi_{1}^{\alpha_{1}}\partial \xi_{2}^{\alpha_{2}}\cdots\partial \xi_{n}^{\alpha_{n}}}\log\varphi_{\sigma_1}(\xi)
\]
exist and are continuous for every $\alpha\in A$.

For every $\sigma_1\in\mathcal{P}_{A}$ and $\alpha\in A$ let $\kappa_{\sigma_1}(\alpha)$
be defined as 
\[
\kappa_{\sigma_1}(\alpha)=i^{-|\alpha|}\frac{\partial^{|\alpha|}}{\partial \xi_{1}^{\alpha_{1}}\partial \xi_{2}^{\alpha_{2}}\cdots\partial \xi_{n}^{\alpha_{n}}}\log\varphi_{\sigma_1}(\boldsymbol{0})
\]
With slight abuse of terminology, we refer to $\kappa_{\sigma_1}\in\mathbb{R}^{A}$ as the \emph{vector of cumulants} of $\sigma_1$. In addition, for every $\sigma_1\in\P_{A}$ and $\alpha\in A$ we denote by $m_{\sigma_1}(\alpha)=\int_{\RR^n}\xi^{\alpha}\,\dd\sigma_1(\xi)$ the mixed moment of $\sigma_1$ of order $\alpha$ and refer to $m_{\sigma_1}\in\RR^{A}$ as the \emph{vector of moments} of $\sigma_1$.

Given two measures $\sigma_1,\sigma_2\in\P(\mathbb{R}^{n})$ we denote by $\sigma_1\ast\sigma_2\in\P(\mathbb{R}^{n})$ the
corresponding convolution. 
\begin{lem}
\label{lem:facts-cmlts}For every $\sigma_1,\sigma_2\in\P_{A}$, and
$\alpha\in A$, $\kappa_{\sigma_1\ast\sigma_2}(\alpha)=\kappa_{\sigma_1}(\alpha)+\kappa_{\sigma_2}(\alpha)$.
\end{lem}
\begin{proof}
The result follows from the well known fact that $\varphi_{\sigma_1\ast\sigma_2}(\xi)=\varphi_{\sigma_1}(\xi)\varphi_{\sigma_2}(\xi)$
for every $\xi\in\mathbb{R}^{n}$.
\end{proof}

The next result, due to \citet*{leonov1959method} \cite*[see also][p. 290]{shiryaev} establishes a one-to-one relation between the vector of moments $m_{\sigma_1}$ and vector of cumulants $\kappa_{\sigma_1}$ of a probability measure $\sigma_1\in \P_A$. Given $\alpha\in A$, let $\Lambda(\alpha)$ be the set of all ordered collections $\left(\lambda^{1},\ldots,\lambda^{q}\right)$ of non-zero vectors in $\mathbb{N}^{n}$ such that $\sum_{p=1}^{q}\lambda^{p}=\alpha$. 

\begin{thm}
\label{thm:Shiryaev}For every $\sigma_1\in\mathcal{P}_{A}$ and $\alpha\in A$,
\begin{enumerate}
\item $m_{\sigma_1}(\alpha)=\sum_{\left(\lambda^{1},\ldots,\lambda^{q}\right)\in\Lambda(\alpha)}\frac{1}{q!}\frac{\alpha!}{\lambda^{1}!\cdots\lambda^{q}!}\prod_{p=1}^{q}\kappa_{\sigma_1}(\lambda^{p})$
\item $\kappa_{\sigma_1}(\alpha)=\sum_{\left(\lambda^{1},\ldots,\lambda^{q}\right)\in\Lambda(\alpha)}\frac{\left(-1\right)^{q-1}}{q}\frac{\alpha!}{\lambda^{1}!\cdots\lambda^{q}!}\prod_{p=1}^{q}m_{\sigma_1}(\lambda^{p})$
\end{enumerate}
\end{thm}

The result yields the following implication. Let $M_A = \{m_{\sigma_1} : \sigma_1 \in \P_A\} \subseteq \RR^A$ and $K_A = \{\kappa_{\sigma_1} : \sigma_1 \in \P_A\} \subseteq \RR^A$. Statement 2 in Theorem \ref{thm:Shiryaev} shows the existence of a continuous function $h \colon M_A \to K_A$ such that $\kappa_{\sigma_1} = h(m_{\sigma_1})$ for every $\sigma_1 \in \P_A$. Moreover, statement 1 implies $h$ is one-to-one.

\subsection{Cumulants and Admissible Measures}

We denote by $\mathcal{A}$ the set of vectors of measures $\sigma = (\sigma_0,\sigma_1,\ldots,\sigma_n)$ that are admissible and such that $\sigma_{i}\in\mathcal{P}_{A}$ for every $i$. To each $\sigma\in\mathcal{A}$ we associate the vector
\[
m_{\sigma}=(m_{\sigma_0},m_{\sigma_{1}},\ldots,m_{\sigma_{n}})\in\RR^{d}
\]
of dimension $d=\left(n+1\right)|A|$. Similarly, we define
\[
\kappa_{\sigma}=\left(\kappa_{\sigma_0},\kappa_{\sigma_{1}},\ldots,\kappa_{\sigma_{n}}\right)\in\RR^{d}.
\]
In this section we study properties of the sets $\mathcal{M} = \left\{m_{\sigma}:\sigma\in\mathcal{A}\right\} $ and $\mathcal{K} = \left\{\kappa_{\sigma}:\sigma\in\mathcal{A}\right\} $.
\begin{lem}\label{lem:interior}
Let $I$ and $J$ be disjoint finite sets and let $(\phi_{k})_{k\in I\cup J}$ be a collection of real valued functions defined on $\mathbb{R}^{n}$.
Assume $\left\{ \phi_{k}:k\in I\cup J\right\} \cup\left\{ 1_{\mathbb{R}^{n}}\right\} $
are linearly independent and the unit vector $\left(1,\ldots,1\right)\in\mathbb{R}^{J}$
belongs to the interior of $\left\{ \left(\phi_{k}\left(\xi\right)\right)_{k\in J}:\xi\in\mathbb{R}^{n}\right\}$. Then
\[
C=\text{\ensuremath{\left\{  \left(\int_{\RR^n}\phi_{k}\,\dd \sigma_1\right)_{k\in I}:\sigma_1\in\P(\mathbb{R}^{n})\text{ has finite support and }\int_{\RR^n}\phi_{k}\,\dd\sigma_1=1\text{ for all }k\in J\right\} } }
\]
is a convex subset of $\mathbb{R}^{I}$ with nonempty interior.
\end{lem}
\begin{proof}
To ease the notation, let $Y=\mathbb{R}^{n}$ and denote by $\mathcal{P}_{o}$ be the set
of probability measures on $Y$ with finite support. Consider $F=\left\{ \phi_{k}:k\in I\cup J\right\} \cup\left\{ 1_{\mathbb{R}^{d}}\right\} $
as a subset of the vector space $\mathbb{R}^{Y}$, where the latter
is endowed with the topology of pointwise convergence. The topological
dual of $\mathbb{R}^{Y}$ is the vector space of signed measures on
$Y$ with finite support. Let
\[
D=\left\{ \left(\int_{\RR^n}\phi_{k}\,\dd\sigma_1\right)_{k\in I\cup J}:\sigma_1\in\mathcal{P}_{o}\right\} \subseteq\mathbb{R}^{I\cup J}.
\]

Fix $k\in I\cup J$. Since $\phi_{k}$ does not belong to the linear space $V$ generated by $F \backslash\{\phi_k\}$, then a standard application of the hyperplane separation theorem implies the existence of a signed measure
\[
\rho=\alpha\sigma_1-\beta\sigma_2
\]
where $\alpha,\beta\geq0$, $\alpha+\beta>0$ and $\sigma_1,\sigma_2\in\mathcal{P}_{o}$, 
such that $\rho$ satisfies $\int\phi_{k}\,\dd\rho>0\geq\int\phi \,\dd\rho$ for every $\phi\in V$. This implies $\int\phi \,\dd\rho=0$ for every $\phi\in V$. By taking $\phi=1_{\mathbb{R}^{n}}$, we obtain $\rho(\mathbb{R}^{n})=0$. Hence, $\alpha=\beta$. Therefore, $\int\phi_{k}\,\dd\sigma_1>\int\phi_{k}\,\dd\sigma_2$ and $\int\phi_{l}\,\dd\sigma_1=\int\phi_{l}\,\dd\sigma_2$ for every $l \neq k$. 
To summarize, we have shown that for every $k \in I \cup J$ there exist vectors $w^k,z^k \in D$ such that $w^k_k > z^k_k$ and $w^k_l = z^k_l$ for $l \neq k$. 

Now let $\mathrm{aff}(D)$ be the affine hull of $D$. As is well known, for every $d \in D$ we have the identity $\mathrm{aff}(D) = d + \mathrm{span}(D-d)$, where $\mathrm{span}(D-d)$ is the vector space generated by $D-d$. Moreover, $\mathrm{span}(D-d)$ is independent of the choice of $d \in D$ \citep*[see, for example,][Lemma 2.4.5]{borwein2010convex}.

Let $k \in I \cup J$ and let $1_k \in \RR^{I \cup J}$ be the corresponding unit vector. By taking $d = z^k$ we obtain that $w^k - z^k \in \mathrm{span}(D-z^k)$. Thus, $1_k \in \mathrm{span}(D-d)$ for every $k$. Hence $\mathrm{span}(D-d) = \RR^{I \cup J}$. Therefore $\mathrm{aff}(D) = \RR^{I \cup J}$. Since $D$ is convex, it has nonempty relative interior as a subset of $\mathrm{aff}(D)$. We conclude that $D$ has nonempty interior.

Now consider the hyperplane $$H=\{ z\in\mathbb{R}^{I\cup J}:z_{k}=1\text{ for all }k\in J\}$$
Let $D^{o}$ be the interior of $D$. It remains to show that the
hyperplane $H$ satisfies $H\cap D^{o}\neq\emptyset$. This will imply that the projection of $H \cap D$ on $\RR^I$, which equals $C$, has nonempty interior.

Let $w \in D^{o}$. By assumption, $(1,\ldots,1) \in \RR^J$ is in the interior of $\left\{ (\phi_{k}(\xi))_{k\in J}:\xi\in Y\right\}$. Hence, there exists $\alpha \in (0,1)$ small enough and $\xi \in Y$ such that $\phi_k(\xi) = \frac{1}{1-\alpha} - \frac{\alpha}{1-\alpha}w_k$ for every $k \in J$. Define $z = \alpha w + (1 - \alpha) (\phi_k(\xi))_{k \in I \cup J} \in D$. Then $z_k = 1$ for every $k \in J$. In addition, because $w \in D^o$ then $z\in D^o$ as well. Hence $z \in H \cap D^o$.
\end{proof}

\begin{lem}\label{lem:M_int}The set $\mathcal{M} = \left\{ m_{\sigma}:\sigma\in\mathcal{A}\right\} $
has nonempty interior.
\end{lem}
\begin{proof}
For every $\alpha\in A$ define the functions $(\phi_{i,\alpha})_{i \in \Theta}$ on $\RR^n$ as
\[
\phi_{0,\alpha}\left(\xi\right)=\xi^{\alpha}\text{ and }\phi_{i,\alpha}\left(\xi\right)=\xi^{\alpha}\ee^{\xi_{i}}\text{ for all \ensuremath{i>0}.}
\]
Define $\psi_0=1_{\mathbb{R}^{n}}$ and $\psi_i(\xi)=e^{\xi_{i}}$ for all $i>0$. It is immediate to verify that 
\[
\left\{ \phi_{i,\alpha}:i\in\Theta,\alpha\in A\right\} \cup \{\psi_i: i \in \Theta \}
\]
is a linearly independent set of functions. In addition,  $(1,\ldots,1) \in \RR^n$ is in the interior of $\{(\ee^{\xi_1},\ldots,\ee^{\xi_n}) : \xi \in \RR^n\}$. Lemma \ref{lem:interior} implies that the set
\[
C=\left\{ \left(\int_{\RR^n}\phi_{i,\alpha}\,\dd\sigma_0\right)_{\substack{i \in \Theta\\\alpha\in A}}:\sigma_0\in\P(\mathbb{R}^{n})\text{ has finite support and }\int_{\RR^n} e^{\xi_{i}}\,\dd\sigma_0(\xi)=1 \text{ for all }i\right\} 
\]
has nonempty interior. Given $\sigma_0$ as in the definition of $C$, construct a vector $\sigma = (\sigma_0,\sigma_1,\ldots,\sigma_n)$ where for each $i>0$ the measure $\sigma_i$ is defined so that $(\dd\sigma_i/\dd\sigma_0) (\xi) = \ee^{\xi_i}$, $\sigma_0$-almost surely. Then, Lemma \ref{lem:admis} implies $\sigma$ is admissible. Because each $\sigma_i$ has finite support then $\sigma \in \mathcal{A}$. In addition, 
\[
    m_\sigma = \left(\int_{\RR^n}\phi_{i,\alpha}\,\dd\sigma_0\right)_{\substack{i \in \Theta\\\alpha\in A}}
\]
hence $C\subseteq\mathcal{M}$. Thus, $\mathcal{M}$ has nonempty interior.
\end{proof}

\begin{thm}
\label{thm:K_int}The set $\mathcal{K} = \left\{ \kappa_{\sigma}:\sigma\in\mathcal{A}\right\} $ has nonempty interior.
\end{thm}
\begin{proof}
Let $h : M_A \to K_A$ be the function defined in the discussion following Theorem \ref{thm:Shiryaev}, mapping vectors of moments to vectors of cumulants. Define $H:\mathcal{M}\to\mathcal{K}$ as
\[
H({m}_{\sigma})= (h(m_{\sigma_0}),h(m_{\sigma_1}),\ldots,h(m_{\sigma_n}))= {\kappa}_{\sigma}
\]
for every $\sigma=(\sigma_0,\sigma_1,\ldots,\sigma_n) \in \mathcal{A}$. Since $h$ is continuous and one-to-one then so is $H$. Lemma \ref{lem:M_int} shows there exists an open set $U\subseteq\mathbb{R}^{d}$ included in $\mathcal{M}$.
Let $H_{U}$ be the restriction of $H$ on $U$. Then $H_{U}$ satisfies
all the assumptions of Brouwer's Invariance of Domain Theorem,\footnote{\cite{brouwer1911beweis}. See also Theorem 2 in \cite{tao2011}.}
which implies that $H_{U}(U)$ is an open subset of $\mathbb{R}^{d}$.
Since $H(\mathcal{M})\subseteq\mathcal{K}$, it follows that $\mathcal{K}$ has nonempty interior.
\end{proof}

\section{Automatic Continuity in the Cauchy Problem for Subsemigroups of $\mathbb{R}^{d}$.}

A \emph{subsemigroup} of $\mathbb{R}^{d}$ is a subset $\mathcal{S}\subseteq\mathbb{R}^{d}$ that is closed under addition, so that $x+y\in\mathcal{S}$ for all
$x,y\in\mathcal{S}$. We say that a map $F\colon\mathcal{S}\to\mathbb{R}_{+}$
is \emph{additive }if $F(x+y)=F(x)+F(y)$ for all $x,y,x+y\in\mathcal{S}$.
We say that $F$ is \emph{linear} if there exists $(a_{1},\ldots,a_{d})\in\RR^{d}$
such that $F(x)=F(x_{1},\ldots,x_{d})=a_{1}x_{1}+\cdots+a_{d}x_{d}$
for all $x\in\mathcal{S}$. 

We can now state the main result of this section:
\begin{thm}
\label{thm:cauchy}Let $\mathcal{S}$ be a subsemigroup of $\mathbb{R}^{d}$
with a nonempty interior. Then every additive function $F\colon\mathcal{S}\to\mathbb{R}_{+}$
is linear.
\end{thm}
Before proving the theorem we will establish a number of claims.
\begin{claim}
\label{claim:open-ball-semigroup}Let $\mathcal{S}$ be a subsemigroup
of $\mathbb{R}^{d}$ with a nonempty interior. Then there exists an
open ball $B\subset\mathbb{R}^{d}$ such that $aB\subset\mathcal{S}$
for all real $a\geq1$. 
\end{claim}
\begin{proof}
Let $B_{0}$ be an open ball contained in $\mathcal{S}$, with center $x_{0}$ and radius $r$. Given a positive integer $k$, note that $kB_{0}$ is the ball of radius $kr$ centered at $kr_{0}$, and that it is contained in $\mathcal{S}$, since $\mathcal{S}$ is a semigroup.
Choose a positive integer $M\geq4$ such that $\frac{2}{3}Mr>\Vert x_{0}\Vert$, and let $B$ be the open ball with center at $Mx_{0}$ and radius $r$ (see Figure \ref{fig:open-balls}). Fix any $a\geq1$, and write
$a=\frac{1}{M}(n+\gamma)$ for some integer $n\geq M$ and $\gamma\in[0,1)$. Then $\frac{n}{M}B$ is the ball of radius $\frac{n}{M}r$ centered at $nx_{0}$, which is contained in $nB_{0}$, since $nB_{0}$ also has center $nx_{0}$, but has a larger radius $nr$. So $\frac{n}{M}B\subset nB_{0}$.
We claim that furthermore $\frac{n+1}{M}B$ is also contained in $nB_{0}$. To see this, observe that the center of $\frac{n+1}{M}B$ is  $(n+1)x_{0}$ and its radius is $\frac{n+1}{M}r$. Hence the center of $\frac{n+1}{M}B$ is at distance $\Vert x_{0}\Vert$ from the center of $nB_{0}$, and so the furthest point in $\frac{n+1}{M}B$ is at distance $\Vert x_{0}\Vert+\frac{n+1}{M}r$ from the center of $nB_{0}$. But the radius of $nB_{0}$ is
\[
nr=\frac{2}{3}nr+\frac{1}{3}nr\geq\frac{2}{3}Mr+\frac{1}{3}nr>\Vert x_{0} \Vert +\frac{n+1}{M}r,
\]
where the first inequality follows since $n\geq M$, and the second since $\frac{2}{3}Mr>\Vert x_{0}\Vert$ and $M\geq4$. So $nB_{0}$ indeed contains both $\frac{n}{M}B$ and $\frac{n+1}{M}B$. Thus it also contains $aB$, and so $\mathcal{S}$ contains $aB$.
\end{proof}
\begin{figure}[H]
\centering{}\includegraphics[scale=0.45]{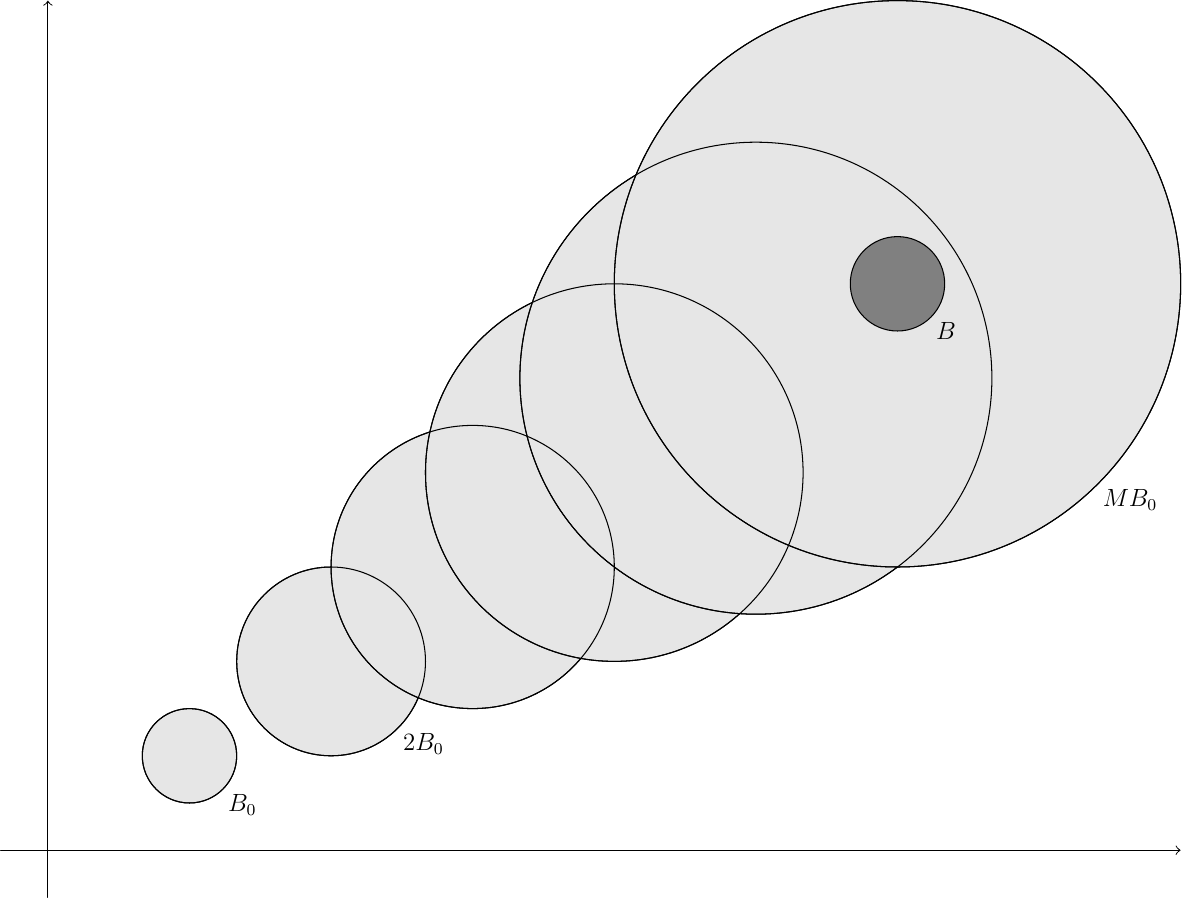}\caption{\label{fig:open-balls}Illustration of the proof of Claim \ref{claim:open-ball-semigroup}.
The dark ball $B$ is contained in the light ones, and it is apparent from this image that so is any multiple of $B$ by $a\protect\geq1$. }
\end{figure}

\begin{claim}
\label{claim:fax-afx}Let $\mathcal{S}$ be a subsemigroup of $\mathbb{R}^{d}$ with a nonempty interior. Let $F\colon\mathcal{S}\to\mathbb{R_{+}}$ be additive and satisfy $F(ay)=aF(y)$ for every $y\in\mathcal{S}$ and $a\in\mathbb{R}_{+}$ such that $ay\in\mathcal{S}$. Then $F$
is linear.
\end{claim}
\begin{proof}
If $\mathcal{S}$ does not include zero, then without loss of generality
we add zero to it and set $F(0)=0$. Let $B$ be an open ball such
that $aB\subset\mathcal{S}$ for all $a\geq1$; the existence of such a ball is guaranteed by Claim \ref{claim:open-ball-semigroup}. Choose a basis $\{ b^{1},\ldots,b^{d}\} $ of $\mathbb{R}^{d}$ that is a subset of $B$, and let $x=\beta_{1}b^{1}+\cdots+\beta_{d}b^{d}$ be an arbitrary element of $\mathcal{S}$. Let $b=\max\left\{ 1/|\beta_{i}|\,:\,\beta_{i}\neq0\right\} $,
and let $a=\max\left\{ 1,b\right\} $. Then
\[
F(ax)=F(a\beta_{1}b^{1}+\cdots+a\beta_{d}b^{d}).
\]
Assume without loss of generality that for some $0\leq k\leq d$ it
holds that the first $k$ coefficients $\beta_{i}$ are non-negative, and the rest are negative. Then for $i\leq k$ it holds that $a\beta_{i}b^{i}\in\mathcal{S}$ and for $i>k$ it holds that $-a\beta_{i}b^{i}\in\mathcal{S}$; this follows from the defining property of the ball $B$, since each $b^{i}$ is in $B$, and since $|a\beta_{i}|\geq1$. Hence we can add $F(-a\beta_{k+1}b^{k+1}-\cdots-a\beta_{d}b^{d})$ to both sides of the above displayed equation, and then by additivity,
\begin{align*}
 & F(ax)+F(-a\beta_{k+1}b^{k+1}-\cdots-a\beta_{d}b^{d})\\
 & =F(a\beta_{1}b^{1}+\cdots+a\beta_{d}b^{d})+F(-a\beta_{k+1}b^{k+1}-\cdots-a\beta_{d}b^{d})\\
 & =F(a\beta_{1}b^{1}+\cdots+a\beta_{k}b^{k}).
\end{align*}
Using additivity again yields
\[
F(ax)+F(-a\beta_{k+1}b^{k+1})+\cdots+F(-a\beta_{d}b^{d})=F(a\beta_{1}b^{1})+\cdots+F(a\beta_{k}b^{k}).
\]
Applying now the claim hypothesis that $F(ay)=aF(y)$ whenever $y,ay\in\mathcal{S}$
yields
\[
aF(x)+(-a\beta_{k+1})F(b^{k+1})+\cdots+(-a\beta_{d})F(b^{d})=a\beta_{1}F(b^{1})+\cdots+a\beta_{k}F(b^{k}).
\]
Rearranging and dividing by $a$, we arrive at
\[
F(x)=\beta_{1}F(b^{1})+\cdots+\beta_{d}F(b^{d}).
\]
We can therefore extend $F$ to a function that satisfies this on
all of $\mathbb{R}^{d}$, which is then clearly linear.
\end{proof}
\begin{claim}
\label{claim:open-linear}Let $B$ be an open ball in $\mathbb{R}^{d}$,
and let $\mathcal{B}$ be the semigroup given by $\cup_{a\geq1}aB$.
Then every additive $F\colon\mathcal{B}\to\mathbb{R}_{+}$ is linear.
\end{claim}
\begin{proof}
Fix any $x\in\mathcal{B}$, and assume $ax\in\mathcal{B}$ for some
$a\in\mathbb{R}_{+}$. Since $\mathcal{B}$ is open, by Claim \ref{claim:fax-afx}
it suffices to show that $F(ax)=aF(x)$. The defining property of
$\mathcal{B}$ implies that the intersection of $\mathcal{B}$ and
the ray $\left\{ bx\,:\,b\geq0\right\} $ is of the form $\left\{ bx\,:\:b>a_{0}\right\} $
for some $a_{0}\geq0$. By the additive property of $F$, we have
that $F(qx)=qF(x)$ for every rational $q>a_{0}.$ Furthermore, if
$b>b'>a_{0}$ then $n(b-b')x\in\mathcal{S}$ for $n$ large enough.
Hence
\begin{align*}
F(bx) & =\frac{1}{n}F(nbx)\\
 & =\frac{1}{n}F\left(nb'x+(n(b-b')x)\right)\\
 & =\frac{1}{n}F\left(nb'x\right)+\frac{1}{n}F\left(n(b-b')x\right)\\
 & =F(b'x)+\frac{1}{n}F\left(n(b-b')x\right)\\
 & \geq F(b'x).
\end{align*}
Thus the map $f\colon(a_{0},\infty)\to\mathbb{R}^{+}$ given by $f(b)=F(bx)$
is monotone increasing, and its restriction to the rationals is linear.
So $f$ must be linear, and hence $F(ax)=aF(x).$
\end{proof}
Given these claims, we are ready to prove our theorem.
\begin{proof}
[Proof of Theorem~\ref{thm:cauchy}.]Fix any $x\in\mathcal{S}$, and
assume $ax\in\mathcal{S}$ for some $a\in\mathbb{R}_{+}$. By Claim
\ref{claim:fax-afx} it suffices to show that $F(ax)=aF(x)$. Let
$B$ be a ball with the property described in Claim \ref{claim:open-ball-semigroup},
and denote its center by $x_{0}$ and its radius by $r$. As in Claim
\ref{claim:open-linear}, let $\mathcal{B}$ be the semigroup given
by $\cup_{a\geq1}aB$; note that $\mathcal{B}\subseteq\mathcal{S}$.
Then there is some $y$ such that $x+y,a(x+y),y,ay\in\mathcal{B}$;
in fact, we can take $y=bx_{0}$ for $b=\max\left\{ a,1/a,|x|/r\right\} $
(see Figure \ref{fig:linear-proof}). Then, on the one hand, by additivity,
\[
F(ax+ay)=F(ax)+F(ay).
\]
On the other hand, since $x+y,a(x+y),y,ay\in\mathcal{B},$and since,
by Claim \ref{claim:open-linear}, the restriction of $F$ to $\mathcal{B}$
is linear, we have that
\[
F(ax+ay)=F(a(x+y))=aF(x+y)=aF(x)+aF(y)=aF(x)+F(ay),
\]
thus
\[
F(ax)+F(ay)=aF(x)+F(ay)
\]
and so $F(ax)=aF(x)$. 
\end{proof}
\begin{figure}[H]
\centering{}\includegraphics[scale=0.45]{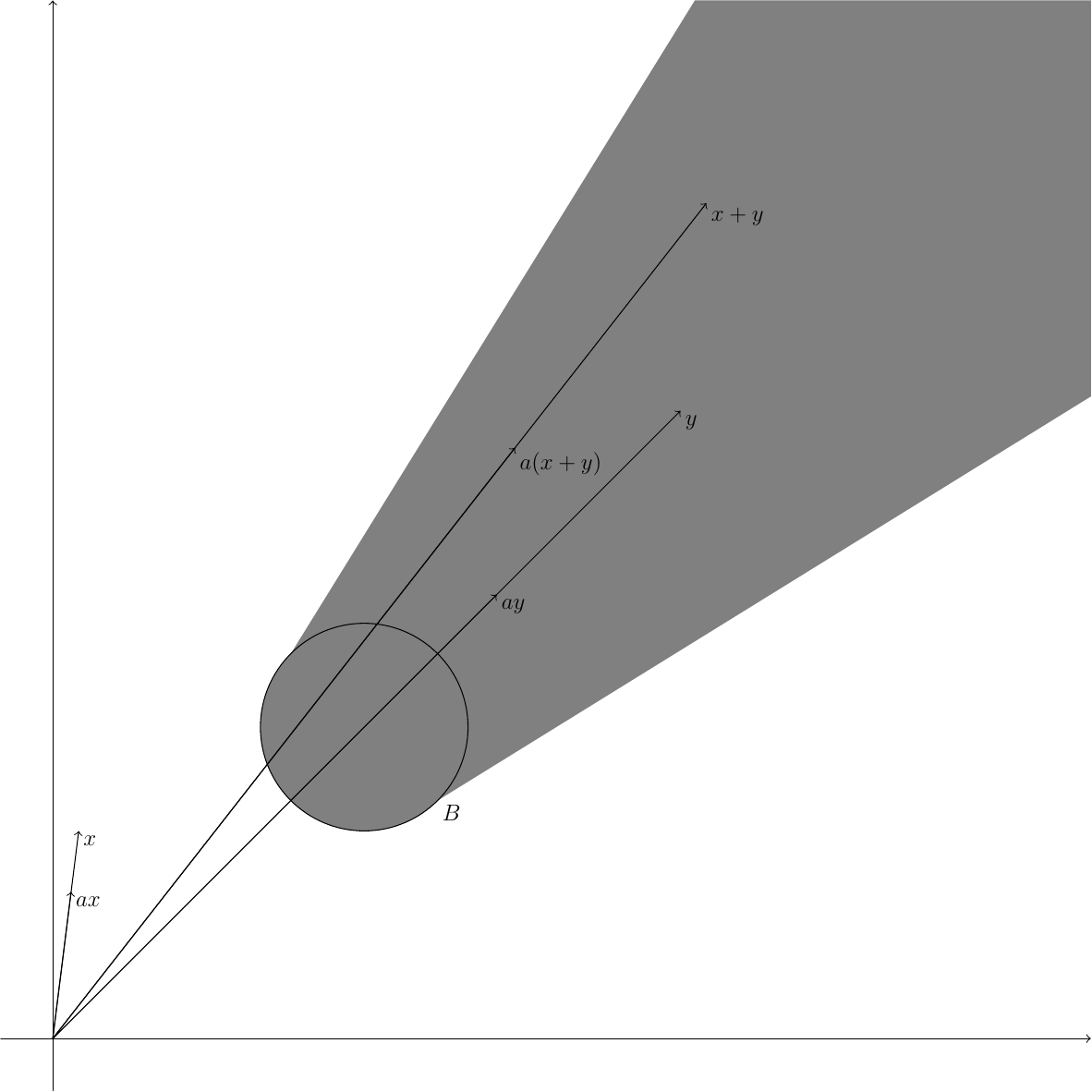}\caption{\label{fig:linear-proof}An illustration of the proof of Theorem \ref{thm:cauchy}. }
\end{figure}

\section{Proof of Theorem \ref{thm:repr-1}}

Throughout this section, we maintain the notation and terminology introduced in \S\ref{sec:pre}. It follows from the results in \S\ref{sec:dkl} that an LLR cost satisfies Axioms 1-4. For the rest of this section, we denote by $C$ a cost function that satisfies the axioms. Let $N$ be such that $C$ is uniformly continuous with respect to the distance $d_N$. We use the same $N$ to define the set $A = \{0,\ldots,N\}^n\backslash\{0,\ldots,0\}$ introduced in \S\ref{sec:cumulants}.
\begin{lem}\label{lem:cost1}
Let $\mu$ and $\nu$ be two experiments that induce the same vector $\sigma\in\mathcal{A}$. Then $C(\mu)=C(\nu)$.
\end{lem}
\begin{proof}
Conditional on each $k\in\Theta$, the two experiments induce the same distribution for $\left(\ell_{i0}\right)_{i\in\Theta}$. Because $\ell_{ij}=\ell_{i0}-\ell_{j0}$ almost surely, it follows that, conditional on each state, the two experiments induce the same distribution over the vector of all log-likelihood ratios $\left(\ell_{ij}\right)_{i,j\in\Theta}$. Hence, $\bar{\mu}_i = \bar{\nu}_i$ for every $i$. Therefore, by Lemma \ref{lem:llr} the two experiments are equivalent in the Blackwell order. The result now follows directly from Axiom \ref{axm:info-content}.
\end{proof}

Lemma \ref{lem:cost1} implies that we can define a function $c:\mathcal{A}\to\mathbb{R}_{+}$ as $c(\sigma)=C(\mu)$ where $\mu$ is an experiment inducing $\sigma$.

\begin{lem}\label{lem:c-prty-1}
Consider two experiments $\mu = (S,(\mu_i))$ and $\nu = (T,(\nu_i))$ that induce $\sigma$ and $\tau$ in $\mathcal{A}$, respectively. Then
\begin{enumerate}
\item The experiment $\mu \otimes \nu$ induces the vector $(\sigma_0\ast\tau_0,\ldots,\sigma_n\ast\tau_n) \in \mathcal{A}$;
\item The experiment $\alpha \cdot \mu$ induces the measure $\alpha \sigma + (1-\alpha)\delta_{\mathbf{0}}$.
\end{enumerate}
\end{lem}
\begin{proof}
(1) For every $E\subseteq\RR^{n}$ and every state $i$,
\begin{eqnarray*}
 &  & (\mu_{i}\times\nu_{i})\left(\left\{ (s,t):\left(\ell_{1}(s,t),\ldots\ell_{n}(s,t)\right)\in E\right\} \right)\\
 & = & (\mu_{i}\times\nu_{i})\left(\left\{ (s,t):\left(\log\frac{\dd\mu_{1}}{\dd\mu_{0}}(s)+\log\frac{\dd\nu_{1}}{\dd\nu_{0}}(t),\ldots,\log\frac{\dd\mu_{n}}{\dd\mu_{0}}(s)+\log\frac{\dd\nu_{1}}{\dd\nu_{n}}(t)\right)\in E\right\} \right)\\
 & = & (\sigma_i\ast\tau_i)(E)
\end{eqnarray*}
where the last equality follows from the definition of $\sigma_i$ and $\tau_i$. This concludes the proof of the claim.

(2) Immediate from the definition of $\alpha \cdot \mu$.
\end{proof}

\begin{lem}\label{lem:c-prty-2}
The function $c : \mathcal{A} \to \mathbb{R}$ satisfies, for all $\sigma,\tau\in\mathcal{A}$ and $\alpha\in [0,1]$:
\begin{enumerate}
    \item $c(\sigma_0\ast\tau_0, \ldots, \sigma_n\ast\tau_n) = c(\sigma) + c(\tau)$;
    \item $c(\alpha \sigma + (1-\alpha)\delta_{\mathbf{0}}) = \alpha c(\sigma)$.
\end{enumerate}
\end{lem}
\begin{proof}
(1) Let $\mu \in \mathcal{E}$ induce $\sigma$ and let $\nu \in \mathcal{E}$ induce $\tau$. Then $C(\mu) = c(\sigma),C(\nu) = c(\tau)$ and, by Axiom \ref{axm:additivity} and Lemma \ref{lem:c-prty-1}, $c(\sigma_0\ast\tau_0, \ldots, \sigma_n\ast\tau_n) = C(\mu\otimes\nu) = c(\sigma) + c(\tau)$. Claim (2) follows directly from Axiom \ref{axm:affinity-1} and Lemma \ref{lem:c-prty-1}.
\end{proof}

\begin{lem}
\label{lem:thm1_c(mu)}If $\sigma,\tau\in\mathcal{A}$ satisfy ${m}_{\sigma}={m}_{\tau}$ then $c(\sigma)=c(\tau)$.
\end{lem}
\begin{proof}
Let $\mu$ be and $\nu$ be two experiments inducing $\sigma$ and $\tau$, respectively. Let $\mu^{\otimes r} = \mu \otimes \ldots \otimes \mu$ be the experiment obtained as the $r$-th fold independent product of $\mu$. Axioms \ref{axm:additivity} and \ref{axm:affinity-1} imply
\[
    C((1/r) \cdot \mu^{\otimes r}) = C(\mu) \text{~~and~~} C((1/r) \cdot \nu^{\otimes r}) = C(\nu)
\]
In order to show that $C(\mu) = C(\nu)$ we now prove that $C((1/r) \cdot \mu^{\otimes r}) - C((1/r) \cdot \nu^{\otimes r}) \to 0$ as $r \to \infty$. To simplify the notation let, for every $r\in\mathbb{N}$,
\[
\mu[r] = (1/r) \cdot \mu^{\otimes r}\text{~~and~~}\nu[r] = (1/r) \cdot \nu^{\otimes r}
\]
Let $\sigma[r] = (\sigma[r]_0,\ldots,\sigma[r]_n)$ and $\tau[r] = (\tau[r]_0,\ldots,\tau[r]_n)$ in $\mathcal{A}$ be the vectors of measures induced by $\mu[r]$ and $\nu[r]$.

We claim that $d_N(\mu[r],\nu[r]) \to 0$ as $r \to \infty$. First, notice that $\overline{\mu[r]}_i$ and $\overline{\nu[r]}_i$ assign probability $(r-1)/r$ to the zero vector $\mathbf{0} \in \RR^{(n+1)^2}$. Hence
\[
    d_{tv}(\overline{\mu[r]}_i, \overline{\nu[r]}_i) = \sup_{E}\frac{1}{r}\left\vert \overline{\mu^{\otimes r} }_i(E)-\overline{\nu^{\otimes r}}_i(E)\right\vert \leq\frac{1}{r}.
\]
For every $\alpha \in A$ we have
\begin{equation}\label{eq:cont}
    M^{\mu[r]}_i(\alpha) = \int \ell^{\alpha_1}_{10}\ldots\ell^{\alpha_n}_{n0}\,\dd \mu[r]_i = \int_{\RR^n} \xi^{\alpha_1}_1\cdots\xi^{\alpha_n}_n\,\dd\sigma[r]_i(\xi) = m_{{\sigma[r]}_i}(\alpha)
\end{equation}
We claim that  ${m}_{\sigma[r]}={m}_{\tau[r]}$. Theorem \ref{thm:Shiryaev} shows the existence of a bijection $H:\mathcal{M}\to\mathcal{K}$ such that $H({m}_{\upsilon})={\kappa}_{\upsilon}$ for every $\upsilon\in\mathcal{A}$. The experiment $\mu^{\otimes r}$ induces the vector $(\sigma_0^{*r},\ldots,\sigma_n^{*r}) \in \mathcal{A}$, where $\sigma_i^{*r}$ denotes the $r$-th fold convolution of $\sigma_i$ with itself. Denote such a vector as $\sigma^{*r}$. Let $\tau^{*r}\in\mathcal{A}$ be the corresponding vector induced by $\nu^{\otimes r}$. Thus we have ${\kappa}_{\sigma}=H({m}_{\sigma})= H({m}_{\tau}) = {\kappa}_{\tau}$, and
\[
 H({m}_{\mu^{*r}})={\kappa}_{\sigma^{*r}}=(\kappa_{\sigma_0^{*r}},\ldots,\kappa_{\sigma_n^{*r}}) = (r\kappa_{\sigma_0},\ldots,r\kappa_{\sigma_n})= r{\kappa}_{\sigma}=r{\kappa}_{\tau}={\kappa}_{\tau^{*r}}=H({m}_{\tau^{*r}})
\]
Hence ${m}_{\sigma^{*r}}={m}_{\tau^{*r}}$. It now follows from
\[
m_{\sigma[r]_{i}}(\alpha)=\frac{1}{r}m_{\sigma_{i}^{*r}}(\alpha)+\frac{r-1}{r}0
\]
that ${m}_{\sigma[r]}={m}_{\tau[r]}$, concluding the proof of the claim.

Equation \eqref{eq:cont} therefore implies that $M^{\mu[r]}_i(\alpha) = M^{\nu[r]}_i(\alpha)$. Thus
\[
    d_N(\mu[r],\nu[r]) = \max_i d_{tv}(\overline{\mu[r]}_i, \overline{\nu[r]}_i) \leq \frac{1}{r}.
\]
Hence $d_N(\mu[r],\nu[r])$ converges to $0$. Since $C$ is uniformly continuous, then $C(\mu[r]) - C(\nu[r]) = 0$ must converge to $0$ as well. This implies $C(\mu) = C(\nu)$.

\end{proof}

\begin{lem}\label{lem:thm1_add}
There exists an additive function $F\colon\mathcal{K}\to\mathbb{R}$ such that $c(\sigma)=F({\kappa}_{\sigma})$.
\end{lem}
\begin{proof}
It follows from Lemma \ref{lem:thm1_c(mu)} that we can define a map $G\colon\mathcal{M}\to\mathbb{R}$ such that $c(\sigma)=G({m}_{\sigma})$ for every $\sigma \in \mathcal{A}$. We can use Theorem \ref{thm:Shiryaev} to define a bijection $H:\mathcal{M}\to\mathcal{K}$ such that $H({m}_{\sigma})={\kappa}_{\sigma}$. Hence $F=G\circ H^{-1}$ satisfies $c(\sigma)=F({\kappa}_{\sigma})$
for every $\sigma$. For every $\sigma,\tau\in\mathcal{A}$, Lemmas \ref{lem:c-prty-1} and \ref{lem:c-prty-2} imply
\[
F({\kappa}_{\sigma})+F({\kappa}_{\tau})=c(\sigma)+c(\tau)=c(\sigma_0\ast\tau_0,\ldots,\sigma_n\ast\tau_n)=F(\kappa_{\sigma_0\ast\tau_0},\ldots,\kappa_{\sigma_n\ast\tau_n}) = F(\kappa_\sigma + \kappa_\tau)
\]
where the last equality follows from the additivity of cumulants with respect to convolution.
\end{proof}
\begin{lem}\label{lem:thm1_lin}
There exist $\left(\lambda_{i,\alpha}\right)_{i\in\Theta,\alpha\in A}$
in $\mathbb{R}$ such that \[c(\sigma)=\sum_{i\in\Theta}\sum_{\alpha\in A}\lambda_{i,\alpha}\kappa_{\sigma_{i}}(\alpha) \text{~~for every~~} \sigma\in\mathcal{A}.\]
\end{lem}
\begin{proof}
As implied by Theorem \ref{thm:K_int}, the set $\mathcal{K}\subseteq\mathbb{R}^{d}$ has nonempty interior. It is closed under addition, i.e.\ a subsemigroup. We can therefore apply Theorem \ref{thm:cauchy} and conclude that the function $F$ in Lemma \ref{lem:thm1_add} is linear.
\end{proof}

\begin{lem}\label{lem:thm1_lin_m}
Let $\left(\lambda_{i,\alpha}\right)_{i\in\Theta,\alpha\in A}$ be as in Lemma \ref{lem:thm1_lin}. Then
\[
c(\sigma)=\sum_{i\in\Theta}\sum_{\alpha\in A}\lambda_{i,\alpha}m_{\sigma_{i}}\left(\alpha\right) \text{~~for every~~} \sigma\in\mathcal{A}
\]
\end{lem}
\begin{proof}
Fix $\sigma \in \mathcal{A}$. Given $t\in\left(0,1\right)$, Lemma \ref{lem:thm1_lin} and Theorem \ref{thm:Shiryaev} imply
\begin{eqnarray*}
c\left(t\sigma+(1-t)\delta_{\mathbf{0}}\right) & = & \sum_{i\in\Theta}\sum_{\alpha\in A}\lambda_{i,\alpha}\left(\sum_{\lambda = \left(\lambda^{1},\ldots,\lambda^{q}\right)\in \Lambda\left(\alpha\right)}\frac{\left(-1\right)^{q-1}}{q}\frac{\alpha!}{\lambda^{1}!\cdots\lambda^{q}!}\prod_{p=1}^{q}m_{t\sigma_{i}+\left(1-t\right)\delta_{0}}\left(\lambda^{p}\right)\right)\\
 & = & \sum_{i\in\Theta}\sum_{\alpha\in A}\lambda_{i,\alpha}\left(\sum_{\lambda =  \left(\lambda^{1},\ldots,\lambda^{q}\right)\in \Lambda\left(\alpha\right)}\frac{\left(-1\right)^{q-1}}{q}\frac{\alpha!}{\lambda^{1}!\cdots\lambda^{q}!}t^{q}\prod_{p=1}^{q}m_{\sigma_{i}}\left(\lambda^{p}\right)\right)\\
 & = & \sum_{i\in\Theta}\sum_{\alpha\in A}\lambda_{i,\alpha}\left(\sum_{\lambda=\left(\lambda^{1},\ldots,\lambda^{q}\right)\in \Lambda\left(\alpha\right)}\rho\left(\lambda\right)t^{q}\prod_{p=1}^{q}m_{\sigma_{i}}\left(\lambda^{p}\right)\right)
\end{eqnarray*}
where for every tuple $\lambda=\left(\lambda^{1},\ldots,\lambda^{q}\right)\in \Lambda(\alpha)$ we let 
\[
\rho\left(\lambda\right)=\frac{\left(-1\right)^{q-1}}{q}\frac{\alpha!}{\lambda^{1}!\cdots\lambda^{q}!}
\]
Lemma \ref{lem:c-prty-2} implies $c(\sigma)=\frac{1}{t}c(t\sigma+\left(1-t\right)\delta_{\mathbf{0}})$ for every $t$. Hence
\[
c(\sigma)=\sum_{i\in\Theta}\sum_{\alpha\in A}\lambda_{i,\alpha}\left(\sum_{\lambda=\left(\lambda^{1},\ldots,\lambda^{q}\right)\in \Lambda(\alpha)}\rho(\lambda)t^{q-1}\prod_{p=1}^{q}m_{\sigma_{i}}(\lambda^{p})\right)\text{ for all }t\in (0,1).
\]
By considering the limit $t\downarrow0$, we have $t^{q-1}\to0$ whenever $q\neq1$. Therefore
\[
c(\sigma)=\sum_{i\in\Theta}\sum_{\alpha\in A}\lambda_{i,\alpha}m_{\sigma_{i}}(\alpha)\text{~~for all~}\sigma\in\mathcal{A}.
\]
\end{proof}

\begin{lem}\label{lem:thm1_last}
Let $\left(\lambda_{i,\alpha}\right)_{i\in\Theta,\alpha\in A}$ be as in Lemmas \ref{lem:thm1_lin} and \ref{lem:thm1_lin_m}. Then, for every $i$, if $\vert\alpha\vert>1$ then $\lambda_{i,\alpha}=0$.
\end{lem}
\begin{proof}
Let $\gamma=\max\text{\ensuremath{\left\{  |\alpha|:\lambda_{i,\alpha}\neq0\text{ for some }i\right\} } }$.
Assume, as a way of contradiction, that $\gamma>1$. Fix $\sigma\in\mathcal{A}$. Theorem \ref{thm:Shiryaev} implies
\begin{eqnarray*}
c(\sigma) & = & \sum_{i\in\Theta}\sum_{\alpha\in A}\lambda_{i,\alpha}m_{\sigma_{i}}(\alpha)\\
 & = & \sum_{i\in\Theta}\sum_{\alpha\in A}\lambda_{i,\alpha}\left(\sum_{\left(\lambda^{1},\ldots,\lambda^{q}\right)\in \Lambda(\alpha)}\frac{1}{q!}\frac{\alpha!}{\lambda^{1}!\cdots\lambda^{q}!}\prod_{p=1}^{q}\kappa_{\sigma_{i}}(\lambda^{p})\right)
\end{eqnarray*}
For all $r\in\mathbb{N}$, let $\sigma^{*r} = (\sigma_0^{*r},\ldots,\sigma_0^{*r})$, where each $\sigma_i^{*r}$ is the $r$-th fold convolution of $\sigma_i$ with itself. Hence, using the fact that $\kappa_{\sigma_i^{*r}}=r\kappa_{\sigma_i}$, we obtain
\begin{equation}\label{eq:sigma_m}
c(\sigma^{*r})=\sum_{i\in\Theta}\sum_{\alpha\in A}\lambda_{i,\alpha}\left(\sum_{\left(\lambda^{1},\ldots,\lambda^{q}\right)\in \Lambda(\alpha)}\frac{1}{q!}\frac{\alpha!}{\lambda^{1}!\cdots\lambda^{q}!}r^{q}\prod_{p=1}^{q}\kappa_{\sigma_{i}}(\lambda^{p})\right)
\end{equation}
By the additivity of $c$, $c(\sigma^{*r})=rc(\sigma)$. Hence, because $\gamma >1$, $c(\sigma^{*r})/r^{\gamma}\to0$ as $r\to\infty$. Therefore, diving \eqref{eq:sigma_m} by $r^\gamma$ implies
\begin{equation}\label{eq:thm1}
\sum_{i\in\Theta}\sum_{\alpha\in A}\lambda_{i,\alpha}\left(\sum_{\left(\lambda^{1},\ldots,\lambda^{q}\right)\in \Lambda(\alpha)}\frac{1}{q!}\frac{\alpha!}{\lambda^{1}!\cdots\lambda^{q}!}r^{q-\gamma}\prod_{p=1}^{q}\kappa_{\sigma_i}(\lambda^{p})\right)\to0\text{ as }r \to\infty.
\end{equation}
We now show that \eqref{eq:thm1} leads to a contradiction. By construction, if $\left(\lambda^{1},\ldots,\lambda^{q}\right)\in \Lambda(\alpha)$ then $q\leq|\alpha|$. Hence $q \leq \gamma$ whenever $\lambda_{i,\alpha}\neq0$. So, in equation \eqref{eq:thm1} we have $r^{q-\gamma}\to0$ as $r \to \infty$ whenever $q<\gamma$. Hence in order for \eqref{eq:thm1} to hold it must be that
\[
 \sum_{i\in\Theta}\sum_{\alpha \in A: |\alpha|=\gamma}\lambda_{i,\alpha}\left(\sum_{\left(\lambda^{1},\ldots,\lambda^{q}\right)\in \Lambda(\alpha), q = \gamma}\frac{1}{q!}\frac{\alpha!}{\lambda^{1}!\cdots\lambda^{q}!}\prod_{p=1}^{q}\kappa_{\sigma_i}\left(\lambda^{p}\right)\right)=0.   
\]
If $q = \gamma$ and $\lambda_{i,\alpha} \neq 0$ then $\gamma = |\alpha|$. In this case, in order for $\lambda=\left(\lambda^{1},\ldots,\lambda^{q}\right)$ to satisfy $\sum_{p=1}^{q}\lambda^{p}=\alpha$, it must be that each $\lambda^{p}$ is a unit vector.
Every such $\lambda$ satisfies%
\footnote{It follows from the definition of cumulant that for every unit vector $1_j \in \RR^n$, $\kappa_{\sigma_i}(1_j) = \int_{\RR^n}\xi_j\,\dd\sigma_i(\xi)$.}
\[
    \prod_{p=1}^{q}\kappa_{\sigma_i}(\lambda^{p}) = \left(\int_{\RR^n} \xi_{1}\,\dd\sigma_{i}\left(\xi\right)\right)^{\alpha_{1}}\cdots\left(\int_{\RR^n} \xi_{n}\,\dd\sigma_{i}\left(\xi\right)\right)^{\alpha_{n}}
\]
and
\[
    \sum_{\left(\lambda ^{1},\ldots,\lambda ^{q}\right)\in \Lambda(\alpha), q=|\alpha|}\frac{1}{q!}\frac{\alpha!}{\lambda ^{1}!\cdots\lambda ^{q}!} = \sum_{\left(\lambda ^{1},\ldots,\lambda ^{q}\right)\in \Lambda(\alpha), q=|\alpha|}\frac{\alpha!}{|\alpha|! } = L(\alpha)
\]
where $L(\alpha)$ is the cardinality of the set of $\left(\lambda ^{1},\ldots,\lambda ^{q}\right)\in \Lambda(\alpha)$ such that $q=|\alpha|$. We obtain that
\begin{equation}\label{eq:notation}
\sum_{i\in\Theta}\sum_{\alpha\in A:|\alpha | = \gamma}L(\alpha)\lambda_{i,\alpha}\left(\int_{\RR^n} \xi_{1}\,\dd\sigma_{i}\left(\xi\right)\right)^{\alpha_{1}}\cdots\left(\int_{\RR^n} \xi_{n}\,\dd\sigma_{i}\left(\xi\right)\right)^{\alpha_{n}}=0.
\end{equation}
By replicating the argument in the proof of Lemma \ref{lem:M_int} we obtain that the set 
\[
\left\{ \left(\int_{\RR^n} \xi_{j}\,\dd\sigma_i(\xi)\right)_{i,j\in\Theta,j>0}:\sigma\in\mathcal{A}\right\} \subseteq\RR^{(n+1)n}
\]
contains an open set $U$. Consider now the function $f:\RR^{(n+1)n}\to\mathbb{R}$
defined as
\[
f(z)=\sum_{i\in\Theta}\sum_{\alpha\in A:|\alpha|=\gamma}L(\alpha)\lambda_{i,\alpha}z_{i,1}^{\alpha_{1}}\cdots z_{i,n}^{\alpha_{n}}, ~~ z\in\RR^{(n+1)n}
\]
Then \eqref{eq:notation} implies that $f$ equals $0$ on $U$. Hence, for every $z \in U$,$i\in\Theta$ and $\alpha\in A$ such that $|\alpha|=\gamma$, 
\[
L(\alpha)\lambda_{i,\alpha}=\frac{\partial^{\gamma}}{\partial^{\alpha_{1}}z_{i,1}\cdots\partial^{\alpha_{n}}z_{i,n}}f(z)=0
\]
hence $\lambda_{i,\alpha} = 0$. This contradicts the assumption that $\gamma>1$ and  concludes the proof.
\end{proof}

For every $j\in\{1,\ldots,n\}$ let $1_{j} \in A$ be the corresponding unit vector. We write $\lambda_{ij}$ for $\lambda_{i,j}$.
Lemma \ref{lem:thm1_last} implies that for every distribution $\sigma\in\mathcal{A}$ induced by an experiment $(S,(\mu_{i}))$, the function $c$ satisfies
\begin{eqnarray*}
c(\sigma) & = & \sum_{i\in\Theta}\sum_{j\in\{1,\ldots,n\}}\lambda_{ij}\int_{\RR^n} \xi_{j}\,\dd\sigma_i(\xi)\\
 & = & \sum_{i\in\Theta}\sum_{j\in\{1,\ldots,n\}}\lambda_{ij}\int_{S}\log\frac{\dd\mu_{j}}{\dd\mu_{0}}(s)\,\dd\mu_{i}(s)\\
 & = & \sum_{i\in\Theta}\sum_{j\in\{1,\ldots,n\}}\lambda_{ij}\int_{S}\log\frac{\dd\mu_{j}}{\dd\mu_{0}}(s)+\log\frac{\dd\mu_{0}}{\dd\mu_{i}}(s)-\log\frac{\dd\mu_{0}}{\dd\mu_{i}}(s)\,\dd\mu_{i}(s)
\end{eqnarray*}
Hence, using the fact that $\frac{\dd \mu_j}{\dd \mu_0} \frac{\dd \mu_0}{\dd \mu_i} = \frac{\dd \mu_j}{\dd \mu_i}$, we obtain
\begin{eqnarray*}
c(\sigma) & = & \sum_{i\in\Theta}\sum_{j\in\{1,\ldots,n\}}\lambda_{ij}\int_{S}\log\frac{\dd\mu_{j}}{\dd\mu_{i}}\,\dd\mu_{i}(s) + \sum_{i\in\Theta}\left(-\sum_{j\in\{1,\ldots,n\}}\lambda_{ij}\right)\int_S\log\frac{\dd\mu_{0}}{\dd\mu_{i}}(s)\,\dd\mu_{i}(s)\\
 & = & \sum_{i,j\in\Theta}\beta_{ij}\int_{S}\log\frac{\dd\mu_{i}}{\dd\mu_{j}}(s)\,\dd\mu_{i}(s)
\end{eqnarray*}
where in the last step, for every $i$, we set $\beta_{ij}=-\lambda_{ij}$
if $j\neq0$ and $\beta_{i0}=\sum_{j\neq0}\lambda_{ij}$.

It remains to show that the coefficients $(\beta_{ij})$ are positive and unique. Because $C$ takes positive values, Lemma \ref{lem:domain} immediately implies $\beta_{ij} \geq 0$ for all $i,j$. The same Lemma easily implies that the coefficients are unique given $C$.

\section{Proofs of the Results of Section~\ref{sec:stochastic}}

\begin{proof}[Proof of Proposition~\ref{prop:foc}]
Let $\mu^\star \in \P(A)^n$ be an optimal experiment. Let $A^\star=\mathrm{supp}(\mu^\star)$ be the set of actions played in $\mu^\star$. It solves
\begin{align}
    &\max_{\mu \in \RR_+^{|\Theta| \times |A^\star|} } \left[ \sum_{i \in \Theta} q_i \left( \sum_{a \in A} \mu_i (a) u(a,i) \right) - \sum_{i,j \in \Theta} \beta_{ij} \sum_{a \in A^\star} \mu_i(a) \log \frac{\mu_i(a)}{\mu_j(a)} \right] \label{eq:choice-probabilties-2}\\
    \text{subject to}&\hspace{1cm}\sum_{a\in A^\star} \mu_i(a) = 1 \text{ for all } i \in \Theta.\label{eq:sum-constraint} 
\end{align}
Reasoning as in \citet[][Theorem 2.7.2]{cover2012elements} the Log-sum inequality implies that the function $\dkl$ is convex when its domain is extended from pairs of probability distributions to pairs of vectors in $\RR_+^{|A^\star|}$. Moreover, expected utility is linear in the choice probabilities. It then follows that the objective function in \eqref{eq:choice-probabilties-2} is concave over $\RR_+^{|\Theta| \times |A^\star|}$. 


As  \eqref{eq:choice-probabilties-2} equals $-\infty$ whenever $\mu_i(a)=0$ for some $i$ and $\mu_j(a)>0$ for some $j\neq i$ we have that $\mu^\star_i(a)>0$ for all $i\in \Theta, a\in A^\star$.
For every $\lambda \in \RR^{|\Theta|}$ we define the Lagrangian $L_\lambda(\mu)$ as
\[
    L_\lambda(\mu) = \left[ \sum_{i \in \Theta} q_i \left( \sum_{a \in A} \mu_i (a) u(a,i) \right) - \sum_{i,j \in \Theta} \beta_{ij} \sum_{a \in A} \mu_i(a) \log \frac{\mu_i(a)}{\mu_j(a)} \right] - \sum_{i \in \Theta} \lambda_i \sum_{a \in A} \mu_i(a) \,.
\]
As $\mu^\star$ is an interior solution to \eqref{eq:choice-probabilties-2}, it follows from the Karush-Kuhn-Tucker theorem that there exists Lagrange multipliers $\lambda\in \RR^{|\Theta|}$ such that $\mu^\star$ maximizes $L_\lambda(\cdot)$ over $\RR_+^{|\Theta| \times |A^\star|}$. As $\mu^\star$ is interior it satisfies the first order condition
\[
\nabla L_\lambda(\mu^\star) = 0 \,.
\]
We thus have that for every state $i \in \Theta$ and every action $a \in A^\star$
\begin{equation}\label{eq:foc-lagrange}
    0 = q_i u_i(a) - \lambda_i - \sum_{j \neq i}  \left \{ \beta_{i j} \left[ \log \left( \frac{\mu^\star_i(a)}{\mu^\star_j(a)} \right) - 1 \right] - \beta_{ji} \frac{\mu^\star_j(a)}{\mu^\star_i(a)} \right\} \,.
\end{equation}

Subtracting \eqref{eq:foc-lagrange} evaluated at $a'$ from \eqref{eq:foc-lagrange} evaluated at $a$ yields the desired necessary conditions for the optimality of $\mu^\star$.
\end{proof}

\begin{proof}[Proof of Proposition~\ref{prop:choice-continuity}]
 We prove a slightly more general result. Assume the coefficients satisfy $\beta_{ij} \geq 1/f(d(i,j))^2$, where $f$ is a strictly positive and increasing function $f$.
 
  The cost of the optimal experiment $\mu^\star$ must satisfy $\Vert u \Vert \geq C(\mu^\star)$, otherwise the decision maker would be better off acquiring no information. Pinsker's inequality \citep[see][p. 13]{borwein2010convex} implies
  \[
        C(\mu^\star) \geq \min\{\beta_{ij},\beta_{ji}\} (\dkl(\mu^\star_i \Vert \mu^\star_j) + \dkl(\mu^\star_j \Vert \mu^\star_i)) \geq \min\{\beta_{ij},\beta_{ji}\} \Vert \mu^\star_i - \mu^\star_j \Vert_1^2.
  \]
  where $\Vert \mu^\star_i - \mu^\star_j \Vert_1 = \sum_{a \in A}| \mu^\star_i(a) - \mu^\star_j(a) |$ denotes the total-variation norm between the two distributions. We then obtain
  \[
    \Vert \mu^\star_i - \mu^\star_j \Vert_1 \leq \sqrt{ \Vert u \Vert \frac{1}{\min\{\beta_{ij},\beta_{ji}\}}} \leq \sqrt{ \Vert u \Vert} f(d(i,j)) \,.
  \]
  In particular, if $f$ is the identity function then $\Vert \mu^\star_i - \mu^\star_j \leq \sqrt{ \Vert u \Vert} d(i,j)$.
 \end{proof}

\begin{proof}[Proof of Proposition~\ref{prop:perception-monotonicity}]

 Given a vector $\mu \in \P(\{B,R\})^{\Theta}$, we use the shorthand $\mu_i$ to denote the probability $\mu_i(B)$ of guessing $B$ in state $i$. For every $\mu$, let
  \begin{equation}\label{eq:perception-monotonicity-1}
	U(\mu) =  \frac{1}{\vert \Theta \vert}\left(\sum_{i < n/2} (1-\mu_i) + \sum_{i > n/2} \mu_i \right) - C(\mu) \, 
 \end{equation}
 be the net expected payoff provided by $\mu$, where $C$ is an LLR cost function such that $\beta_{ij} = f(\vert i - j \vert)$ for some positive and strictly decreasing function $f$.
 
 Let $\P_+$ be the set of probabilities $\mu$ such that each $\mu_i$ has support $\{B,R\}$. Let $\mu^\star$ be a solution to the problem $\max_{\mu \in \P_+} U(\mu)$. Such a solution exists and is unique. In fact, the problem $\max_{\mu \in \P(\{B,R\})^{\Theta}} U(\mu)$ has a solution. Now, if $\mu^\star$ is optimal and $\mu^\star \notin \P_+$, then either $\mu^\star_i = 0$ for every $i$ or $\mu^\star_i = 1$ for every $i$. In either case $U(\mu^\star) = U(\mu)$, where $\mu \in \P_+$ is defined as $\mu_i = 1/2$ for every $i$. It follows that the problem $\max_{\mu \in \P_+} U(\mu)$ admits a solution $\mu^\star$. Over $\P_+$ the function $C$ is strictly convex,\footnote{See Corollary 1.55 in \cite*{liese1987convex}} and thus $U$ is strictly concave. Thus, the solution is unique.

 We claim that $\mu^\star$ satisfies $\mu^\star_{n/2 + r} = 1 - \mu^\star_{n/2 - r}$ for every $r$. To see this, define $\mu \in \P_+$ as $\mu_{n/2 + r} = 1 - \mu^\star_{n/2 - r}$ for every $r$. Because $U(\mu^\star) = U(\mu)$ and $U$ is strictly concave on $\P_+$, we conclude that $\mu = \mu^\star$.
 
 Let $I  \subseteq \P(\{B,R\})^{\Theta}$ be the set of vectors $\mu$ that are increasing, that is, satisfy $\mu_i \leq \mu_{i+1}$ for every $i < n$, and consider the optimization problem
 \[
    \max_{\mu \in I \cap \P_+} U(\mu).
 \]
 The set $I$ is closed and $U$ is upper semi-continuous. Thus, the problem $\max_{\mu \in I} U(\mu)$ has a solution. The same argument applied in the previous paragraph implies $\max_{\mu \in I \cap \P_+} U(\mu)$ admits a solution as well, and that such a solution is unique. We denote it by $\hat{\mu}$.
 
 As we show in the next paragraph, the vector $\hat{\mu}$ is strictly increasing: it satisfies $\hat{\mu}_i < \hat{\mu}_{i+1}$ for every $i$. This implies $\mu^\star = \hat{\mu}$. Indeed, we have $U(\mu^\star) \geq U(\hat{\mu})$, since $\mu^\star$ is obtained by maximizing $U$ over a larger domain. If $U(\mu^\star) > U(\hat{\mu})$ the concavity of $U$ implies $U(\alpha \mu^\star + (1-\alpha) \hat{\mu}) > U(\hat{\mu})$ for all $\alpha \in [0,1]$. Because $\hat{\mu}$ is strictly increasing, then for $\alpha$ small enough the vector $\alpha \mu^\star + (1-\alpha) \hat{\mu}$ belongs to $I$, contradicting the optimality of $\hat{\mu}$. It follows that $U(\mu^\star) = U(\hat{\mu})$, and hence $\mu^\star = \hat{\mu}$, since the problem $\max_{\mu \in \P_+} U(\mu)$ has a unique solution.
 
 We now show $\hat{\mu}$ is strictly increasing. Given $\nu,\rho \in (0,1)$ we denote by $D_1(\nu \Vert \rho)$ and $D_2(\nu \Vert \rho)$ the partial derivatives of the Kullback-Leibler divergence $\dkl$ with respect to its the first and second arguments:
 \begin{align*}
     D_1(\rho \Vert \nu) &= \log{\frac{\rho}{\nu}} - \log{\frac{1-\rho}{1-\nu}} \\
     D_2(\rho \Vert \nu) &= -{\frac{\rho}{\nu}} + {\frac{1-\rho}{1-\nu}}.
 \end{align*}
 Both derivatives are equal to zero if and only if $\nu = \rho$.
 
 As a way of contradiction, suppose $\hat{\mu}$ is not strictly increasing. Let $[i,k]$ be a maximal interval of states over which $\hat{\mu}$ is constant. Let $\mu^\varepsilon$ be the vector obtained from $\hat{\mu}$ by increasing $\hat{\mu}_k$ by $\varepsilon > 0$ and decreasing $\hat{\mu}_i$ by $\varepsilon$ (since $\hat{\mu} \in \P_+$, both operations are feasible for $\varepsilon$ small enough). The function $\varepsilon \mapsto U(\mu^\varepsilon)$ is differentiable. Its derivative at $\varepsilon = 0$ is equal to
 \begin{equation}\label{eq.derivative}
 \frac{\mathrm{sgn}(k - n/2)}{\vert\Theta\vert} - \sum_{j \neq k} \beta_{jk}(D_2(\hat{\mu}_j \Vert \hat{\mu}_k) + D_1(\hat{\mu}_k \Vert \hat{\mu}_j)) - \frac{\mathrm{sgn}(i - n/2)}{\vert\Theta\vert} + \sum_{j \neq i} \beta_{ij}(D_2(\hat{\mu}_j \Vert \hat{\mu}_i) + D_1(\hat{\mu}_k \Vert \hat{\mu}_i))\,.
 \end{equation}
 Since $\hat{\mu}$ is constant in the interval $[i,k]$, then $D_1(\hat{\mu}_j \Vert \hat{\mu}_m) = D_2(\hat{\mu}_j \Vert \hat{\mu}_m)$ whenever $i \leq j \leq m \leq k$. We can therefore rewrite \eqref{eq.derivative} as
 \begin{equation}\label{eq.align.derivative}
 \begin{split}
 &\frac{\mathrm{sgn}(k - n/2)}{\vert\Theta\vert} - \sum_{j > k} \beta_{jk}(D_2(\hat{\mu}_j \Vert \hat{\mu}_k) + D_1(\hat{\mu}_k \Vert \hat{\mu}_j)) - \sum_{j < i} \beta_{jk}(D_2(\hat{\mu}_j \Vert \hat{\mu}_k) + D_1(\hat{\mu}_k \Vert \hat{\mu}_j))  \\
         - &\frac{\mathrm{sgn}(i - n/2)}{\vert\Theta\vert} + \sum_{j > k} \beta_{ij}(D_2(\hat{\mu}_j \Vert \hat{\mu}_i) + D_1(\hat{\mu}_k \Vert \hat{\mu}_i)) + \sum_{j < i} \beta_{ij} (D_2(\hat{\mu}_j \Vert \hat{\mu}_i) + D_1(\hat{\mu}_i \Vert \hat{\mu}_j))\, .
 \end{split}
 \end{equation}
 The derivative \eqref{eq.align.derivative} is strictly positive. Indeed, because $k \geq i$ then $ \mathrm{sgn}(k - n/2) - \mathrm{sgn}(i - n/2) \geq 0$. Whenever $j > k$, since $\hat{\mu}_j > \hat{\mu}_k = \hat{\mu}_i$ and $D$ is strictly convex over $\P_+$, we have
 \[
    D_2(\hat{\mu}_j \Vert \hat{\mu}_k) = D_2(\hat{\mu}_j \Vert \hat{\mu}_i) < 0 \text{~and~} D_1(\hat{\mu}_k \Vert \hat{\mu}_j) = D_1(\hat{\mu}_i \Vert \hat{\mu}_j) < 0
 \]
 Moreover $\beta_{jk} > \beta_{ji}$ since $|j-k| < |i-k|$. It follows that
 \[
    - \sum_{j > k} \beta_{jk}(D_2(\hat{\mu}_j \Vert \hat{\mu}_k) + D_1(\hat{\mu}_k \Vert \hat{\mu}_j)) + \sum_{j > k} \beta_{ij}(D_2(\hat{\mu}_j \Vert \hat{\mu}_i) + D_1(\hat{\mu}_i \Vert \hat{\mu}_j))
 \]
 is strictly positive if $k < n$, and equal to $0$ if $k = n$. An analogous argument shows that
 \[
    - \sum_{j < i} \beta_{jk}(D_2(\hat{\mu}_j \Vert \hat{\mu}_k) + D_1(\hat{\mu}_k \Vert \hat{\mu}_j)) +  \sum_{j < i} \beta_{ij} (D_2(\hat{\mu}_j \Vert \hat{\mu}_i) + D_1(\hat{\mu}_i \Vert \hat{\mu}_j))
 \]
 is strictly positive if $i > 0$, and equal to $0$ if $i = 0$. Because $\hat{\mu} \in \P_+$, then either $k < n/2 + r$, $i > n/2 - r$, or both. This implies that \eqref{eq.align.derivative} is strictly positive. Hence, for small enough $\varepsilon$, the vector $\mu^\varepsilon$ satisfies $U(\mu^\varepsilon) > U(\hat{\mu})$, contradicting the hypothesis that $\hat{\mu}$ is optimal. We therefore conclude that $\hat{\mu}$ is strictly increasing, and thus $\mu^\star$ is strictly increasing as well.
 
 Because $\mu^\star$ satisfies $\mu^\star_{n/2 + r} = \mu^\star_{n/2 - r}$ for every $r$, and $\mu^\star$ is strictly increasing, it follows that $m_i > m_j$ for every pair of states such that $\vert i - n/2 \vert > \vert j - n/2 \vert$.
\end{proof}

\begin{proof}[Proof of Proposition \ref{prop:binary-choice}]
 Denote by $\P_+$ be the set of probabilities $\mu \in \P(\{a_1,a_2\})^2$ such that $\mathrm{supp}(\mu) = \{a_1,a_2\}$. Let $\mu \in \P_+$ be an optimal experiment. We first show that $\mu$ satisfies $\mu_1(a_1) = \mu_2(a_2)$. To see this, define $\mu'$ as $\mu'_1(a_1) = \mu_2(a_2)$ and $\mu'_2(a_2) = \mu_1(a_1)$. Let $\mu'' = \frac{1}{2} \mu + \frac{1}{2}\mu'$. By the symmetry of the payoffs functions and of the prior, we have
 \[
  \sum_{i \in \Theta} q_i \left( \sum_{a \in A} \mu_i (a) u(a,i) \right) = \sum_{i \in \Theta} q_i \left( \sum_{a \in A} \mu'_i (a) u(a,i) \right) = \sum_{i \in \Theta} q_i \left( \sum_{a \in A} \mu''_i (a) u(a,i) \right) \,.
 \]
 Moreover, $C(\mu'') \leq \frac{1}{2}C(\mu) + \frac{1}{2}C(\mu')$ if $\mu \neq \mu'$, as $C$ is strictly convex on $\P_+$. Since $\mu$ is optimal, it must be that $\mu = \mu'$.
 
 The optimality equation $\mathrm{MB}_1(a_1,a_2) = \mathrm{MC}_1(a_1,a_2)$ can now be rewritten as
 \[
    \frac{1}{2}v = \beta\left[\xi\left(\log\left(\frac{\mu_1(a_1)}{\mu_2(a_1)}\right)\right)- \xi\left(\log\left(\frac{\mu_1(a_2)}{\mu_2(a_2)}\right)\right)\right].
 \]
 with $\xi(x) = x + \ee^x$. Simple calculations show the expression is in turn equal to
 \[
    \frac{v}{2\beta} = \xi\left(\log\left(\frac{\mu[v]}{1-\mu[v]}\right)\right) - \xi\left(\log\left(- \frac{\mu[v]}{1-\mu[v]} \right)\right) = \zeta\left(\log\left(\frac{\mu[v]}{1-\mu[v]}\right)\right)
 \]
 where $\zeta(x) = 2 x + e^x - e^{-x}$. The result now follows by defining $\eta = \zeta^{-1}$.
\end{proof}

 \begin{proof}[Proof of Proposition~\ref{prop:comparative-statics-in-beta}]
 Consider a decision problem described by a payoff function $u$ and a prior $q$. let $\mu$ and $\mu'$ be the optimal choice probabilities obtained under the coefficients $(\beta_{ij})$ and $(\beta'_{ij})$. The optimality of $\mu$ and $\mu'$ implies
 \begin{align*}
      \sum_{i,a} q_i u(i,a)\mu_i(a) - \sum_{i,j} \beta_{ij} D(\mu_i \Vert \mu_j) &\geq \sum_{i,a} q_i u(i,a)\mu'_i(a) - \sum_{i,j} \beta_{ij} D(\mu'_i \Vert \mu'_j) \\
      \sum_{i,a} q_i u(i,a)\mu'_i(a) - \sum_{i,j} \beta'_{ij} D(\mu'_i \Vert \mu'_j) &\geq \sum_{i,a} q_i u(i,a)\mu_i(a) - \sum_{i,j} \beta'_{ij} D(\mu_i \Vert \mu_j) \,
 \end{align*}
 Rearranging the two inequalities leads to
 \[
    \sum_{i,j }\beta_{ij} (D(\mu'_i \Vert \mu'_j) - D(\mu_i \Vert \mu_j)) \geq \sum_{i,a} q_i u(i,a)(\mu'_i(a) - \mu_i(a)) \geq  \sum_{i,j }\beta'_{ij} (D(\mu'_i \Vert \mu'_j) - D(\mu_i \Vert \mu_j)).
 \]
 The result now follows.
 \end{proof}

\section{Proof of Proposition \ref{prop:onedimens} and Extensions}

\begin{proof}[Proof of Proposition~\ref{prop:onedimens}]
Denote by $w >0$ the length of $W$. Let $|\Theta| = n$. By Axiom~\ref{axb1} there exists a function $f \colon (0,w) \to \RR_+$ such that $\beta^\Theta_{ij}=f(|i-j|)$ for $i \neq j$. Hence, if we translate $W$ then $\beta^\Theta_{ij}$ remains unchanged. We can therefore assume without loss of generality that $W = (-\delta,w-\delta)$, for any $\delta \in (0,w)$.

Let $g \colon (0,w) \to \RR_{+}$ be given by $g(t) = \frac{1}{2}f(t)t^2$. The Kullback-Leibler divergence between two normal distributions with unit variance and expectations $i$ and $j$ is $(i-j)^2/2$. Hence, by Axiom~\ref{axb2} there exists a constant $\kappa \geq 0$, independent of $n$, so that 
\begin{align}
  \label{eq:kappa}
  \frac{1}{2}\kappa = C^\Theta(\zeta^\Theta) = \sum_{i \neq j\in \Theta} \beta^\Theta_{ij}\frac{(i-j)^2}{2} = \sum_{i \neq j\in \Theta} g(|i-j|)\quad\quad\text{for any } \Theta \in \mathcal{T}
\end{align}
We show that \eqref{eq:kappa} implies that
\begin{align*}
  g(t) = \frac{\kappa}{2n(n-1)},
\end{align*}
so that
\begin{align*}
  \beta^\Theta_{ij} = 2g(|i-j|)\frac{1}{(i-j)^2} = \frac{\kappa}{n(n-1)}\frac{1}{(i-j)^2},
\end{align*}
which will complete the proof. The case $n=2$ is immediate, since then $\Theta = \{i,j\}$ and so \eqref{eq:kappa} reduces to
\[
    \frac{1}{2}\kappa = 2g(|i-j|).
\]

We now consider the case $n  > 2$. Let $\Theta = \{i_1,i_2,\ldots, i_{n-1}, x\}$ with $i_1 < i_2 < \cdots < i_{n-1} < x$ and $x \in (0,w-\delta)$.  Then \eqref{eq:kappa} implies
\begin{align*}
  \kappa &= 2\sum_{\ell=1}^{n-1}g(x-i_\ell) + 2\sum_{k=1}^{n-1}\sum_{\ell=1}^{k-1} g(i_k-i_\ell).
\end{align*}
Taking the difference between this equation and the analogous one corresponding to $\Theta' = \{i_1,i_2,\ldots,i_{n-1},y\}$ with $y \in (x,w-\delta)$ yields
\begin{align*}
  0 = \sum_{\ell=1}^{n-1}g(x-i_\ell)-g(y-i_\ell).
\end{align*}
Denoting $i_1=-\varepsilon$, for some $\varepsilon \in (\delta,0)$, we can write this as
\begin{align*}
  0 = g(x+\varepsilon) -g(y+\varepsilon)+ \sum_{\ell=2}^{n-1}g(x-i_\ell)-g(y-i_\ell).
\end{align*}
Again taking a difference, this time of this equation with the analogous one obtained by setting $i_1 = 0$, we get
\begin{align*}
  g(x) -g(y) = g(x+\varepsilon) -g(y+\varepsilon).
\end{align*}
Rearranging yields
\begin{align}
  \label{eq:double-diff}
     g(y+\varepsilon)-g(y) = g(x+\varepsilon) -g(x) \quad\quad\text{for all }x,y \in (0,w-\delta) \text{ and } \varepsilon \in (0,\delta). 
\end{align}
Accordingly, for $\varepsilon \in (0,\delta)$ denote
\begin{align}
  \label{eq:h}
  h(\varepsilon) = g(x+\varepsilon)-g(x),
\end{align}
where by \eqref{eq:double-diff} the right hand side does not depend on the choice of $x \in (0,w-\delta-\varepsilon)$. It follows that
\begin{align}
  \label{eq:h-cauchy}
  h(\varepsilon_1 + \varepsilon_2) = [g(x+\varepsilon_1+\varepsilon_2) - g(x+\varepsilon_1)]+[g(x+\varepsilon_1)-g(x)] = h(\varepsilon_1) + h(\varepsilon_2)
\end{align}
for all $\varepsilon_1,\varepsilon_2 \in (0,\delta/2)$. That is, $h$ satisfies the Cauchy functional equation on $(0,\delta/2)$.

Since $g$ is non-negative, it follows from \eqref{eq:kappa} that $g$ is bounded by $\kappa$. Hence the absolute value of $h$ is bounded by $\kappa$, by \eqref{eq:h}. It follows that $\lim_{\varepsilon \to 0}h(\varepsilon)=0$. Otherwise, there is some $n$ such that $|h(\varepsilon)|>\kappa/n$ for arbitrarily small $\varepsilon$, and then, by repeated application of \eqref{eq:h-cauchy},
\begin{align*}
  h(n\varepsilon) = n h(\varepsilon) > \kappa,
\end{align*}
where we choose $\varepsilon$ small enough so that $n\varepsilon < \delta/2$.

From $\lim_{\varepsilon \to 0}h(\varepsilon)=0$ and \eqref{eq:h-cauchy} it follows that $h$ is continuous on $(0,\delta/2)$. As the Cauchy equation easily implies that $h$ is linear when restricted to the rationals, continuity implies that $h$ is linear on $(0,\delta/2)$. Thus, by \eqref{eq:double-diff} $g$ is affine on $(0,w-\delta)$, and of the form $g(t) = at+b$ for some $a,b \in \RR$. We claim that it must be that $a=0$. Otherwise, for a given $\Theta = \{i_1,\ldots,i_{n-1},x\}$, $\sum_{i\neq j \in \Theta}g(|i-j|)$ changes with $x$, in violation of \eqref{eq:kappa}. It follows that $g$ is constant on $(0,w-\delta)$. And since we can take $\delta$ arbitrarily small, $g$ is constant on its domain $(0,w)$. Finally, for $\eqref{eq:kappa}$ to be satisfied, this constant must be  $\frac{\kappa}{2n(n-1)}$.
\end{proof}

 Axiom b calibrates the parameters $(\beta^\Theta_{ij})$ using an experiment consisting of a measurement with Normally distributed noise. Different distributions for the noise would lead to different representations for the coefficients. For example, a natural alternative would be an experiment $(\RR, (\xi_i)_{i\in\Theta})$ where each $\xi_i$ is Laplace distribution with variance 1 and mean equal to the state $i$ (the corresponding probability density function is $f(x) = \frac{1}{2}\ee^{\vert x - i \vert}$). The divergence $D(\xi_i\Vert\xi_j)$ between any two such distribution is
 \[
    \ee^{-\vert i - j \vert} + \vert i - j \vert - 1.
 \]
 As in the Normal case, this is a decreasing function of the distance between states. Even if the distribution used in axiom b is different, the proof of Proposition~\ref{prop:onedimens} can be applied with almost no modifications, and leads to a representation with parameters
 \[
    \beta^\Theta_{ij} = \frac{\kappa}{n(n-1)}\,\,\frac{1}{\ee^{-\vert i - j \vert} + \vert i - j \vert - 1}.
 \]

\section{Identification}\label{sec:identification}

 Consider the setup of \S\ref{sec:identifying-the-cost}, and given a pair of choice probabilities $(\mu_1,\mu_2)$ define the quantities
 \[
    \hat{\beta}_{12} = \frac{l_2-l_1+\log \frac{l_1}{l_2}}{\frac{(l_1-l_2)^2}{l_1 l_2}- (\log \frac{l_1}{l_2})^2} \hspace{1cm}\text{and}\hspace{1cm}
    \hat{\beta}_{21} = \frac{v}{2} \frac{ \frac{l_2-l_1}{l_1 l_2} +  \log \frac{l_1}{l_2}}{ \frac{(l_1-l_2)^2}{l_1 l_2}-(\log \frac{l_1}{l_2})^2}
 \]

 \begin{prop}
  The choice probabilities $(\mu_1,\mu_2)$ are the optimal solution with respect to an  LLR cost function if and only if $\hat{\beta}_{12}$ and $\hat{\beta}_{21}$ are non-negative and at least one is positive.
 \end{prop}
 
 \begin{proof}
 As shown by Proposition~\ref{prop:llr-convex}, $C$ is a convex function. 
 We note that the condition \eqref{eq:identified-beta-example} is equivalent to \eqref{eq:foc-example} which equals the first order condition for the optimization problem, which is sufficient because of the concavity of the optimization problem. If at least one of $\hat\beta_{1,2},\hat\beta_{2,1}$ is positive, then the solution of the optimization problem is internal and the first order condition applies. Conversely, if both are zero then the optimization problem has no solution within its domain.
 \end{proof}

\section{The cost of bounded experiments with binary state}

In this section we restrict ourselves to the case of a binary state space $\Theta = \{0,1\}$, and the class of {\em bounded} experiments $\mathcal{B}$: an experiment is said to be bounded if the beliefs that it induces are bounded away from $0$ and $1$. In terms of log-likelihood ratios, it is bounded if there is some $M$ such that $\ell_{01}(s)$ is $\mu_0$- and $\mu_1$-almost surely in $[-M,M]$. The class of bounded experiments is contained in the class $\mathcal{E}$ of experiments considered in the rest of the paper. The bounded experiments contain all the experiments that have a finite set of possible realizations, and in which not state is ever conclusively excluded. 

As we discuss above, a strengthening of Axiom~\ref{axm:info-content} is Blackwell monotonicity: $C$ is said to be Blackwell monotone if $C(\mu) \geq C(\nu)$ whenever
If $\mu$ Blackwell dominates $\nu$.

For the class of bounded experiments, we show that~\ref{axm:additivity} and~\ref{axm:affinity-1} are sufficient for proving that a Blackwell monotone cost is an LLR cost: the continuity axiom~\ref{axm:continuity} is not needed. This proof heavily relies on a recent result of \cite*{mu2020blackwell}, which characterizes the monotone and additive functions on the class of bounded Blackwell experiments with binary state. An extension of this result to large state spaces is currently out of reach, and so we do not have a more general proof. Nevertheless, we conjecture that the continuity axiom is generally redundant.

\begin{thm}\label{thm:repr-bayes}
Let $\Theta=\{0,1\}$. A Blackwell monotone information cost function $C \colon \mathcal{B} \to \mathbb{R}_+$ satisfies Axioms~\ref{axm:additivity} and \ref{axm:affinity-1} if and only if there exist $\beta_{01},\beta_{10} \geq 0$ such that for every experiment $\mu \in \mathcal{B}$,
\begin{equation*}
C(\mu)=\beta_{01}\dkl(\mu_0 \Vert \mu_1) + \beta_{10}\dkl(\mu_1\Vert\mu_0).
\end{equation*}
\end{thm}

Before proving Theorem~\ref{thm:repr-bayes}, we will introduce some definitions and results from \cite{mu2020blackwell}.

For $t \in (0,\infty]$, we denote by $R_t(\mu_0 \Vert \mu_1)$ the R\'enyi $t$-divergence between two probability $\mu_0,\mu_1$ defined on the same measurable space $S$. For $t \neq 1$, $t \neq \infty$,
\begin{equation*}
    R_t(\mu_0 \Vert \nu_1) = \frac{1}{t-1}\log\int_S \left(\frac{\dd \mu_0}{\dd \mu_1}(s)\right)^{t-1}\,\dd \mu_0(s).
\end{equation*}
For $t=1$
\begin{equation*}
    R_1(\mu_0 \Vert \mu_1) = \int_S \log\frac{\dd \mu_0 \hfill}{\dd \mu_1}(s) \,\dd \mu_0 (s) = \dkl(\mu_0 \Vert \mu_1).
\end{equation*}
For $t=\infty$, $R_\infty(\mu_0 \Vert \mu_1)$ is the essential maximum of the log-likelihood ratio $\log\frac{\dd \mu_0 \hfill}{\dd \mu_1}$. Note that $R_t(\mu_0\Vert\mu_1)$ is always non-negative, and positive whenever $\mu_0 \neq \mu_1$. Note also that if $\log\frac{\dd \mu_0 \hfill}{\dd \mu_1}$ is almost surely in $[-M,M]$ (as is always the case for bounded experiments, for some $M$) then $R_t \leq M$. 

The following result is a reformulation of Theorem 2 in \cite{mu2020blackwell} (see also Lemmas 5 and 6).\footnote{The {\em data processing inequality} in that paper is monotonicity with respect to deterministic garblings, which is implied by Blackwell monotonicity. The additivity there translates immediately to additivity in the sense of Axiom~\ref{axm:additivity}.}

\begin{thm}[\citealt{mu2020blackwell}]
\label{thm:mu}
An information cost function $C \colon \mathcal{B} \to \mathbb{R}_+$ satisfies Axioms~\ref{axm:info-content} and \ref{axm:additivity} if and only if there exist two finite Borel measures $m_0,m_1$ on $[1/2,\infty]$ such that for every bounded experiment $\mu = (S,\mu_0,\mu_1)$ it holds that
\begin{align*}
   C(\mu) = \int_{[1/2,\infty]}R_t(\mu_0 \Vert \mu_1)\,\dd m_0(t) + \int_{[1/2,\infty]}R_t(\mu_1 \Vert \mu_0)\,\dd m_1(t).
\end{align*}
\end{thm}

Using this result, we can now prove Theorem~\ref{thm:repr-bayes}.

\begin{proof}[Proof of Theorem~\ref{thm:repr-bayes}]
The argument that this representation satisfies the axioms is identical to the same argument in the proof of Theorem~\ref{thm:repr-1}. It thus remains to be shown that the representation is implied by the axioms.

By Theorem~\ref{thm:mu}, 
\begin{align}
   C(\mu) &= \beta_{01}\dkl(\mu_0\Vert\mu_1)+\beta_{10}\dkl(\mu_1\Vert\mu_0)\nonumber\\
   &\quad+\int_{[1/2,1)}R_t(\mu_0 \Vert \mu_1)\,\dd m_0(t) + \int_{[1/2,1)}R_t(\mu_1 \Vert \mu_0)\,\dd m_1(t) \nonumber\\
   &\quad+\int_{(1,\infty]}R_t(\mu_0 \Vert \mu_1)\,\dd m_0(t) + \int_{(1,\infty]}R_t(\mu_1 \Vert \mu_0)\,\dd m_1(t).\label{eq:mu1infty}
\end{align}
for some $\beta_{01},\beta_{10} \geq 0$ and  $m_0$, $m_1$ finite Borel measures on $[1/2,\infty]$ that assign measure $0$ to the singleton $\{1\}$. To prove the claim, we show that $m_0$ and $m_1$ are the zero measures.

Let $\mu=(S,\mu_0,\mu_1)$ be a non-trivial bounded experiment, and let $\nu = (1/r) \cdot \mu^{\otimes r}$ for some $r$. It follows from the definition of R\'enyi $t$-divergences that for $t \neq 1$, $t \neq \infty$
\begin{align*}
    R_t(\nu_0 \Vert \nu_1)
    &=\frac{1}{t-1}\log\left(\frac{r-1}{r}+\frac{1}{r}\left(\int_S \left(\frac{\dd \mu_0}{\dd \mu_1}(s)\right)^{t-1}\,\dd \mu_0(s)\right)^r\right).
\end{align*}
Now, for $x > 1$,
\begin{align*}
    \lim_{r \to \infty}\log \left(\frac{r-1}{r}+\frac{1}{r}x^r\right) = \infty,
\end{align*}
and for $x < 1$ this same limit is $0$. It thus follows that for $t > 1$ (including, trivially, $t=\infty$)
\begin{align}
\label{eq:R_t_infty}
    \lim_{r \to \infty}R_t(\nu_0 \Vert \nu_1) = \infty,
\end{align}
since $R_t$ is positive for non-trivial experiments, and so the integral in the expression for $R_t$ is strictly greater than 1. For $t < 0$
\begin{align}
\label{eq:R_t_0}
    \lim_{r \to \infty}R_t(\nu_0 \Vert \nu_1) = 0,
\end{align}
since, again by the positivity of $R_t$, the integral in the expression for $R_t$ is strictly less than 1.

It follows from~\eqref{eq:R_t_infty} that both $m_0$ and $m_1$ must assign no mass to $(1,\infty]$, i.e.\  $m_0((1,\infty])=m_1((1,\infty])=0$, since otherwise the integral $\int_{(1,\infty]}R_t(\mu_0 \Vert \mu_1)\,\dd m_0(t)$ or $\int_{(1,\infty]}R_t(\mu_0 \Vert \mu_1)\,\dd m_1(t)$ would diverge and by \eqref{eq:mu1infty} the cost of the experiment $(1/r)\cdot\mu^{\otimes r}$ would diverge 
\begin{align*}
    \lim_{r \to \infty} C((1/r)\cdot\mu^{\otimes r}) = \infty\,.
\end{align*}
This would contradict the axioms which imply that $C((1/r)\cdot\mu^{\otimes r}) = C(\mu)$. It then follows from \eqref{eq:R_t_0} that $m_0((1/2,1))=m_1((1/2,1)) =0$, since otherwise
\begin{align*}
    \lim_{r \to \infty} C((1/r)\cdot\mu^{\otimes r}) < C(\mu). ~~~\qedhere
\end{align*}
\end{proof}

\section{Uniform Separable Bayesian LLR Cost}
\label{app:uniform}


\begin{proof}[Proof of Proposition~\ref{prop:uniform-posterior}]
 It is straightforward to verify that if the parameters satisfy $\beta_{ij}(q) = \gamma_{ij}q_i$, then $C$ is uniformly posterior separable. We now prove the opposite implication.
 
 Fix a prior $q$ with full support, and consider an experiment $\mu$ where the set of signal realizations is a product $S_1 \times S_2$, with $S_1$ a finite set, and each $\mu_i$ satisfies $\mu_i(\{s\} \times S_2) >0$ for every $s \in S_1$. We denote by $\mu_i^1$ the marginal of $\mu_i$ on $S_1$, and by $\mu_i(\cdot \vert s)$ the measure on $S_2$ obtained by conditioning $\mu_i$ on $s \in S_1$.
 
 The chain rule for the KL-divergence implies that the cost of such an experiment can be written as
 \begin{equation}\label{eq:ups1}
    C(\mu,q) = \sum_{ij} \beta_{ij}(q) \left[ \dkl(\mu^1_i\Vert\mu^1_j) + \sum_{s_1 \in S_1} \mu^1_i(s_1) \dkl(\mu_i(\cdot\vert s_1) \Vert \mu_j(\cdot \vert s_1) ) \right].
 \end{equation}
 
 Now assume $C$ is uniformly posterior separable with respect to a function $G$. The cost of the experiment $\mu$ can then be written as follows. It will be convenient to denote posterior beliefs as random variables defined over the probability space $(\Theta \times S_1 \times S_2, \PP)$ where $\PP$ is obtained from $q$ and $\mu$ in the obvious way. Let $p^2$ be the posterior belief over $\Theta$ obtained by conditioning $q$ on a realization $(s_1,s_2)$, and let $p^1$ be the posterior belief obtained by conditioning $q$ on a realization $s_1$. Then
 \begin{align*}
     C(\mu,q) &= \mathbb{E}\left[ G(p^2) - G(p^1) + G(p^1) - G(q)\right] \\
              &= \mathbb{E}\left[G(p^1) - G(q)\right] + \sum_{s_1 \in S_1} \PP(s_1) \mathbb{E}\left[G(p^2) - G(p^1) \vert p^1 = q(\cdot\vert s_1)\right].
 \end{align*}
 Now consider the experiment $((\mu^1_i),S)$ which consists of observing the first realization $s_1$ but not the second. By uniform posterior separability, its cost, at the prior $q$, is given by
 \[
    \mathbb{E}\left[G(p^1) - G(q)\right] = \sum_{ij} \beta_{ij}(q) \dkl(\mu^1_i\Vert\mu^1_j).
 \]
 Given a realization $s_1 \in S$, consider the experiment $((\mu_i(\cdot\vert s_1)),S_2)$. By considering now $p^1 = q(\cdot \vert s_1)$ as a prior, uniform separability implies that the cost of the experiment $((\mu_i(\cdot\vert s_1)),S_2)$ is equal to
 \[
    \mathbb{E}\left[G(p^2) - G(p^1) \vert p^1 = q(\cdot\vert s^1)\right] = \sum_{ij} \beta_{ij}(q(\cdot \vert s_1))\dkl(\mu_i(\cdot \vert s_1) \Vert \mu_j(\cdot \vert s_1)).
 \]
 The last two equations imply that the cost $C(\mu,q)$ can be rewritten as
 \begin{equation}\label{eq:ups2}
    \sum_{ij} \beta_{ij}(q) \dkl(\mu^1_i\Vert\mu^1_j) + \sum_{s_1 \in S_1}\PP(s_1) \left(\sum_{ij} \beta_{ij}(q(\cdot \vert s_1))\dkl(\mu_i(\cdot \vert s_1) \Vert \mu_j(\cdot \vert s_1))\right).
 \end{equation}
 This equation can be interpreted as saying that the cost of running the experiment $\mu$ is equal to the cost of running the first experiment $((\mu^1_i),S_1)$ plus the expected cost of running the second experiment $((\mu_i(\cdot\vert s_1)),S_2)$, conditional on the signal realization $s_1$ from the first experiment. 
 By equating \eqref{eq:ups1} and \eqref{eq:ups2} we obtain that
 \begin{equation}\label{eq:ups3}
    \sum_{s_1 \in S_1}  \sum_{ij} \left[ \beta_{ij}(q) \mu_i^1(s_1) - \PP(s_1)\beta_{ij}(q(\cdot\vert s_1)) \right] \dkl(\mu_i(\cdot\vert s_1) \Vert \mu_j(\cdot \vert s_1))  = 0.
 \end{equation}
 Given a particular realization $s_1 \in S_1$, we are free to choose $\mu$ such that all the conditional experiments $((\mu_i(\cdot\vert s'_1)),S_2)$, $s'_1 \neq s_1$, are completely uninformative, and hence have cost $0$. Thus, it must hold that for every $s_1 \in S_1$, 
 \[
    \sum_{ij} \left[ \beta_{ij}(q) \mu_i^1(s_1) - \PP(s_1)\beta_{ij}(q(\cdot\vert s_1)) \right] \dkl(\mu_i(\cdot\vert s_1) \Vert \mu_j(\cdot \vert s_1))  = 0.
 \]
 By Lemma~\ref{lem:domain}, the latter can hold only if $$\beta_{ij}(q) \mu_i^1(s_1) = \PP(s_1)\beta_{ij}(q(\cdot\vert s_1)).$$ By dividing and multiplying the left-hand side by $q_i$ and then applying Bayes' rule we obtain that
 \[
    \frac{\beta_{ij}(q)}{q_i} = \frac{\beta_{ij}(q(\cdot\vert s_1))}{q(\cdot \vert s_1)}.
 \]
 Given any $q' \in \P(\Theta)$ with full support, we can choose $\mu$ such that $q(\cdot \vert s_1) = q'$ for some $s_1$. The conclusion now follows by defining $\gamma_{ij} = \beta_{ij}(q)/q_i$.
\end{proof}

\paragraph{Prior Dependence of Bayesian LLR Cost.}
As we prove in Proposition~\ref{prop:uniform-posterior}, the only uniformly posterior separable LLR cost potentially assigns different cost to the same experiment at different prior beliefs.
%
We next explore which experiments have prior dependent cost, through a simple example of binary experiments. Consider the standard setting of a binary state space $\Theta = \{1,2\}$, and an experiment $\mu$ with a binary signal which equals the state with some probability $1/2 < r < 1$. For concreteness, imagine a coin whose probability of heads depends on the state and is either $r$ or $1-r$, and the experiment $\mu$ consists of tossing the coin. Consider a Bayesian LLR cost, with $b_{12}=b_{21}=b$. In this case, even though the effective $(\beta_{ij})$'s depend on the prior, a simple calculation shows that the cost of the experiment does not, and equals
 $$
  C(\mu,q) = b (2 r-1) \log \frac{r}{1-r}
 $$
for every prior $q$.\footnote{This contrasts with mutual information, where the prior affects the cost of this experiment: the cost is highest for the uniform prior, and vanishes as the prior tends towards certainty.}

Consider now the experiment $\nu$ in which the coin is tossed until a ``heads'' outcome. Under Bayesian LLR costs, the cost can be calculated to be
\begin{align*}
C(\nu,q) = \left(\frac{q_1}{r}+\frac{q_2}{1-r}\right)C(\mu,q).
\end{align*}
This cost does depend on the prior: as the above display shows, it is equal to the cost of one toss of the coin, times the expected number of times that it is to be tossed. The latter quantity depends on the prior, in the obvious way. This cost is thus consistent with our additivity axiom, in the sense that this one-shot experiment $\nu$---which is equivalent to a dynamic experiment in which $\mu$ is carried out a random number of times---has a cost that equals the expected number of repetition of $\mu$, times the cost of each independent realization of $\mu$. 


We generalize the example of a biased coin toss to any experiment $\mu$ for which $\dkl(\mu_1 \Vert\mu_2) = \dkl(\mu_2 \Vert\mu_1)$. As the next proposition shows, this condition exactly captures prior independence of Bayesian LLR costs, in the symmetric case in which $b_{12}=b_{21}$.
\begin{prop}\label{prop:prior-independence}
Let $\Theta=\{1,2\}$. Let $C$ be a uniformly posterior separable Bayesian LLR cost specified by $b_{12}=b_{21}=b > 0$. Let $\mu$ be a Blackwell experiment. Then the following are equivalent.
\begin{enumerate}[(i)]
    \item $\dkl(\mu_1 \Vert\mu_2) = \dkl(\mu_2 \Vert\mu_1)$.
    \item $C(\mu,q)$ is independent of the prior $q$.
\end{enumerate}
\end{prop}

\begin{proof}
Under the assumption that $b_{12}=b_{21}=b > 0$, the cost of an experiment $\mu$ at prior $q$ is
\[
    C(\mu,q) = b \left[q_1 \dkl(\mu_1\Vert\mu_2) +  q_2 \dkl(\mu_2\Vert\mu_1)\right].
\]
Clearly, this quantity depends on $q$ if and only if $\dkl(\mu_1 \Vert\mu_2) \neq \dkl(\mu_2 \Vert\mu_1)$.
\end{proof}

\end{appendices}
\end{document}